\documentclass[10pt, a4paper,reqno]{amsart}

\usepackage{color} \definecolor{bleu_sombre}{rgb}{0,0,0.6}  \definecolor{rouge_sombre}{rgb}{0.8,0,0}\definecolor{vert_sombre}{rgb}{0,0.6,0}
\usepackage[plainpages=false,colorlinks,linkcolor=bleu_sombre,citecolor=rouge_sombre,urlcolor=vert_sombre,breaklinks]{hyperref}

\usepackage[T1]{fontenc}	
\usepackage[latin1]{inputenc}	
\usepackage[english]{babel}

\usepackage{amsmath,amssymb,amsthm,amsfonts}
\usepackage{url,color}
\usepackage{nicefrac,enumerate,stmaryrd,dsfont}

 \newenvironment{bsmallmatrix}{\left[\begin{smallmatrix}}{\end{smallmatrix}\right]}

\theoremstyle{plain}
\newtheorem{theorem}{{Theorem}}[section] 
\newtheorem*{theorem*}{{Theorem}}
\newtheorem{proposition}[theorem]{Proposition}
\newtheorem*{proposition*}{Proposition}
\newtheorem{corollary}[theorem]{Corollary}
\newtheorem*{corollary*}{Corollary}
\newtheorem{lemma}[theorem]{Lemma}
\newtheorem*{lemma*}{Lemma}

\theoremstyle{definition}
\newtheorem{definition}[theorem]{Definition}
\newtheorem*{definition*}{Definition}

\theoremstyle{remark}
\newtheorem*{remarque*}{Remarque}

\newtheorem{remark}[theorem]{Remark}

\newtheorem*{example*}{Example}

\newtheorem*{examples*}{Examples}

\newcommand{\commm}[1]{{}}

\makeatletter

\@addtoreset{equation}{section}  
\makeatother

\renewcommand{\leq}{\leqslant}	\renewcommand{\geq}{\geqslant}
\renewcommand{\bar}[1]{\overline{#1}}
\renewcommand\over[2]{{\,\buildrel #1\over#2\,}}

\newcommand{\ie}{{\it{i.e. }}}

\newcommand{\inv}{^{-1}}

\newcommand {\limt}[2]{\xrightarrow[#1 \to #2]{}}

\newcommand{\abs}[1]{\left\vert #1\right\vert}        
\newcommand{\nr}[1]{\left\Vert #1\right\Vert}         
\newcommand{\innp}[2]{\left< #1 , #2 \right>}         

\newcommand{\Dom}{\Dc}			
\newcommand{\Four}{\mathcal{F}}		

\newcommand{\Opw}{{\mathop{\rm{Op}}}_h^w}

\newcommand{\pppg}[1] {\left< #1 \right>} 	
\newcommand{\symbor}{C_b^\infty}		
\newcommand{\symb} {S}		

\newcommand{\bigo}[2]{\mathop{O}\limits_{#1 \to #2}}

\newcommand{\singl}[1]{\left\{ #1 \right\}}		
\newcommand{\Ii}[2]{\llbracket #1,#2 \rrbracket}	
\newcommand{\R}{\mathbb{R}}		\newcommand{\C}{\mathbb{C}}
\newcommand{\N}{\mathbb{N}}

\newcommand{\1}[1]{\mathds 1 _{#1}}

\newcommand{\st}{\,:\,}					

\newcommand{\seq}[2]{\left({#1}_{#2}\right)_{#2 \in\N}} 

\newcommand{\restr}[2]{\left.#1\right|_{#2}}         
\renewcommand{\Re}{\mathop{\rm{Re}}\nolimits}        
\renewcommand{\Im}{\mathop{\rm{Im}}\nolimits}        
	
\DeclareMathOperator{\Id}{Id}                        
 
\DeclareMathOperator{\supp}{supp}                    

\renewcommand{\a}{\alpha}\renewcommand{\b}{\beta}\newcommand{\g}{\gamma}\renewcommand{\d}{\delta}\newcommand{\D}{\Delta}\newcommand{\e}{\varepsilon}\newcommand{\z}{\zeta} \renewcommand{\th}{\theta}\newcommand{\Th}{\Theta}\renewcommand{\l}{\lambda}\renewcommand{\L}{\Lambda}\newcommand{\m}{\mu}\newcommand{\n}{\nu}\newcommand{\x}{\xi}\newcommand{\s}{\sigma}\renewcommand{\t}{\tau}\newcommand{\f}{\varphi}\newcommand{\vf}{\phi}\newcommand{\h}{\chi}\newcommand{\p}{\psi}\renewcommand{\O}{\Omega}

\newcommand{\Ac}{{\mathcal A}}\newcommand{\Cc}{{\mathcal C}}\newcommand{\Dc}{{\mathcal D}}\newcommand{\Hc}{{\mathcal H}}\newcommand{\Kc}{{\mathcal K}}\newcommand{\Lc}{{\mathcal L}}\newcommand{\Nc}{{\mathcal N}}\newcommand{\Rc}{{\mathcal R}}\newcommand{\Sc}{{\mathcal S}}\newcommand{\Tc}{{\mathcal T}}\newcommand{\Uc}{{\mathcal U}}\newcommand{\Vc}{{\mathcal V}}\newcommand{\Wc}{{\mathcal W}}

\newcommand{\ad}{{\rm{ad}}}

\newcommand{\divg}{\mathop{\rm{div}}\nolimits}

\newcounter{stepproof}
\newcommand{\stepp}{\stepcounter{stepproof} \noindent {\bf $\bullet$}\quad }

\newcommand{\newl} {{\l}}

\newcommand{\Poo}{P_0}  \newcommand{\Pol}{{P_z}}   
\newcommand{\Hz}{H_z}          
\newcommand{\Hl}{H_\l}      \newcommand{\Hul}{H_\l^1}   
\newcommand{\Ho}{H_0}

\newcommand{\hhs}{H_h^\s}
\newcommand{\rehz} {{\Re(\Hz)}}

\newcommand{\RP}{R_\iota}
\newcommand{\RPz}{\RP(z)}  
\newcommand{\tRtau}{\tilde R_\iota(z)}
\newcommand{\tRz}{\tilde R(z)}

\newcommand{\tThiota}{\tilde \Th}

\newcommand{\opinsert}{\Cc}

\newcommand{\Ptau}{{ \tilde P_{z}}} \newcommand{\Ptauo}{ \tilde P_{z}^0} \newcommand{\Ptaur}{ \Re ( \tilde P_{z})}

\begin{document}

\title {Local Energy Decay for the Damped Wave Equation}

\author{Jean-Marc Bouclet and Julien Royer}
\curraddr{Institut de Math\'ematiques de Toulouse \\
118 route de Narbonne, 31062 Toulouse C\'edex 9}
\email[J.-M. Bouclet] {bouclet@math.univ-toulouse.fr}
\email[J. Royer] {julien.royer@math.univ-toulouse.fr}


\begin{abstract}
We prove local energy decay for the damped wave equation on $\R^d$. The problem which we consider is given by a long range metric perturbation of the Euclidean Laplacian with a short range absorption index. Under a geometric control assumption on the dissipation we obtain an almost optimal polynomial decay for the energy in suitable weighted spaces. The proof relies on uniform estimates for the corresponding ``resolvent'', both for low and high frequencies. These estimates are given by an improved dissipative version of Mourre's commutators method.
\end{abstract}

\maketitle

\tableofcontents

\section{Introduction} \label{sec-intro}

We consider on $\R^d$, $d\geq 3$, the damped wave equation:
\begin{equation} \label{wave-lap}
 \begin{cases}
\partial_t^2 u(t,x) + \Ho u(t,x) + a(x) \partial_t u(t,x)= 0  & \text{for  } (t,x) \in  \R_+ \times \R^d, \\
u(0,x) = u_0(x), \quad \partial_t u (0,x) = u_1(x) &  \text{for } x \in \R^d.
 \end{cases}
\end{equation}
Here $\Ho$ is an operator in divergence form
\[
\Ho = - \divg (G(x) \nabla),
\]
where $G(x)$ is a positive symmetric matrix with smooth entries, which is a long range perturbation of the identity (see \eqref{dec-metric-a}).
Laplace-Beltrami operators will be considered as well, but the case of operators in divergence form captures all the difficulties.  The operator $\Ho$ is self-adjoint and non-negative on $L^2(\R^d)$ with domain $H^2(\R^d)$. 
The function $a \in C^\infty(\R^d)$ is the absorption index. It takes non-negative values and is a short range potential.
%
%
More precisely we assume that there exists $\rho > 0$ such that for $j,k\in\Ii 1 d$, $\a \in \N^d$  and $x \in \R^d$ we have
\begin{equation} \label{dec-metric-a}
\abs{\partial^\a (G_{j,k}(x) - \d_{j,k})} \leq c_\a \pppg x ^{-\rho - \abs \a} \quad \text{and} \quad \abs{\partial^\a a (x)} \leq c_{\a} \pppg x ^{-1 -\rho - \abs \a},
\end{equation}
where $\pppg x = \big( 1 + \abs x^2 \big)^{\frac 12}$, $\d_{j,k}$ is the Kronecker delta and $\N$ is the set of non negative integers.\\ 

%
Let $\Hc$ be the Hilbert completion of ${\mathcal S}(\R^d) \times {\mathcal S}(\R^d)$ for the norm
\begin{equation}\label{normemixte}
 \nr {(u,v)}_\Hc^2 = \nr {\Ho^{1/2} u}_{L^2}^2   + \nr {v}_{L^2}^2. 
\end{equation}
Here we use the square root $ \Ho^{1/2} $ of the self-adjoint operator $ H_0 $ but the corresponding term in the above energy can also be written  $ \innp{G(x)\nabla u}{ \nabla u}_{L^2}$.
Then $\Hc = \dot H^1 \times L^2$, $\dot H^1$ being the usual homogeneous Sobolev space on $\R^d$. We consider on $\Hc$ the operator
\begin{equation} \label{def-Ac}
 \Ac = \begin{pmatrix} 0 & I \\ \Ho & -i a \end{pmatrix}
\end{equation}
with domain 
\begin{equation} \label{dom-A}
\Dom(\Ac) = \singl{(u,v) \in \Hc \st (v , \Ho u) \in \Hc},
\end{equation}
where  $ \Ho u $ is taken in the temperate distributions sense. We refer to Section \ref{sec-diss} for more details on $ {\mathcal H}, {\mathcal D}({\mathcal A}) $ and $ {\mathcal A} $ which we can omit in this introduction. 
If we next let $(u_0,u_1) \in \Dom(\Ac)$, then $u$ is a solution to the problem \eqref{wave-lap} if and only if $\Uc = (u,i \partial_t u)$ is a solution to
\begin{equation} \label{wave-A}
\begin{cases}
 (\partial_t  + i \Ac  ) \Uc(t) = 0,\\
\Uc(0) = \Uc_0,
\end{cases}
\end{equation}
where $\Uc_0 = (u_0,i u_1)$.\\

We are going to prove that the operator $\Ac$ is maximal dissipative on $\Hc$ (in the sense of Definition \ref{def-diss}, see Proposition \ref{prop-A-diss}). According to the Hille-Yosida Theorem, this implies in particular that it generates on $\Hc$ a contractions semigroup $t \mapsto e^{-it {\mathcal A}}$, $t\geq 0$. Therefore the problem \eqref{wave-A} has a unique solution $\Uc \in C^0 (\R_+, \Dom(\Ac)) \cap C^1 (\R_+, \Hc)$ for any $\Uc_0 = (u_0,iu_1) \in \Dom(\Ac)$. The first component of $\Uc$ is the solution to \eqref{wave-lap}, while the second is its time derivative.
Moreover the energy function
\[
  t \mapsto \nr {\Uc(t)}_\Hc^2 = \nr{\Ho^{1/2} \, u(t)}_{L^2}^2 + \nr{\partial_t u(t)}_{L^2}^2
\]
is non-increasing, the decay being due to the absorption index $a$:
\begin{equation} \label{global-decay}
 \forall t \geq 0, \quad \frac {d}{dt}\nr {\Uc(t)}_\Hc^2 = - 2 \int_{\R^d} a(x) \abs {\partial_t u(t,x)}^2 \, dx \leq 0.
\end{equation}

An important question about the long time behavior of the solution to a wave equation is the local energy decay. This has been widely studied in the self-adjoint case (\ie without the damping term $a\partial_t u$ in \eqref{wave-lap}). Let us mention \cite{laxmp63}, where the free wave equation outside a star-shaped obstacle (with Dirichlet boundary conditions) is considered. An exponential decay for the local energy is obtained in odd dimensions, using a polynomial decay from \cite{morawetz61} and the theory of Lax-Phillips \cite{lax-phillips}. This has been generalized to non-trapping obstacles in \cite{morawetzrs77} and \cite{melrose79} (using the results about propagation of singularities given in \cite{melroses78}). Note that in all these papers the obstacle has to be bounded.

J. Ralston (\cite{ralston69}) proved that the non-trapping assumption is necessary to obtain uniform local energy decay, as was conjectured in \cite{lax-phillips}. However N. Burq (\cite{burq98}) has proved logarithmic decay with loss of regularity without non-trapping assumption, by proving that there are no resonances in a region close to the real axis. As in the previous works, the obstacle is bounded and initial conditions have to be compactly supported.

In contexts close to ours, results about local energy decay for long range perturbations of various evolution equations have been obtained in \cite{bonyh12} and \cite{bouclet11}.
 Both of these papers prove polynomial decay by mean of Mourre estimates.
 To our knowledge, the best estimates known so far on local energy decay for the wave equation with long range perturbations have been obtained in \cite{tataru13} in three dimensions and in \cite{guillarmouhs} in odd dimension $d$. Both obtain a decay of order $t^{-d}$. 

All these papers deal with the self-adjoint case. The local energy decay for the dissipative wave equation on an exterior domain has been studied by L. Aloui and M. Khenissi in \cite{alouik02}. They obtain exponential decay in odd dimension using the theory of Lax-Phillips and the contradiction argument with semiclassical measures of G. Lebeau \cite{lebeau96}. In this setting, the non-trapping assumption is replaced by a condition of exterior geometric control (see {\it e.g.} \cite{raucht74,bardoslr92} for more on this condition): every (generalized) geodesic has to leave any fixed bounded region or meet the damping region in (uniform) finite time. In this work the problem is a compact perturbation of the free case. \\

In the present paper we prove a polynomial time decay of the local energy for some asymptotically vanishing and dissipative perturbations of the free wave equation. As is well known the main difficulties are in the low frequency and high frequency regimes. For the latter we use a damping assumption on the classical flow which is similar to the condition of exterior geometric control. This will be explicited in \eqref{hyp-amort}. Here we simply record that when $a = 0$ this assumption is the usual non-trapping condition, and when $a$ is positive everywhere it is automatically satisfied.\\

Let us denote by $L^{2,\d}$ the weighted space $L^2 \big(\pppg x ^{2\d} \, dx\big)$, while $H^{k,\d}$ for $k\in \N$ is the corresponding weighted Sobolev space, with norm 
\[
\nr u_{H^{k,\d}} =  \nr { \pppg {H_0}^{\frac k 2} u}_{L^{2,\d}}.
\]
This norm is equivalent to the more standard one (in term of $\nr{  \langle x \rangle^{\delta} \partial^{\alpha} u}_{L^2} $) since $\pppg {H_0}^{\frac k 2}$ is an elliptic pseudo-differential operator of order $k$.\\

The purpose of this paper is to prove the following theorem:

\begin{theorem}  [Local energy decay]\label{th-loc-decay}
Assume that every bounded trajectory of the classical flow goes through the damping region (see \eqref{hyp-amort}). Let $\d > d + \frac 12$ and $\e > 0$. Then there exists $C \geq 0$ such that for $(u_0,u_1) \in H^{2,\d} \times H^{1,\d}$ and $t \geq 0$ we have
\[
\nr{\nabla  u(t)}_{L^{2,-\d}} + \nr{\partial_t u(t)}_{L^{2,-\d}} \leq C \pppg t ^{-(d-\e)} \left( \nr {u_0} _{H^{2,\d}} + \nr {u_1}_ {H^{1,\d}}\right),
\]
where $u$ is the solution to the damped wave equation \eqref{wave-lap}.
\end{theorem}
In Theorem \ref{low-freq-Laplacien}, we obtain the same result when $ H_0 $ is replaced by the Laplace-Beltrami operator $ - \Delta_g $ associated to a metric $g$ which is a long range perturbation of the flat one. As we will see in Section \ref{sec-laplacien}, this result is essentially reduced to the one on $ H_0 $ (more precisely on the related resolvent estimates) \emph{via} a fairly simple perturbation argument.\\

As in \cite{bonyh12}, we  obtain a $t^{\e-d}$ decay for all $ \e > 0$. Let us note that the $ t^{-d} $ decay shown in \cite{tataru13} and \cite{guillarmouhs} is obtained under special long range assumptions. 
More precisely, in \cite{guillarmouhs} the authors obtain asymptotic expansions for the resolvent (hence for the spectral measure) which allow to compute asymptotics of the wave kernel and infer the corresponding decay. This is a strong result but relies on strong assumptions on the metric which has to be of scattering type. In \cite{tataru13}, the author is also able to give an asymptotic expansion of its resolvent, assuming in particular that the metric and its perturbation are radial at infinity (modulo short range terms).


An important new feature here is that we allow non self-adjoint perturbations. We discuss the related problems below (once we have introduced the relevant resolvent). In the self-adjoint case, we recover the time decay proved in \cite{bonyh12} with similar long-range assumptions. However our analysis also provides resolvent estimates which are new both in the self-adjoint and non self-adjoint cases.

Let us note that, in odd dimension, if $a$ and $G-I$ decay fast enough, one may expect a time decay proportional to their spatial decay rate, as is proved by Bony-H\"afner \cite{bonyh13} when $a=0$. For exponentially decaying perturbations, one may also expect an exponential decay by using a suitable theory of resonances. On the other hand, in even dimension, the behaviour of the free wave equation suggests one cannot expect such improvements.

We do not know whether the short-range assumption on the absorption index is sharp or not. However previous results obtained in \cite{shibata83, dans95, ikehataty13}, where the absorption index is constant or at least bounded from below by $c_0 \pppg  x \inv$, $c_0 > 0$, provide estimates of order $t^{-\frac d 2}$. This is related to the ``overdamping'' phenomenon: when the absorption is too strong, the equation tends to behave as a heat equation at low frequencies.\\

%

We are going to prove Theorem \ref{th-loc-decay} by a spectral approach.   After a Fourier transform, the solution $u(t)$ can be written as an integral over frequencies $\t = \Re z$ of
\begin{equation}  \label{def-Rz} 
R(z) = \big(\Hz -z^2 \big) \inv,  
\quad \text{where} \quad  \Hz  = \Ho - i z a(x).
\end{equation}
Everywhere in the sequel, $R(z)$ will be called resolvent although it is not a resolvent in the usual sense since the operator $\Hz$ depends on the spectral parameter.
We will see in Proposition \ref{prop-R-diss} that it is well-defined for every $z \in \C_+$, where 
\[
\C_+ = \singl {z \in \C \st \Im z > 0}.
\]
The problem is thus reduced to proving uniform estimates for the resolvent $R(z)$ and its derivatives when $\Im z \searrow 0$. Once such a limiting absorption principle is proved, we have to control the dependence on $\Re z$. The difficulties arise when $| \Re z |$ goes to 0 (low frequencies) and $+ \infty$ (high frequencies).
\\

The main new difficulty, compared to the situations considered in \cite{bouclet11,bonyh12}, is the non self-adjoint character of the operator which leads to several new problems. A deep one is due to the absence of suitable functional calculus for $ H_0 - i z a $ allowing sharp spectral localizations. Another problem follows from the fact that derivatives of the resolvent cannot be expressed by pure powers thereof.
Here, there may be some factors $a(x)$ inserted between the resolvents $R(z)$; for instance, we have
\begin{equation} \label{expr-der-R}
R'(z) = i R(z)a(x) R(z) + 2z R(z)^2
\end{equation}
(see Proposition \ref{prop-der-R2} for the general case). Due to these inserted factors (and of course to the non self-adjointness of the operator), the estimates on $ R^{\prime}(z) $ (and higher order derivatives) require  a significant review of multiple commutators estimates in this framework (see Section \ref{sec-Mourre}).
\\



 Theorem \ref{th-loc-decay} follows essentially from a series of estimates on the resolvent $ R (z) $ and its derivatives, which are the core of this paper. Estimates at high frequency and intermediate frequency are unavoidably technical but are quite well understood in this dissipative situation (see \cite{art-mourre}). The biggest issue will be about low frequency. \\

Let us begin with the statement about intermediate frequencies:



\begin{theorem} [Resolvent estimates for intermediate frequencies] \label{th-inter-freq}
 Let $K$ be a compact subset of $\C \setminus \singl 0$. Let $n \in \N$ and $\d > n + \frac 12$. Then there exists $C \geq 0$ such that for all $z \in \C_{+} \cap K$ we have
\[
 \nr{\pppg x ^{-\d} R^{(n)}(z) \pppg x ^{-\d} }_{\Lc(L^2)} + \nr{\pppg x ^{-\d} \nabla R^{(n)}(z) \pppg x ^{-\d} }_{\Lc(L^2)}  \leq C.
\]
\end{theorem}

Here we have denoted by $\Lc(L^2)$ the space of bounded operators on $L^2$, while $R^{(n)}(z)$ is as usual the $n$-{th} derivative of $R(z)$ with respect to $z$. 
\\

As is well known, even for the resolvent of the free Laplacian, such estimates cannot hold uniformly when $z$ goes to 0 if $n$ is too large. This explains the restriction on the rate of decay in Theorem \ref{th-loc-decay}. The following statement is the main technical result of this paper:

\begin{theorem}[Resolvent estimates for low frequencies] \label{th-low-freq}
Let $ \e > 0$ and $n \in \N$. Let $\d$ be greater than $n + \frac 12$ if $n \geq \frac d 2$ and greater than $n + 1$ otherwise. Then there exist a neighborhood $\Uc$ of 0 in $\C$ and $C\geq 0$ such that for all $z \in \Uc \cap \C_{+}$ we have 
\[
\nr{\pppg x ^{-\d}  R^{(n)}(z) \pppg x ^{-\d} }_{\Lc(L^2)}  \leq C \left( 1 + \abs z ^{d-2- n- \e} \right)
\]
and
\[
\nr{\pppg x ^{-\d} \nabla   R^{(n)}(z) \pppg x ^{-\d} }_{\Lc(L^2)}  \leq C \left( 1 + \abs z ^{d-1- n- \e} \right).
\]
\end{theorem}

Unless $d$ is even and $n = d-2$ we can in fact remove $\e$ in the first estimate:

\begin{theorem}[Sharp resolvent estimates for low frequencies] \label{th-low-freq-bis}
Let $n \in \N$. Let $\d_1,\d_2 > n + \frac 12$ be such that $\d_1 + \d_2 > \min(2(n+1),d)$. If $d$ is odd or $n \neq d-2$ then there exist a neighborhood $\Uc$ of 0 in $\C$ and $C\geq 0$ such that for all $z \in \Uc \cap \C_+$ we have
\[
\nr{\pppg x ^{-\d_1}  R^{(n)}(z) \pppg x ^{-\d_2} }_{\Lc(L^2)}  \leq C \left( 1 + \abs z ^{d-2- n} \right).
\]
\end{theorem}

Theorems \ref{th-low-freq} and \ref{th-low-freq-bis} are proved at the end of Section \ref{sec-small}, after Proposition \ref{prop-b589}. Note that we obtain estimates at all orders $n$, but only those given in Theorem \ref{th-low-freq} for $n \leq d$ will contribute to the proof of Theorem \ref{th-loc-decay}.

Note also that in Theorem \ref{th-low-freq-bis} the condition on $\d_1 + \d_2$ is automatically fulfilled when $n \geq \frac d 2$. We need a weight $\pppg x ^{-n-\frac 12-0^+}$ on both sides to have a uniform estimate for $R^{(n)}(z)$ when $\Im z$ goes to 0 (see Theorem \ref{th-inter-freq}) and an additionnal $\pppg x \inv$ arbitrarily distributed among the left and the right to improve the dependence on $\Re z$ and get a uniform estimate when $n$ is not too large. This fact has already been emphazised in \cite{bonyh10b} when $n = 0$. We also remark that in passing we have improved the weights used in \cite{bouclet11}.\\

For high frequencies, we know that the propagation along classical rays is in some sense (made rigorous by semiclassical analysis) a good approximation for the propagation of the wave. This is why geometric assumptions are crucial for the local energy decay: the energy of the wave (at least the contribution of high frequencies) escapes at infinity if the classical rays go to infinity. In the dissipative case, we know that this non-trapping assumption can be replaced by a damping condition on bounded classical trajectories: the energy which does not escape at infinity has to be dissipated by the medium (see \cite{alouik07, art-mourre} for resolvent estimates of the dissipative Schr\"odinger operator on $\R^d$, see also \cite{datchevv,christiansonsvw} for the correspondence between non-trapped and damped trajectories).

\begin{theorem}[Resolvent estimates for high frequencies] \label{th-high-freq}
Let $n \in \N$ and $\d > n+ \frac 12$. Under the damping condition \eqref{hyp-amort} on the classical flow, there exists $C \geq 0$ such that for all $z \in \C_{+}$ with $\abs z \geq C$ we have
\[
\nr {\pppg x ^{-\d} R^{(n)}(z) \pppg x ^{-\d}} _{\Lc(L^2)} \leq \frac C { \abs z }
\]
and
\[
\nr {\pppg x ^{-\d} \nabla R^{(n)}(z) \pppg x ^{-\d}} _{\Lc(L^2)} \leq  C .
\]
\end{theorem}

In the self-adjoint case (see \cite{bonyh12,bouclet11}) the local energy decay problem can be dealt with by using general estimates of the form
\[
\nr{\pppg x ^{-\g} e^{-\frac {it}h Q_h} \h(Q_h) \pppg x ^{-\g}} \leq C \pppg t^{-\g},
\]
where $\g \geq 0$, $Q_h$ is a semiclassical self-adjoint Schr\"odinger operator and $\h \in C_0^\infty(\R)$ is supported near a non-trapping energy for $Q_h$. We refer to \cite{wang87} for more details and a proof of this result. However such an estimate is not available in the dissipative situation (we cannot even give a sense to $\h(Q_h)$) and this explains the loss of regularity in Theorem \ref{th-loc-decay} (see Proposition \ref{prop-Uc1} below). Proving such an improvement on high frequency estimates is an important problem on its own which cannot be solved trivially. In this paper we mainly focus on the low frequency part, which already requires a fairly long analysis. Notice however that even if they do not provide directly the sharpest inequalities (in terms of the regularity of the initial data) for the time dependent problem, the estimates of Theorem \ref{th-high-freq} are optimal in the sense that we recover the self-adjoint ones in the non-trapping case.\\


In order to prove Theorems \ref{th-inter-freq}, \ref{th-low-freq}, \ref{th-low-freq-bis} and \ref{th-high-freq} we use the commutators method of E. Mourre (\cite{mourre81}, see \cite{amrein} and references therein for an overview on the subject). In \cite{art-mourre}, the second author has generalized the original result of Mourre to the dissipative setting. Here we also extend the results of \cite{jensenmp84,jensen85} about the derivatives of the resolvent. More precisely we first study the powers of the resolvent, and then prove that we can insert some suitable factors between these resolvents.\\


Let us close this introduction by fixing some general notation (for the reader's convenience, the main notations of the paper, including some technical ones, are recorded in Appendix \ref{section-notations}). We set
\[
\C_{\pm,+} = \singl{z \in \C \st \pm \Re (z) > 0 , \Im (z) > 0}.
\]
Given $I\subset \R$ we also define
\[
\C_{I,+} = \singl{z \in \C \st  \Re (z) \in I , \Im (z) > 0}.
\]
For $m \in \N$ we denote by $C_0^\infty(\R^m)$ the set of smooth and compactly supported functions on $\R^m$, $\Sc(\R^m)$ is the Schwartz class, $\symbor(\R^m)$ is the set of smooth functions whose derivatives of all orders are bounded, and for $\d \in \R$ we denote by $\symb^\d(\R^m)$ the set of smooth functions $u$ on $\R^m$ such that 
\[
 \forall \a \in \N^m, \exists c_\a \geq 0, \forall x \in \R^m, \quad \abs{\partial^\a u(x)} \leq c_\a \pppg x ^{\d -\abs \a}.
\]
For $j \in \Ii 1 d$ (everywhere we write $\Ii n m$ for $[n,m] \cap \N$) we set $D_j = -i\partial_j$. If $A$ and $H$ are operators, we denote by $\ad_A(H)$ the commutator $[H,A]$. \\


\noindent 
{\bf Acknowledgements.}
We are grateful to the referees for their careful reading of the first version of the paper and for suggesting several improvements on the presentation.
The research of the second author is partially supported by the French National Research Project NOSEVOL (Non self-adjoint operators, semiclassical analysis and evolution equations, ANR 2011 BS01019 01).

\section{Outline of the paper} \label{sec-outline}

In this section we describe the different sections of the paper, that is the global strategy of the proof of Theorem \ref{th-loc-decay}. We also give some heuristics which should help the reader to understand the key arguments. We explain in particular the main ideas of the proof of Theorem \ref{th-low-freq}, which is the most technical part of the paper.\\

%
In Section \ref{sec-diss} we give some general properties of maximal dissipative operators. Once we have proved that the operator $\Ac$ defined by \eqref{def-Ac} is maximal dissipative, we obtain in particular the well-posedness of the problems \eqref{wave-A} and hence \eqref{wave-lap}. We also give some dissipative properties of the resolvent $R(z)$ (as defined in \eqref{def-Rz}) which will be useful later in the paper.\\

%
%
%
%
%
%
%
%


In Section \ref{sec-time} we derive the local energy decay stated in Theorem \ref{th-loc-decay} from the resolvent estimates of Theorems \ref{th-inter-freq}, \ref{th-low-freq} and \ref{th-high-freq}. For this we take the Fourier transform in time in the wave equation \eqref{wave-lap} and obtain in a standard way (at least formally)
%
%
%
\begin{equation} \label{fourier-integral}
(\nabla , \partial _t) u (t) = \frac 1 {(it)^n} \int_{\R} e^{-it\t} \frac {d^n}{d\t^n} \big((\nabla ,-i\t) R(\t+i0) (au_0 -i\t u_0 +u_1) \big) \, d\t.
\end{equation}
Here $R(\t + i0)$ denotes the limit of $R(\t + i\m)$ when $\m \searrow 0$. In this integral we introduce the following partition of unity in $\t$
\begin{eqnarray}
 1 = \chi_0 (\t) + \sum_{h \ {\rm dyadic}} \chi \left( h \t \right) , \label{partdyadic}
\end{eqnarray} 
with $ \chi_0 \in C_0^{\infty}(\R)  $, $ \chi \in C_0^{\infty}(\R \setminus 0) $,
 and treat separately the different regimes: intermediate frequencies ($\abs \t \sim 1$ ), low frequencies ($\abs \t \ll 1$) and high frequencies ($\abs \t \gg 1$).

The contribution of intermediate frequencies follows fairly easily from Theorem \ref{th-inter-freq}. Here $n$ can be as large as we wish, which means that the corresponding contribution decays rapidly in time.

The restriction on the time decay in Theorem \ref{th-loc-decay} is due to low frequencies. On the support of  $\chi_0 $, the integrand in (\ref{fourier-integral}) is controlled by Theorem \ref{th-low-freq} as long as $n \leq d-1$, since it is at worst of size $\abs \t^{-\e}$, which is integrable around 0. This gives a decay of order $O\big(t^{-(d-1)}\big)$. When $d = n$, we only have an estimate of size $\t^{-1-\e}$, which is why we do not reach a time decay of size $O\big(t^{-d}\big)$. But $\t^{-1-\e}$ is not far from being integrable, so with an interpolation argument (Lemma \ref{lem-holder2}) we can finally prove a decay of size $O\big(t^{-(d-\e)}\big)$ for any $\e > 0$.

On the support of $ 1 - \chi_0 (\t) $, the resolvent estimates (once appriopriatly weighted) are $O(1) $ in $ \tau $, which does not give integrability. We overcome this problem by introducing the dyadic decomposition in (\ref{partdyadic}) to be able to use an almost orthogonality argument.
The rough idea   is to convert the cutoff $ \chi (h \tau) $ into a spectral cutoff  $ \chi (h H_0^{1/2}) $.  This works well in the self-adjoint case but is more tricky here for the following reason. In the self-adjoint case, if $ \tilde{\chi} \equiv 1 $ near the support of $ \chi  $,  the Spectral Theorem provides  the estimate
$$ \nr{\chi (h \t) \big(1-\tilde{\chi}(hH_0^{1/2})\big) \big(H_0 - (\tau \pm i 0)^2 \big)^{-1}}_{L^2 \rightarrow L^2} \lesssim h^2 \sim \tau^{-2} $$
as well as similar estimates for powers of the resolvent. In our case, proving such an estimate when we replace the resolvent  $ \big(H_0 - (\tau \pm i 0)^2 \big)^{-1} $ by $ R (\tau + i 0)  $ (but keep of course the specral cutoff $ \big(1-\tilde{\chi}(hH_0^{1/2})\big)$ requires much more work, part of which is the purpose of Section \ref{sec-time}. Using this analysis we can then infer that, up to nicer remainders, we may replace $ R (\tau + i 0) $ by
$$ \tilde{\chi}(h H_0^{1/2}) R (\tau + i 0) \tilde{\chi}(h H_0^{1/2}) $$
which then allows to sum over $h$ by  almost orthogonality.
Let us point out that in this analysis there is no restriction on $n$ so that we prove a fast decay in time. \\



Then it remains to prove the resolvent estimates of Theorems \ref{th-inter-freq}, \ref{th-low-freq}, \ref{th-low-freq-bis} and \ref{th-high-freq}. 
They will all follow from the commutators method of Mourre. \\

In Section \ref{sec-Mourre}, we provide  a suitable version of the Mourre theory which allows to estimate powers of the resolvent, including inserted factors, for general dissipative perturbations of self-adjoint operators.  
In \cite{art-mourre}, the second author considered dissipative operators of the form  $H = H_1 -i V$ where $V \geq 0$ is relatively bounded with respect to the self-adjoint operator $H_1$. It was proved there that if a  commutator estimate of the form 
\begin{equation} \label{outline-hyp-mourre}
\1 J (H_1) \big( [H_1,iA] + \b V \big) \1 J (H_1) \geq \a \1 J (H_1)
\end{equation}
holds for some self-adjoint operator $A$ (here $J \subset \R$, $\b \geq 0$, $\a > 0$),
then
\[
\sup _ {\substack{ \Re z \in I \Subset J \\ \Im z > 0}} \nr{\pppg A ^{-\frac{1}{2} - \epsilon} (H-z)\inv \pppg A ^{-\frac{1}{2} - \epsilon}}_{L^2 \rightarrow L^2} < \infty.
\]
When $H$ is (a perturbation of) the flat Laplacian on $\R^n$, one usually takes (a perturbation of) the generator of dilations for $A$. This will be the case here. 

In the self-adjoint case, it is also well known that one has similar estimates for powers of the resolvent.
As pointed out in  \eqref{expr-der-R}, the derivatives of $R(z)$ involve not only powers of $R(z)$ but also such powers with inserted factors. In Section \ref{sec-Mourre}, we show in a general setting that we can estimate both pure powers of $  (H-z)^{-1}$ and such powers with inserted operators $\Th$,
as long as $\Th$ has reasonnable commuting properties with the conjugate operator $A$. Since $\Th$ is inserted between resolvents of $H$, we can even allow unbounded operators $\Th$ which are relatively bounded with respect to $H$. This will be important in our case even if the multiplication by $a(x)$ is a bounded operator on $L^2$.\\

In Section \ref{sec-inter-freq}, we derive fairly directly Theorem \ref{th-inter-freq}, about intermediate frequencies, from the general theory of Section \ref{sec-Mourre}. Even here we have to use a parameter-dependent version of the resolvent estimates (since the operator $H_0 - i z a (x)$ depends on the spectral parameter $z$), but this does not rise any particular difficulty in this case.\\


Low frequencies are much more problematic. As is well-known, the reason is that the standard Mourre method with a Laplacian and the generator of dilations only works near a positive energy.\\

In Section \ref{sec-small} we prove Theorems \ref{th-low-freq} and \ref{th-low-freq-bis}. We  observe first that
\begin{eqnarray} R^{(n)}(z) = \mbox{linear combination of } \ z^k R (z) a^{j_1} R (z) \cdots R (z) a^{j_m} R (z)  \label{rec2}
\end{eqnarray}
with $0 \leq m \leq n$, $j_1,\dots,j_m \in \{0,1\}$ and $n + k + j_1 + \cdots + j_m = 2 m$ (see Proposition \ref{prop-der-R2}).
We explain here the ideas in the cases where there is at most one inserted factor $a$, that is $\sum j_k \leq 1$. 

Our analysis begins with the scaling argument of \cite{bouclet11}. We introduce
\begin{eqnarray*}
 \Ptau =e^{-i A \ln |z|} \frac{H_z}{|z|^2} e^{i A \ln |z|} 
\quad \text{and} \quad \tRz = \big( \Ptau - \hat{z}^2 \big)^{-1},
\end{eqnarray*}
where $\hat z = z /\abs z$ and $e^{iA\ln\abs z}$ is the dilation by $\abs z$ : $\big(e^{iA\ln\abs z} u \big)(x) = |z|^{d/2} u(\abs z x)$.
For $\Re z \geq 0$, the new spectral parameter $\hat z^2$ only approaches the real axis at point 1, which is what we want to use the Mourre theory. To implement  this idea, we only consider a small perturbation of the Laplacian. This amounts to substract a compactly supported part to the full perturbation whose contribution will be studied afterwards by a compactness argument which we do not discuss here (see Subsection \ref{sec-non-small}). Then, the smallness of $-\Delta - \Ptau$ will (in particular) allow to apply the Mourre  method with the generator of dilations $A$ (see Proposition \ref{prop-estim-res-amort}).

When all $ j_k $ vanish, we have to estimate operators of the form
\begin{equation} \label{interminable}
|z|^k \pppg{x}^{-\delta}  R (z)^{m+1}  \pppg{x}^{-\delta} = \abs z^{k - 2m-2} \pppg{x}^{-\delta} e^{iA\ln\abs z} \tRz ^{m+1} e^{-iA\ln\abs z}\pppg{x}^{-\delta}.
\end{equation}
Using the resolvent identity for $\Ptau$, this term can be rewritten as a sum of terms of the form 
\begin{equation*} 
E_z := \abs z^{k - 2m-2} \pppg{x}^{-\delta}  e^{i A \ln |z|} \big( \Ptau + 1 \big)^{-(m_e+1)} e^{-iA \ln|z|} \pppg{x}^{-\delta}
\end{equation*}
where $m_e \geq m$, and
\[
\Sigma_z :=  \abs z^{k - 2m-2} \pppg{x}^{-\delta}  e^{i A \ln |z|} \big( \Ptau + 1 \big)^{-N_1} \tRz^{m_s+1} 
\big( \Ptau + 1 \big)^{-N_2} e^{-iA \ln|z|} \pppg{x}^{-\delta},
\]
where $ m_s \leq n$ and $N_1,N_2$ are as large as we wish. We omit some factors $(1+\hat z^2)$ which do not play any role. Since $\Ptau$ is close to the Laplacian, we can prove the elliptic estimates
\begin{equation} \label{estim-res-outline}
\nr{\big( \Ptau + 1 \big)\inv}_{H^{s-1} \to H^{s+1}} \lesssim 1 \quad \text{for } s \in \big] - d/2 ,d/2\big[,
\end{equation}
where the limitation in the range of $s$ is due to the form of the coefficients of $ \Ptau $ (typically $ G_{jk}(x/|z|) $). If $m_e < \frac d 2$ then $\big( \Ptau + 1 \big)^{-(m_e+1)}$ is uniformly bounded from $H^{-m_e}$ to $H^{m_e}$. Using the Sobolev embedding $H^{m_e} \subset L^{q}$ with $q = \frac {2d}{d-2{m_e}}$ and the simple but crucial estimate
\begin{equation} \label{estim-eiA-outline}
\nr{e^{iA\ln\abs z}}_{L^{q} \to L^{q}} \lesssim \abs{z} ^{m_e},
\end{equation}
together with adjoint estimates (see Proposition \ref{prop-A}), we finally obtain that the elliptic term $E_z$ is uniformly bounded on $L^2$ (we use the weight $\pppg x^{-\d}$ to map $L^{q}$ to $L^2$). When ${m_e}$ is too large we only get $\abs z^{d-\e}$ in (\ref{estim-eiA-outline}) by this method and then $E_z$ is of size $O \big(\abs{z}^{d-\e - n-2}\big)$ (the possible removal of $\e$ will be discussed later).

We next consider $ \Sigma_z $.
Using the dissipative Mourre theory of Section \ref{sec-Mourre}, we know that the operator 
\begin{equation} \label{outline-mourre}
\pppg {A}^{-\d}   \tRz^{m_s+1} \pppg {A}^{-\d}
\end{equation}
is uniformly bounded as an opertor on $L^2$. Thus we have to estimate an operator of the form 
\[
W_z = \pppg x ^{-\d} e^{iA \ln \abs z} \big( \Ptau + 1 \big) ^{-N_1} \pppg A ^\d.
\]
Let $s \in\big[0, \frac d 2\big[$ and $q = \frac {2d}{d-2s}$. Using Sobolev embeddings we obtain 
\[
\nr{W_z}_{L^2 \rightarrow L^2} \leq \nr{\pppg x ^{-\d} e^{iA \ln \abs z}  \big(1 + \abs x^\d\big)}_{H^s \to L^2} \nr{\big(1 + \abs x^\d\big)\inv \big( \Ptau + 1 \big) ^{-N_1} \pppg A ^\d}_{L^2 \to H^s}. 
\]
The second factor is uniformly bounded. The rough idea is that the powers of $x$ given by $\pppg A^\d \simeq 1 + \abs x^\d \abs D^\d$ are controlled by $(1 + \abs x^\d)\inv$ and the derivatives by the resolvents. Of course this is not simple pseudo-differential calculus, so many explicit commutators will be involved. In particular we recall that there is a restriction for the Sobolev index in \eqref{estim-res-outline}, so we cannot simply control a derivative of high order by $\big( \Ptau + 1 \big) ^{-N_1}$ even if $N_1$ is large. For the first factor we write 
\begin{equation} \label{decomp-delta}
\pppg x ^{-\d} e^{i A \ln \abs z} \big(1 + \abs x^\d\big) = \pppg x ^{-\d} e^{i A \ln \abs z} +  \abs z^\d \pppg x ^{-\d} \abs x^\d e^{iA\ln \abs z} .
\end{equation}
The second term in \eqref{decomp-delta} is clearly of order $\abs z^\d$, and for the first term we use \eqref{estim-eiA-outline}. We remark that the weight is used either to go from $L^q$ to $L^2$ or to control the powers of $x$ given by $\pppg A^\d$, but not for both at the same time. Finally, if $\d \geq s$ we obtain that $W_z$ is of size $O(\abs z^s)$ and hence $\Sigma_z$ is of size $O\big(\abs z^{2s-n-2}\big) = O\big(\abs z^{d-\e - n- 2}\big)$ 
 since $s$ can be chosen arbitrarily close to $ d/2 $. \\



To study next the case where $ \sum_k j_k = 1 $, it is useful to recall that the inserted factors come from \eqref{expr-der-R}, where factors $(2z +ia)$ appear instead of $2z$ is the self-adjoint case. It is thus natural to seek an estimate of size $O(\abs z)$ for the contribution of $a$. This is actually possible for the following reason. 
After rescaling, the contribution of the inserted factor reads $a_{\abs z} := a(\cdot / \abs z) = e^{-iA\ln\abs z} a e^{iA \ln \abs z}$. Since $a$ is of short range, it turns out  that
\begin{equation} \label{outline-az}
\forall s \in \big] - d/ 2 ,  d/ 2 -1\big[ , \quad \nr{a_{\abs z}}_{H^{s+1} \to H^s} \lesssim \abs z
\end{equation}
(see Proposition \ref{prop-dec-sob}). So $a_{\abs z}$ behaves like a derivative and is indeed of size $O(\abs z)$ at low frequencies.

This being said, let us come back to the estimate of the analogue of (\ref{interminable}) when one  of the $j_k$ of (\ref{rec2}) is equal to $1$.
%
When estimating a term involving $ (\Ptau + 1)^{-m_1} a_{|z|} (\Ptau + 1)^{- m_2 }$ with $ m_1 + m_2 = {m_e}+1 $, 
 (\ref{outline-az}) costs one derivative but provides one power of $|z|$. This loss of derivative may in some cases be at the expense of using a slightly worse Lebesgue exponent in (\ref{estim-eiA-outline}) but, in the end, we  recover the same estimates as when there was no inserted factors.
%
 %
 Concerning an inserted factor $a_{\abs z}$ in \eqref{outline-mourre}, we can directly apply the abstract results of Section \ref{sec-Mourre}. Note that even if $a_{\abs z}$ is uniformly bounded on $L^2$, we have to use the version with unbounded inserted factors (Theorem \ref{th-estim-insert2}) to use \eqref{outline-az}.

For the second estimate in Theorem \ref{th-low-freq} we only observe that we can proceed for a derivative as we did for $a_{\abs z}$, except that we do not have the restriction on the Sobolev index as in \eqref{outline-az}.\\

The argument described in the previous paragraph allows to prove Theorem \ref{th-low-freq} where we have an $\e$ loss (which is harmless for the time decay estimate). Let us now briefly explain how to remove this loss in certain cases (this is the purpose of Theorem \ref{th-low-freq-bis}).
%
%
Consider for instance
$\abs z^{-4} \big( \Ptau + 1 \big) \inv a_{\abs z} \big( \Ptau + 1 \big) \inv$ which appears when $d = 3$ and $n=1$. It is of size $O(\abs z^{-3})$ as an operator from $H^{-3/2}$ to $H^{3/2}$, which is critical for Sobolev embeddings. Using only \eqref{estim-eiA-outline} we  get an estimate of size $O(\abs z ^{-\e})$. Now we can improve \eqref{outline-az} as follows: if $\s < \rho$, $\rho$ given by \eqref{dec-metric-a}, we can modify (\ref{outline-az}) and see $a_{\abs z}$ as a derivative of order $1+\s$, and hence of size $\abs {z}^{1+\s}$ (see Proposition \ref{prop-Th}). Then our operator is now of size $O(\abs z^{-3 + \s})$ as an operator from $H^{-(3-\s)/2}$ to $H^{(3-\s)/2}$. We avoid the critical Sobolev index and 
get a uniform estimate. The point is that we can always do this, except when $d$ is even and $n = d-2$. For instance when $d=4$ and $n=2$ there is a term $\abs z^{-4} \big( \Ptau + 1 \big) ^{-2}$ which cannot be estimated uniformly by this method. \\

In Section \ref{sec-high-freq} we deal with high frequencies. The proof of Theorem \ref{th-high-freq} is also quite technical but the ideas are mostly well-known (see \cite{art-mourre}), so we only give a quick overview here. In the self-adjoint case, the idea is to apply Mourre Theory to the semi-classical Schr\"odinger operator with a conjugate operator given by the quantization of an escape function (according to the trick of \cite{gerardm88}).
In our dissipative setting, we construct a symbol which is increasing along the flow outside the damping region. Roughly speaking, that we can relax the usual non-trapping condition comes from the term $\b V$ in \eqref{outline-hyp-mourre} which provides additional positivity in the damping region.\\

In Section \ref{sec-laplacien}, we prove the analogue of Theorem \ref{th-loc-decay} when $ H_0 $ is replaced by a Laplace-Beltrami operator associated to a long range perturbation of the flat metric. 

\section{Resolvent of dissipative operators} \label{sec-diss}

In this section we record some general properties on the spaces and the unbounded operators we shall use. 
In particular, we deduce  that the problem \eqref{wave-lap} has a solution defined for all positive times, and that the resolvent $R(z)$ is well-defined for $z \in \C_+$ with nice properties away from the real axis.\\

We first remark that the norm on $\Hc$ is equivalent to the norm $(u,v) \mapsto \nr{\nabla u}_{L^2} + \nr v_{L^2}$ of $\dot H^1 \times L^2$. Then we recall the following classical proposition:

\begin{proposition} \label{proprieteH0} 
\begin{enumerate}[(i)] 
\item The space $ {\mathcal S}({\mathbb R^d} ) $ is stable by the resolvent of $ H_0 $ and by $ f (H_0) $ for any $ f \in C_0^{\infty}({\mathbb R}) $.
\item If $ f \in C_0^{\infty}({\mathbb R}) $, $ f (H_0) $ maps compactly supported $L^2$ functions into $ {\mathcal S} ({\mathbb R}^d) $.
\end{enumerate}
\end{proposition}

Next, to emphasize that $ \dot{H}^1 $ is not contained in $L^2$, we record the following characterization
$$ \dot{H}^1 = \left\{ u \in L^{\frac{2d}{d-2}} \st \nabla u \in L^2 \right\} , $$
where $ \nabla u $  is taken in the  distributions sense.\\

The space $ {\mathcal D}({\mathcal A}) $ introduced in \eqref{dom-A} will be equipped with the
norm
$$ \nr{(u,v)} _{{\mathcal D}({\mathcal A})} = \nr{(u,v)}_{\mathcal H} + \nr{H_0 u}_{L^2} + \nr{\sqrt {\Ho} v}_{L^2} . $$
Notice that, if $ (u,v) \in {\mathcal D}({\mathcal A}) $, $ u $ does not belong to the domain of $ H_0 $ in general since we do not know that $ u \in L^2 $. In practice, we can get rid of such a problem by using the following proposition.

\begin{proposition} \label{densite-Banach}
\begin{enumerate}[(i)] 
\item{ The space $ {\mathcal S}({\mathbb R^d} ) \times {\mathcal S}({\mathbb R}^d) $ is dense in $ {\mathcal D}({\mathcal A}) $. }
\item{ $ {\mathcal D} ({\mathcal A}) $ is a Banach space.} 
\end{enumerate}
\end{proposition}

\begin{proof} Most of the proof is routine. We only sketch the main points of the item (i).  Let $ \chi \in C_0^{\infty}({\mathbb R}) $ be equal to $ 1 $ near $ 0 $. We show first that, if $ (u,v) \in {\mathcal D}({\mathcal A}) $, then
$$ (u_{\e} , v_{\e}) := \big( \chi (\e x) u , \chi (\e x) v \big) $$
also belongs to $ {\mathcal D}({\mathcal A}) $ and converges to $ (u,v) $ for $ \nr{\cdot }_{{\mathcal D}({\mathcal A})} $. For this purpose, we use that $ \e \nabla \chi (\e x) u $ goes to zero in $ L^2 $ since the operator $ \abs{x} \e \nabla \chi (\e x) $ goes strongly to zero on $ L^2 $ and $ \abs{x}^{-1} u $ belongs to $ L^2 $ by the Hardy inequality. A similar argument shows that $ H_0 u_{\e} \rightarrow H_0 u $ in $ L^2 $. Then, it suffices to approach $ (u_{\e},v_{\e}) $ by Schwartz functions.  We introduce
$$ (u_{\e,n} , v_{\e,n}) := \big( \chi (H_0 / n) u_{\e} , \chi (H_0 / n) v_{\e} \big) . $$
By the item (ii) of Proposition \ref{proprieteH0}, this pair belongs to $ {\mathcal S} ({\R}^d) $. This uses in particular  that $ u_{\e} \in L^2 $ which also shows that $ u_{\e}  $ belongs to the domain of $ H_0$. It only remains to see that $ (u_{\e,n} , v_{\e,n}) \rightarrow  (u_{\e} , v_{\e}) $ in $ {\mathcal D} ({\mathcal A})$, which is clear since $u_{\e,n} \to u_\e$ in $\Dom(H_0)$ and $v_{\e,n} \to v_\e$ in $\Dom(H_0^{1/2})$.
%
\end{proof}

Let us now introduce dissipative operators:

\begin{definition} \label{def-diss}
We say that the operator $T$ with domain $\Dom(T)$ on the Hilbert space $\Kc$ is \emph{dissipative} if 
\[
 \forall \f \in \Dom(T), \quad \Im \innp {T \f } \f \leq 0
\]
(here the inner product is anti-linear on the right).
Moreover $T$ is said to be \emph{maximal dissipative} if it has no other dissipative extension on $\Kc$ than itself. 
\end{definition}

A dissipative operator $T$ is maximal dissipative if and only if $(T-\z)$ has a bounded inverse on $\Kc$ for some (and hence any) $\z \in \C_+$. In this case we have
\begin{equation} \label{estim-res-diss}
\forall \z \in \C_+, \quad \nr{(T-\z)\inv} _{\Kc} \leq \frac 1 {\Im \z}.
\end{equation}
This estimate together with the Hille-Yosida Theorem proves that $-iT$ generates a contractions semigroup. Then for any $\f_0 \in \Dom(T)$, the function $\f : t \mapsto e^{-itT} \f_0$ belongs to $C^1 (\R_+ , \Kc) \cap C^0 (\R_+ , \Dom(T))$ and solves the problem 
\[
\begin{cases} (\partial _t  + iT) \f(t) = 0, \quad \forall t \geq 0, \\ \f(0) = \f_0.
\end{cases}
\]

\begin{proposition} \label{prop-R-diss}
\begin{enumerate}[(i)]
\item For all $z \in \C_{+,+}$ the operators $\Ho -iza$ and $-i(\Ho - iza)$ are maximal dissipative with domain $H^2$. 
\item  For all $z \in \C_+$, the operator
$
\Ho -iza(x) -z^2
$
defined on $H^2$ has a bounded inverse on $L^2$
\[
R(z) = \big(\Ho -iza(x) -z^2\big) \inv, 
\]
as introduced in \eqref{def-Rz}. Moreover $R(-\bar z) = R(z)^*$.
\item There exists $C \geq 0$ such that for all $z \in \C_+$ we have
\[
 \nr{R(z)}_{\Lc(L^2)} \leq \frac C {\Im z (\Im z + \abs {\Re z})}.
\]
\end{enumerate}
\end{proposition}

\begin{proof}
(i)\quad The operator $\Ho$ is self-adjoint and non-negative on $L^2$, with domain $ H^2 $. In particular $\Ho$ and $-i\Ho$ are maximal dissipative on $H^2$. If $z \in \C_{+,+}$ the operators $-iza(x)$ and $-za$ are dissipative and bounded on $L^2$, so the operators $\Ho -iza(x)$ and $-i\Ho - za$ are maximal dissipative on $H^2$ by a standard perturbation argument (see Lemma 2.1 in \cite{art-mourre}). 

\noindent
(ii) \quad If $z \in \C_{+,+}$ then $\Ho -iza$ is maximal dissipative and $\Im (z^2) > 0$ so $\Ho - iza(x) - z^2$ has a bounded inverse. For $z \in \C_{-,+}$ we can use the equality
\[
  \Ho - i za(x) -  z^2= \big( \Ho + i \bar za(x) - \bar z^2 \big) ^*.
\]
If $\Re z = 0$, then we only have to remark that $\Ho -iza(x)$ is self-adjoint and non-negative, and $z^2 < 0$. 

\noindent
(iii) \quad According to \eqref{estim-res-diss} applied to $\Ho -iza$, we have for all $z \in \C_{+,+}$ (and similarly if $z \in \C_{-,+}$ according to (ii))
\[
 \nr{R(z)}_{\Lc(L^2)} \leq \frac 1 {2 \abs {\Re(z)} \Im (z)}.
\]
When $\abs{\Re z} \geq \frac {\Im z}2$, the right-hand side can be replaced by $C/ (\Im z (\Im z + \abs{\Re z}))$. If $\abs {\Re z} \leq \frac {\Im z} 2 $ we apply \eqref{estim-res-diss} with $T = -i\Ho -za$ and $\z = -iz^2$ and obtain
\[
\nr{R(z)}_{\Lc(L^2)} = \nr{(-i\Ho -za + i z^2)\inv }_{\Lc(L^2)} \leq \frac 1 {\Im (z)^2 - \Re (z)^2} \leq \frac 4 {3 \Im (z)^2},
\]
which gives the estimate in this case. This concludes the proof.
\end{proof}

%
%

\begin{proposition} \label{prop-A-diss}
\begin{enumerate}[(i)]
\item
For all $ z \in {\mathbb C}_+ $, the operator
\begin{equation} \label{res-Ac}
 \begin{pmatrix} 
		 R(z) (ia + z) &   R(z)\\
               I +  R(z) (zia + z^2) &  z R(z)
              \end{pmatrix} ,  
\end{equation}
defined on $ {\mathcal S} ({\mathbb R}^d) \times {\mathcal S} ({\mathbb R}^d) $, has a bounded closure in $\Lc(\Hc, \Dom(\Ac))$ which we denote by $R_\Ac(z)$.

 \item 
 The operator $\Ac$ defined in \eqref{def-Ac} is maximal dissipative on $\Hc$. Moreover for all $z \in \C_+$ we have
\[
 (\Ac-z)\inv
 = {\mathcal R}_{\mathcal A}(z).
\]
\end{enumerate}
\end{proposition}

%
%
%
%

\begin{proof} 
\stepp
Let $z \in \C_+$ be fixed in all the proof. We set $R_0(z) = (\Ho-z^2)\inv$.
Using that $a$ is of short range and  the Hardy inequality, we have for $u \in \Sc(\R^d)$
\begin{equation*}
\nr{a R_0(z) u}_{L^2} \lesssim_z \nr{\nabla R_0(z) u}_{L^2} \lesssim_z \nr{H_0^{1/2} R_0(z) u}_{L^2}\lesssim_z \nr{H_0^{1/2} u}_{L^2}. 
\end{equation*}
By the resolvent identity
\begin{equation} \label{sousresolvent}
 R (z) = R_0(z) + i z R (z) a R_0(z)
\end{equation} 
and the $L^2\to L^2$ boundedness of $\nabla R(z)$, we obtain
\begin{equation}\label{1ere}
\nr{ \nabla R (z) u}_{L^2} \lesssim_z \nr{H_0^{1/2} u}_{L^2}.
\end{equation}
On the other hand 
\[
\nr{ \big( I + z^2 R_0(z) \big) u}_{L^2} = \nr{H_0^{1/2} R_0(z) H_0^{1/2} u}_{L^2} \lesssim_z \nr{H_0^{1/2} u}_{L^2} ,
\]
and using again \eqref{sousresolvent}:
\begin{equation} \label{2eme}
\nr{(I + z^2 R(z))u}_{L^2} \lesssim_z \nr{H_0^{1/2} u}_{L^2}.
\end{equation}

\stepp 
Using the Hardy inequality, \eqref{1ere}, \eqref{2eme} and the fact that $R(z)$ is bounded from $L^2$ to $H^2$ it is easy to conclude that the operator \eqref{res-Ac} extends to a bounded operator from $\Hc$ to $\Hc$. Now let us check that it is bounded from $\Hc$ to $\Dom(\Ac)$. Let $(u,v) \in {\mathcal S} ({\mathbb R}^d) \times {\mathcal S} ({\mathbb R}^d) $ and define
\[
\begin{pmatrix} U \\ V \end{pmatrix} = R_\Ac(z) \begin{pmatrix} u \\ v \end{pmatrix} = \begin{pmatrix} R (z) (i a + z) u+ R (z) v \\  u + R (z)(zia + z^2) u + z R (z) v \end{pmatrix}
\]
By the resolvent identity (\ref{sousresolvent}), we have
$$ H_0 U = z H_0 R_0(z) u + i z^2 H_0 R (z) a R_0(z) u + H_0 R (z)ia u + H_0 R (z)v  , $$
so that, by the same estimates as above,
\begin{equation*}
\nr{H_0 U}_{L^2} \lesssim \left( \nr{H_0^{1/2} u}_{L^2}  + \nr{v}_{L^2} \right) . 
\end{equation*}
Using again the resolvent identity, 
\[
V = H_0 R_0(z) u + iz^3 R (z)a R_0(z) u + R (z) z i a u + z R (z) v ,
\]
and we have similarly
\begin{equation*}
\nr{H_0^{1/2} V}_{L^2} \lesssim \left( \nr{H_0^{1/2} u}_{L^2} + \nr{v}_{L^2} \right) .  
\end{equation*}
Since $ { \mathcal D} ({\mathcal A}) $ is complete and $ {\mathcal S} ({\mathbb R}^d) \times {\mathcal S} ({\mathbb R}^d)  $ is dense in $ {\mathcal H} $, the first statement follows.

\stepp
 Let $(u,v) \in \Dom(\Ac)$. We have
\begin{align*}
\innp{ \Ac(u,v)}{(u,v)}_\Hc
& = \innp{H_0^{1/2} v}{H_0^{1/2} u}_{L^2} + \innp{\Ho u}{v}_{L^2} -i \innp {av}{v}_{L^2}\\
& = 2 \Re \innp{H_0^{1/2} v}{H_0^{1/2} u}_{L^2}  -i \innp {av}{v}_{L^2}.
\end{align*}
Here the formal integration by parts $\innp{H_0 u} {v} = \innp{H_0^{1/2} u}{H_0^{1/2} v}$ can be justified by approximating $u$ by a Schwartz function (recall that $u$ belongs to $\dot H^1$ and is not necessarily in the domain of $H_0^{1/2}$). This yields
\[
 \Im \innp{ \Ac(u,v)}{(u,v)}_\Hc \leq 0
\]
and proves that $\Ac$ is dissipative on $\Hc$. Then we only have to check that, for all $z \in \C_+$, $ {\mathcal R}_{\mathcal A}(z) $  is a two sided inverse for $(\Ac-z)$. Using the density of $ {\mathcal S} ({\mathbb R}^d) \times {\mathcal S} ({\mathbb R}^d)  $ in $ {\mathcal H} $ (by definition) and in $ {\mathcal D}({\mathcal A}) $ (by Proposition \ref{densite-Banach}) and using the continuity of $ {\mathcal R}_A (z) $ and $ {\mathcal A} - z $ as operators from $\Hc$ to $\Dom(\Ac)$ and from $\Dom(\Ac)$ to $\Hc$ respectively, it suffices to check that $  {\mathcal R}_A (z) (\Ac-z) = I $ and $ (\Ac-z)  {\mathcal R}_A (z)= I $ on ${\mathcal S} ({\mathbb R}^d) \times {\mathcal S} ({\mathbb R}^d)$, which is then a simple calculation.
\end{proof}

This proposition ensures in particular that $\Ac$ generates a contractions semigroup and hence the well-posedness of the problem \eqref{wave-A} when $\Uc_0 \in \Dom(\Ac)$.\\

An important property of the resolvent of a dissipative operator are the so-called quadratic estimates. This will be used for the abstract Mourre theory but also in Section \ref{sec-non-small} (see the proof of Proposition \ref{prop-b589}). Note that this result of dissipative nature is already used for the self-adjoint theory since the Mourre technique consists in viewing the positive commutator as a dissipative perturbation of the operator.

\begin{lemma} \label{lem-quad-estim}
 Let $T = T_0 -i T_d$ be a maximal dissipative operator on a Hilbert space $\Kc$, where $T_0$ is self-adjoint, $T_d$ is self-adjoint and non-negative and $ {\mathcal D}(T) \subset {\mathcal D}(T_0) \cap {\mathcal D} (T_d) $. Let $B$ be a bounded operator on $\Kc$ such that $B^*B \leq T_d$. Let $Q \in \Lc(\Kc)$. Then for all $\z \in \C_+$ we have
\[
\nr{B^* (T -\z)\inv Q}_{\Lc(\Kc)} \leq \nr{Q^* (T -\z)\inv Q}_{\Lc(\Kc)}^{\frac 12}.
\]
and
\[
\nr{B^* \big(T^* -\bar \z\big)\inv Q}_{\Lc(\Kc)} \leq \nr{Q^* (T -\z)\inv Q}_{\Lc(\Kc)}^{\frac 12}.
\]
\end{lemma}

Let us recall the proof of this result:

\begin{proof}
For $\f \in \Hc$ we have
\begin{align*}
 \nr{B (T -\z)\inv Q \f}_\Kc^2 
& = \innp{B^* B (T-\z)\inv Q \f} {(T-\z)\inv Q \f}_\Kc\\
& \leq \frac 1 {2i} \innp{ (2iT_d + 2i \Im \z)  (T-\z)\inv Q \f} {(T-\z)\inv Q\f}_\Kc\\
& \leq \frac 1 {2i} \innp{ Q^*(T^* - \bar \z)\inv \big((T^* - \bar \z ) - (T-\z)\big)  (T-\z)\inv Q \f} {  \f} _\Kc\\
& \leq\nr{ Q^* (T -\z)\inv Q}_{\Lc(\Kc)} \nr{\f}_\Kc^2,
\end{align*}
which gives the first estimate. Here we used that $ {\mathcal D}(T) $ is contained in $ {\mathcal D}(T^*) $ which is a simple consequence of the assumption $  {\mathcal D}(T) \subset {\mathcal D}(T_0) \cap {\mathcal D} (T_d) $. The second estimate is proved similarly.
\end{proof}

Applied to $\Ho -iza$ and $-i(\Ho -iza)$, Lemma \ref{lem-quad-estim} gives the following estimate for the resolvent $R$:

\begin{proposition}\label{prop-res-diss-am}
Let $R$ be given by \eqref{def-Rz} and $Q \in\Lc(L^2)$. Then for all $z \in \C_+$ we have
\[
 {\abs z} \nr{\sqrt {a(x)} R(z) Q}^2_{\Lc(L^2)} \leq   \sqrt 2 \nr {Q^* R(z)  Q}_{\Lc(L^2)} .
\]
\end{proposition}

\begin{proof}
Let $z \in \C_{+}$. If $\Re z \geq \Im z$ we apply Lemma \ref{lem-quad-estim} with $T_0= \Ho + \Im (z) a$, $T_d = \Re (z) a$ and $\z = z^2$, $B = B^* = \sqrt{\abs z a}$ which gives
\[
{\abs z} \nr {\sqrt {a(x)} R(z) Q}^2 \leq   \sqrt 2  \nr{\sqrt {\Re (z) a(x)} R(z) Q}^2  \leq  \sqrt 2 \nr{Q^* R(z) Q}.
\]
The case $\Re(z) \leq - \Im z$ is proved similarly, using the second estimate of Lemma \ref{lem-quad-estim}.
Then assume that $\Im z \geq \abs{\Re z}$. In this case we apply Lemma \ref{lem-quad-estim} with $T_0 = -\Re (z) a$, $T_d = \Ho + \Im (z) a$, $B = B^* = \sqrt {\Im (z) a}$ and $\z = - i z^2$, which gives the same estimate.
\end{proof}

\section{Time decay for the solution of the wave equation}  \label{sec-time}

\newcommand{\finter} {\rm{c}}

In this section we show how to derive the local energy decay of Theorem \ref{th-loc-decay} from  Theorems \ref{th-inter-freq}, \ref{th-low-freq} and \ref{th-high-freq} about uniform resolvent estimates.\\

Let $(u_0,u_1) \in \Dom(\Ac)$. We denote by $u \in C^1 (\R_+,\Hc)$ the function given by
\[
 (u(t) , i \partial_t u(t)) = e^{-it\Ac} (u_0,u_1).
\]
Then for $\m > 0$ we define
\[
u_\m : t \mapsto e^{-t\m} \1{\R_+}(t) u(t) \quad  \text{and} \quad \tilde u_\m : t \mapsto e^{-t\m}\1{\R_+}(t) \partial_t u(t).
\]
Since $e^{-it\Ac}$ is a contraction of $\Hc$ for all $t\geq 0$, the function $t \mapsto a \partial_t u$ belongs to $L^\infty(\R_+, L^2(\R^d))$. Then, since $(\partial_t u, H_0 u - ia\partial_t u) = \partial_t (u , i \partial_t u) = -i\Ac (u , i \partial_t u) = -ie^{-it\Ac} \Ac (u_0,u_1)$, the functions $\Ho u$, $\partial_t^2 u$ are also bounded with respect to $t$. In particular for all $\m > 0$ the functions $u_\m$ and $\tilde u_\m$ decay rapidly in time, which justifies the computations below.\\

We consider the inverse Fourier transform of $u_\m$, defined for $\t \in \R$ by
\begin{equation*} 
(\Four\inv u_\m) (\t) = \int_{0}^{+\infty} e^{it\t} u_\m(t) \, dt.
\end{equation*}
For all $n \in \N$ it satisfies
\begin{equation} \label{fourier-um}
\frac {d^n}{d\t^n} (\Four\inv u_\m) (\t) =  \int_{0}^{+\infty} (it)^n  e^{it\t} u_\m(t) \, dt.
\end{equation}
Let $\t \in \R$ and $z = \t +i\m$. We multiply \eqref{wave-lap} by $e^{it(\t +i\m)}$ and integrate over $t \in \R_+$. After partial integrations we obtain
\begin{equation} \label{rel-Four-R}
\Four\inv u_\m(\t) = R(z) (au_0 -izu_0 +  u_1),
\end{equation}
and then:
\begin{equation} \label{Four-tildeu}
\Four\inv \tilde u_\m(\t) = - u_0 - iz \Four\inv u_\m(\t) = - u_0 - z R(z) (ia u_0 + z u_0 +iu_1). 
\end{equation}

\begin{remark}
All the computations below could have been performed with the resolvent of $\Ac$ instead of $R$, starting from the relation 
\[
(\Ac-(\t+i\m))\inv(u_0,iu_1) = i \int_0^{+\infty} e^{-it(\Ac-(\t+i\m))} (u_0,iu_1)\, d\t  
\]
instead of \eqref{rel-Four-R}. However, the proofs given below would be more complicated from this point of view.
\end{remark}

Now we prove the time decay of $\nabla u_\m$ and $\tilde u_\m$, keeping in mind that all the estimates which are uniform in $\m > 0$ remain true for $\nabla u$ and $\partial_t u$.
Theorems \ref{th-inter-freq}, \ref{th-low-freq} and \ref{th-high-freq} provide uniform estimates on the derivatives of $\Four\inv u_\m$ and $\Four\inv \tilde u_\m$. In order to obtain information on $u_\m$, we first have to inverse the relation \eqref{fourier-um}.

\newcommand{\nonmu}{\n}

%
%
%
%

\begin{proposition}
Let $\m > 0$, $v \in \Sc(\R^d)$, $n \geq 2$ and $\h_0 \in C_0^\infty(\R,[0,1])$ be equal to 1 on a neighborhood of 0. Let $\a \in\N^d$ with $\abs a = 1$. Then for all $\d > n+\frac 12$ the function
\[
 \t \mapsto (1 - \h_0(\t))    \pppg x ^{-\d} \t D^\a  R^{(n)}(\t +i\m) v
\]
 belongs to $L^1(\R,L^2)$. Moreover the same applies with $\t D^\a$ replaced by $\t^2$.
\end{proposition}

\begin{proof}
Let $\t \in \supp(1-\h_0)$ and $z = \t + i\m$. Using three times the equality
\[
R(z) = \big(1+z^2\big)\inv  \big( R(z) (\Hz+1) -1 \big)
\]
we can write
\[
R(z) = - \big(1+z^2\big)\inv - \big(1+z^2\big)^{-2} (\Hz + 1)   - \big(1+z^2\big)^{-3} (\Hz + 1)^2   + \big(1+z^2\big)^{-3} R(z) (\Hz + 1)^3 
\]
on $\Sc(\R^d)$. Using Theorems \ref{th-inter-freq} and \ref{th-high-freq} we can check that there exists $C \geq 0$ (which depends on $v$) such that for all $\t \in \supp(1-\h_0)$ we have
\[
 \nr{ \pppg x ^{-\d}  R^{(n)}(\t +i\m) v}_{L^2} \leq C \pppg \t ^{- 4}
\]
and 
\[
 \nr{ \pppg x ^{-\d} D^\a  R^{(n)}(\t +i\m) v}_{L^2} \leq C \pppg \t ^{- 3},
\]
which proves the proposition.
\end{proof}

Let $(u_0,u_1) \in \Sc(\R^d)$ and $w : z \mapsto au_0 - i zu_0 + u_1 \in \Sc(\R^d)$. By the integrability around 0 given by Theorem \ref{th-low-freq} we can now take the inverse Fourier transform of \eqref{fourier-um} for $n \in \Ii 2 {d-1}$ and $\d > n + \frac 12$: for all $t \geq 0$ we have 
\[
(it)^n \pppg x^{-\d} (\nabla u_\m , \tilde u_\m) (t) = \frac 1 {2\pi} \int_\R e^{-it\t} \pppg x^{-\d}   \frac {d^n}{d\t^n} \big(\big(\nabla , -iz \big) R(z) w(z) \big) \, d\t,
\]
where, here and below, $z$ stands for $\t + i\m$. Note that for $\tilde u_\m$ (see \eqref{Four-tildeu}) we used that $\frac {d^n}{d\t^n} u_0 = 0$ since $n\neq 0$.\\

We consider $\h_0 \in C_0^\infty(\R,[0,1])$ and $\h \in C^\infty(\R^*,[0,1])$  two even functions such that $\h_0 + \sum_{j=1}^\infty \h_j = 1$ on $\R_+$, where for $j \in \N^*$ and $\t \in \R_+$ we have set $\h_j(\t) = \h\big(\frac \t  {2^{j-1}}\big)$. In particular $\h_0$ is equal to 1 in a neighborhood of 0.

Using the partition of unity $\h_0(\t) + (1-\h_0)(\t) = 1$, we split the integral into two terms. After an additional partial integration in the second term we obtain 
\begin{equation} \label{decomp-time-decay}
(it)^n \pppg x^{-\d} (\nabla  u_\m , \tilde u_\m) (t) = (v_{0,\m}, \tilde v_{0,\m})(t) + \frac {1}{it} (v_{\finter,\m}, \tilde v_{\finter,\m})(t) + \frac {1}{it} \sum_{j=1}^\infty (v_{j,\m}, \tilde v_{j,\m})(t),
\end{equation}
where
\begin{equation} \label{def-vom}
(v_{0,\m}, \tilde v_{0,\m})(t) = \frac 1 {2\pi} \int_\R \h_0(\t)  e^{-it\t} \pppg x^{-\d}   \frac {d^n}{d\t^n} \big(\big(\nabla , -iz\big) R(z) w(z) \big) \, d\t,
\end{equation}
\[
(v_{\finter,\m}, \tilde v_{\finter,\m})(t) = - \frac 1 {2\pi} \int_\R \h'_0(\t)  e^{-it\t} \pppg x^{-\d}   \frac {d^{n}}{d\t^{n}} \big(\big(\nabla , -iz\big) R(z) w(z) \big) \, d\t
\]
and for $j \geq 1$:
\[
(v_{j,\m}, \tilde v_{j,\m})(t) =  \frac 1 {2\pi} \int_\R \h_j(\t)  e^{-it\t} \pppg x^{-\d}   \frac {d^{n+1}}{d\t^{n+1}} \big( \big(\nabla, -iz\big) R(z) w(z) \big) \, d\t.
\]

According to Theorem \ref{th-low-freq}, the derivatives of $\nabla  R (\t +i\m)$ and $\t  R (\t +i\m)$ are uniformly (in $\m>0$) integrable (in $\t$) around 0 up to order $d-1$ in suitable weighted spaces. By \eqref{decomp-time-decay} this will lead to a $t^{1-d}$ decay rate for $\nabla u_\m$ and $\tilde u_\m$. Since we cannot perform one more partial integration in \eqref{def-vom} (the derivative of the resolvent becomes too singular, see Theorem \ref{th-low-freq}) we cannot clearly get a $t^{-d}$ decay. In order to get an ``almost'' $t^{-d}$ decay, we use the following lemma:



\begin{lemma} \label{lem-holder2}
Let $\Kc$ be a Hilbert space. Let $f \in C^1 (\R^*, \Kc)$ be equal to 0 outside a compact subset of $\R$, and assume that for some $\g \in ]0,1/2[$ and $M_f \geq 0$ we have
\[
\forall \t \in \R^*, \quad  \nr{f(\t)}_\Kc \leq M_f \abs \t^{-\g} \quad \text{and} \quad \nr{f'(\t)}_\Kc \leq M_f \abs \t ^{-1 -\g}.
\]
Then there exists $C \geq 0$ which does not depend on $f$ and such that for all $t \in \R$ we have
\[
\nr{\hat f (t)}_\Kc \leq C \, M_f \, \pppg t ^{-1 + 2\g}.
\]
\end{lemma}

\begin{proof}
We first remark that $f \in L^1 (\R,\Kc)$ and hence $\hat f \in L^\infty(\R,\Kc)$ (with a bound which only depends on $M_f$). The difficulty thus comes from large values of $t$. If we set $\b = \frac {1-2\g} {1-\g} \in ]0,1[$ then for all $t \in \R$ we have
\begin{align*}
\nr{\hat f (t)} 
\lesssim M_f \pppg t ^{-1+2\g} + \nr{\int_{\abs \t  >  t^{-\b}} e^{-i\t t} f(\t ) \, d\t }.
\end{align*}
Let $\vf \in C_0^\infty(]-1,1[,[0,1])$ be such that $\int_\R \vf = 1$. For $t \geq 2$ and $\abs \t > t\inv$ we set 
\[
f_t (\t ) = \int_{-1}^1 f\left( \t  - \frac u { t} \right) \vf(u)\, du =  t \int_{\t - \frac 1t}^{\t + \frac 1t} f(y) \vf \big( t (\t -y) \big) \, dy.
\]
According to the mean value inequality we have
\begin{align*}
\nr{\int_{\abs \t  \geq  t ^{-\b}} e^{-i\t t} \big( f(\t ) - f_t(\t ) \big) \, d\t  }
& \leq \int_{\abs \t  \geq  t ^{-\b}} \int_{-1}^1 \nr{f(\t) - f\left( \t - \frac {u}{ t} \right) } \vf(u)\,du  \, d\t \\
& \lesssim   M_f \int_{\abs \t  \geq  t ^{-\b}} \int_{-1}^1 \left(\abs{\t} - \frac 1 t\right)^{-1-\g} \frac {\abs u} { t} \vf(u) \,du  \, d\t \\
& \lesssim   M_f \, t^{\g \b - 1} .
\end{align*}
On the other hand, since
\[
 \nr{f_t\big(  t^{-\b} \big)} \leq \int_{\abs u \leq 1} \nr{ f\left(  t^{-\b}- \frac u  t  \right) } \vf (u)\, du \lesssim  t^{\b \g}
\]
(the same estimate holds for $f_t( - t^{-\b})$), we obtain by partial integration and using $\int \phi^{\prime}(u) du = 0 $,
\begin{align*}
\nr{\int_{\abs \t  \geq  t ^{-\b}} e^{-i\t t} f_t(\t )  d\t}
& \lesssim  M_f \,t^{\g -1} +   t\inv \nr{\int_{\abs \t  \geq  t ^{-\b}} e^{-i\t t} f'_t(\t )  \, d\t }\\
& \lesssim  M_f \,t^{\g -1} + t  \int_{\abs \t  \geq  t ^{-\b}} \nr{ \int_{\t - \frac 1t}^{\t + \frac 1t}  f(y) \vf'\big(t  (\t -y)\big) \, dy\, } d\t \\
& \lesssim  M_f \,t^{\g -1} +  \int_{\abs \t  \geq  t ^{-\b}}\int_{-1}^1 \nr{f\left( \t  - \frac {u}{ t} \right) - f(\t )} \abs{\vf'(u)} \, du \, d \t \\
& \lesssim   M_f \,t^{\g-1}.
\end{align*}
This concludes the proof.
\end{proof}

We can now estimate the contribution of low frequencies:

\begin{proposition} \label{prop-low-freq2}
Let $\e > 0$ and $n= d-1$. Then there exists $C \geq 0$ which does not depend on $(u_0,u_1) \in \Sc(\R^d)\times\Sc(\R^d)$ and such that for $t \geq 0$ and $\m > 0$ we have 
\begin{equation*}
\nr { v_{0,\m}(t)}_{L^2} +\nr { \tilde v_{0,\m}(t)}_{L^2} \leq C \pppg {t}^{-(1-\e)} \nr {(u_0,u_1)}_{L^{2,\d} \times L^{2,\d}}.
\end{equation*}
\end{proposition}

Recall that there is a factor $ t^{d-1} $ in the left hand side of \eqref{decomp-time-decay} (applied with $n=d-1$) so the final the contribution of $v_{0,\m}(t)$ and $\tilde v_{0,\m}(t)$ in the estimate of Theorem \ref{th-loc-decay} will clearly follow from Proposition \ref{prop-low-freq2}.

\begin{proof}
The key observation is that we have either a gradient for $v_{0,\m}(t)$ or an extra factor $z = \t+i\m$ for $\tilde v_{0,\m}(t)$, so that both can be estimated by mean of Theorem \ref{th-low-freq}, with the same bound as in the second estimate of this theorem. The result is then a direct consequence of  Lemma \ref{lem-holder2}.
\end{proof}

In a similar (and even simpler) fashion, the pair $(v_{\finter,\m}(t), \tilde v_{\finter,\m}(t))$ is estimated thanks to Theorem \ref{th-inter-freq}. Notice that the resolvent can be differentiated as much as we wish, provided that we consider the appropriate weighted spaces, so the decay in time is actually as fast as we wish. We record  the estimate we need without proof in the following proposition.

\begin{proposition} Let $n$ be chosen arbitrarily and $\d > n + \frac 12$. There exists $C \geq 0$ which does not depend on $(u_0,u_1) \in \Sc(\R^d)\times\Sc(\R^d)$ and such that for $t \geq 0$ and $\m > 0$ we have 
\begin{equation*}
\nr { v_{\finter,\m}(t)}_{L^2} +\nr { \tilde v_{\finter,\m}(t)}_{L^2} \leq C \nr {(u_0,u_1)}_{L^{2,\d} \times L^{2,\d}}.
\end{equation*}
\end{proposition}

In the rest of this section we estimate the contribution of high frequencies, \ie of the terms $v_{j,\m}$ and $\tilde v_{j,\m}$, $j\in\N^*$.\\

Let $\tilde \h \in C_0^\infty(\R^*,[0,1])$ be equal to 1 on a neighborhood of $\supp \h$. For $\t \in \R$ and $j \geq 1$ we set $\tilde \h_j(\t) = \tilde \h \big(\frac \t  {2^{j-1}}\big)$.
Let $n\geq 2$, $\d > n + \frac 12$, $s \in \{0,1\}$, $\a \in \N^n$ with $\abs \a \in \{0,1\}$ and $v \in \Sc(\R^d)$. For $j \in \N^*$, $\m > 0$ and $t\geq 0$ we set 
\[
I_{j,\m}(t) =  \int_{\R} e^{-it\t}  \t^{s+1-\abs \a}\h_j(\t)    \pppg x^{-\d}  D^\a   R^{(n)}(\t+i\m)  \, d\t \in \Lc(L^2)
\]
and then:
\[
u_{1,j,\m} (t) = \tilde\h_j \big( \Ho^{1/2} \big) I_{j,\m}(t) \pppg x ^{-\d}  \tilde \h_j \big( \Ho^{1/2} \big) \pppg x^{\d} v,
\]
\[
u_{2,j,\m} (t) = \tilde \h_j \big( \Ho^{1/2} \big) I_{j,\m}(t) \pppg x ^{-\d}  \left( 1 - \tilde \h_j \big( \Ho^{1/2} \big)\right) \pppg x^{\d} v
\]
and
\[
u_{3,j,\m} (t) = \left( 1 - \tilde \h_j \big( \Ho^{1/2} \big)\right) I_{j,\m}(t)  v .
\]
Both $v_{j,\m}$ and $\tilde v_{j,\m}$ are linear combinations of terms of the form $u_{1,j,\m} (t) + u_{2,j,\m} (t) +u_{3,j,\m} (t)$ for some $s \in \{0,1\}$, $\a \in \N^n$ with $\abs \a \in \{0,1\}$ and $v \in \{ u_0,u_1\}$, so we have to prove that there exists $C\geq 0$ such that for all $t\geq 0$ and $\m>0$ we have
\begin{equation} \label{reduc-high-freq}
\nr{ \sum_{j=1}^\infty \big( u_{1,j,\m} (t) + u_{2,j,\m} (t) +u_{3,j,\m} (t) \big)}_{L^2} \leq C \nr{v}_{H^{1+s,\d}} .
\end{equation}

%
%
%
%
%
%
%
%
%
%
%

\begin{proposition} \label{prop-Uc1}
There exists $C \geq 0$ which does not depend on $v \in \Sc(\R^d)$ and such that for all $t \geq 0$ and $\m > 0$ we have
\[
\nr{ \sum_{j =1}^\infty  u_{1,j,\m}(t)}_{L^2} \leq C \nr {\pppg x ^\d v}_{H^{1+s}}.
\]
\end{proposition}

We record that the loss of a derivative in Theorem \ref{th-loc-decay} is due to this proposition.

\begin{proof}
There exists $N \in \N$ such that
\begin{equation*} 
\forall j,k \in \N^*, \quad \abs{j-k} \geq N \quad \implies \quad  \supp (\tilde \h_j) \cap \supp (\tilde \h_k) = \emptyset,
\end{equation*}
and hence
\[
\nr{ \sum_{j =1}^\infty u_{1,j,\m}(t) }^2_{L^2}\leq  N \sum_{j =1}^\infty\nr{u_{1,j,\m}(t) }^2_{L^2}.
\]
The right-hand side is bounded by the product of
\[
N \sup_{j \in \N^*} 2^{-2(s+1)j}  \left(\int_\R \t^{s+1-\abs\a} \h_j(\t) \nr{ \pppg x ^{-\d} D^\a R^{(n)}(\t+i\m) \pppg {x}^{-\d} }_{\Lc(L^2)} d\t \right)^2 ,
\]
(which is independent of $t$ and uniformly bounded in $\m$ by Theorem \ref{th-high-freq}), and 
\[
 \sum_{j =1}^\infty 2^{2(s+1)j} \nr{  \tilde \h_j \big( \Ho^{1/2} \big) \pppg x^{\d}v}^2_{L^2}.
\]
The latter is controlled by $\nr{ \pppg \Ho^{\frac {s+1}2} \pppg x^{\d}v}^2_{L^2}$, again by almost orthogonality.
%
\end{proof}

\begin{proposition} \label{prop-Uc23}
There exists $C \geq 0$ which does not depend on $v \in \Sc(\R^d)$ and such that for all $j \in \N^*$, $t \geq 0$ and $\m \in ]0,1]$ we have
\[
\sum_{j =1}^\infty \nr{u_{2,j,\m}(t)}_{L^2} + \sum_{j =1}^\infty \nr{u_{3,j,\m}(t)}_{L^2} \leq C \nr{\pppg x ^\d v}_{L^2}.
\]
\end{proposition}

\begin{proof}
\stepp
We only prove the estimate for $u_{3,j,\m}(t)$. The estimate for $u_{2,j,\m}(t)$ can be proved similarly. For this we prove by induction on $n \in \N$ the following statement: for $m \geq n$, $\d > m + \frac 12$ and $\p \in C_0^\infty(\R^*,[0,1])$ equal to 1 on a neighborhood of $\supp \h$ there exists $C \geq 0$ such that for $j \in \N^*$, $\t \in \supp(\h_j)$ and $\m \in ]0,1]$ we have
\begin{equation} \label{estim-u23-rec}
\t ^{s+1-\abs\a} \nr{\left( 1 - \p_j \big( \Ho^{1/2} \big)\right)   \pppg x ^{-\d} D^\a R^{(m)}(\t+i\m) \pppg x^{-\d}}_{\Lc(L^2)} \leq C \, 2^{(s-1-n) j}
\end{equation}
(where $\p_j (\th) = \p(\th/2^{j-1})$ for all $\th \in \R$ and $j \in \N^*$). Applied with $m=n$ and $\p = \tilde \h$, and after integration on $\t$, this gives 
\[
\nr{u_{3,j,\m}(t)} \lesssim 2^{(s-n)j} \nr{\pppg x^\d v}_{L^2},
\]
from which we get the estimate on $ \sum_{j} ||u_{3,j,\mu}(t) ||_{L^2} $ since $n \geq 2$. Note that since $\t \in \supp(\h_j)$, \eqref{estim-u23-rec} can be rewritten as
\[
\nr{\left( 1 - \p_j \big( \Ho^{1/2} \big)\right)   \pppg x ^{-\d} D^\a R^{(m)}(\t+i\m) \pppg x^{-\d}}_{\Lc(L^2)} \lesssim \, 2^{(\abs \a - 2 - n)j}.
\]


\stepp
Let $n \in \N$ and $m \geq n$. If $n>0$ we assume that \eqref{estim-u23-rec} holds up to order $n-1$.
Here we can use pseudo-differential calculus on $\p_j \big( \Ho^{1/2} \big)$ (see \cite{dimassis} or Proposition 2.1 in \cite{bouclett08} in a closer context). Indeed 
%
$\p$ vanishes on a neighborhood of 0 and hence $\p_j \big( \Ho^{1/2} \big)$ can be rewritten as $\tilde \p  (h^2 H_0)$ where $\tilde \p : \l \mapsto \p(\sqrt \l)$ belongs to $C_0^\infty(\R)$ and $h = 2^{1-j}$.
Let $\vf \in C_0^\infty(\R^*,[0,1])$ be equal to 1 on a neighborhood of $\supp \h$ and such that $\p = 1$ on a neighborhood of $\supp \vf$. As before we set  $\vf_j (\th) = \vf(\th/2^{j-1})$ for $\th \in \R$ and $j \in \N^*$.
Using Theorem \ref{th-high-freq}, we see that for $j \in \N^*$, $\t \in \supp(\h_j)$ and $\m \in ]0,1]$ we have
\begin{eqnarray*}
\lefteqn{ \nr{\left( 1 - \p_j \big( \Ho^{1/2} \big)\right) \pppg x ^{-\d} D^\a \vf_j\big(\Ho^{1/2}\big) R^{(m)}(\t + i\m)  \pppg x ^{-\d} }}\\
&& \leq \nr{\left( 1 - \p_j \big( \Ho^{1/2} \big)\right)  \pppg x ^{-\d} D^\a \vf_j\big(\Ho^{1/2}\big) \pppg x ^\d } \nr {\pppg x ^{-\d}  R^{(m)}(\t + i\m)  \pppg x ^{-\d} }\\
&& \lesssim 2^{(\abs \a-2-n)j}.
\end{eqnarray*}
Actually the decay could be as fast as we wish since we use the fact that the symbols of $  \left( 1 - \p_j \big( \Ho^{1/2} \big)\right) $ and $ \vf_j\big(\Ho^{1/2}\big)$ have disjoint supports. Here we have not used the inductive assumption. Since $\left( 1 - \p_j \big( \Ho^{1/2} \big)\right)$ is a bounded operator on $L^2$ uniformly in $j$, it only remains to estimate
\[
 \pppg x ^{-\d}D^\a \left( 1 - \vf_j\big(\Ho^{1/2}\big) \right) R^{(m)}(\t + i\m)  \pppg x ^{-\d}.
\]

\stepp
For $z \in \C_+$ we have
\[
R(z) = R_0(z) +iz R_0(z)  a  R(z), \qquad \text{where } R_0(z) = \big(\Ho -z^2\big)\inv.
\]
Differentiating $m$ times we get
\begin{equation*} 
R^{(m)}(z) =  R_0^{(m)}(z) + i \sum _{k=0}^{m} C_{m}^k (zR_0)^{(k)}(z) a R^{(m-k)}(z),
\end{equation*}
where $zR_0$ is the function $z \mapsto z R_0(z)$. For all $k \in \N$ the derivative $R_0^{(k)}(z)$ can be written as a linear combination of terms of the form $z^{\s_1} R_0(z)^{\s_2}$ where $2\s_2 = \s_1 + k +2$.
%
Using the Spectral Theorem and the fact that $1-\vf$ and $\h$ have disjoint supports, this implies that for $\t \in \supp \h_j$, $\m \in ]0,1]$ and $j \in \N^*$ we have
\[
 \nr{ (1 - \vf_j)\big(\Ho^{1/2}\big)R_0^{(m)}(\t + i\m)  } \lesssim 2^{-(m+2)j} .
\]
In particular we have 
\[
\nr{\pppg x ^{-\d} D^\a \left( 1 - \vf_j\big(\Ho^{1/2}\big) \right) R_0^{(m)}(\t + i\m)  \pppg x ^{-\d} }  \lesssim 2^{(\abs \a - m-2)j}.
\]
Here and below the contribution of $D^\a$ and the uniform control with respect to $z = \t + i\m$ follow from the fact that 
$
\big( 1 - \vf_j(H_0^{1/2}) \big) R_0^{(m)} (z)
$
is a linear combination of terms of the form 
\[
(hz)^{\s_1} \big( 1 - \vf(h H_0^{1/2}) \big) (h^2 H_0 - (hz)^2)^{-\s_2}  h ^{m+2},
\]
where $h = 2^{-(j-1)}$, $2\s_2 = \s_1 + m +2$ and $\Re (hz)$ belongs to a compact disjoint from the support of $1-\vf$.
Now let $k \in \Ii 0 m$. Let $\s \in C_0^\infty(\R^*,[0,1])$ be equal to 1 on a neighborhood of $\supp \h$ and such that $\vf = 1$ on a neighborhood of $\supp \s$, and $\s_j (\th) = \s(\th/2^{j-1})$ for $\th \in \R$ and $j \in \N^*$. For $\t \in \supp \h_j$, $\m \in ]0,1]$, $z = \t + i \m$ and $j \in \N^*$ we have 
\begin{eqnarray*}
\lefteqn{  \nr{\pppg x ^{-\d}D^\a  \left( 1 - \vf_j\big(\Ho^{1/2}\big) \right) (zR_0)^{(k)}(z) a R^{(m-k)}(z)  \pppg x ^{-\d}}}\\
&& \lesssim  \nr{\pppg x ^{-\d}D^\a  \left( 1 - \vf_j\big(\Ho^{1/2}\big) \right) (zR_0)^{(k)}(z) \pppg x ^\d a\s_j\big(\Ho^{1/2}\big)
} \nr{\pppg x ^{-\d}  R^{(m-k)}(z)  \pppg x ^{-\d}}\\
&& \,+   \nr{\pppg x ^{-\d} D^\a \left( 1 - \vf_j\big(\Ho^{1/2}\big) \right) (zR_0)^{(k)}(z) \pppg x ^\d a} \nr{ \left( 1 - \s_j\big(\Ho^{1/2}\big)\right) \pppg x ^{-\d}  R^{(m-k)}(z)  \pppg x ^{-\d}}\\
&& \lesssim 2^{(\a - 2 - n)j} +  2 ^{ (\a - k - 1)j} 2^{(k-n-1)j} \lesssim 2^{(\a - 2 -n)j}.
\end{eqnarray*}
The first term and the first norm of the second term are estimated as above by the pseudo-differential calculus and Theorem \ref{th-high-freq}. 
The estimate of the second norm comes from Theorem \ref{th-high-freq} if $n \leq k$, and from \eqref{estim-u23-rec} applied with $n-1-k$ and $\a = 0$ otherwise (we observe that if $n=0$ the estimate on this norm uses only Theorem \ref{th-high-freq} and no inductive assumption). This gives \eqref{estim-u23-rec} by induction and hence concludes the proof of the proposition. 
\end{proof}


%

According to \eqref{reduc-high-freq}, Propositions \ref{prop-low-freq2}, \ref{prop-Uc1} and \ref{prop-Uc23} prove Theorem \ref{th-loc-decay}. Now the rest of the paper is devoted to the proofs of Theorems \ref{th-inter-freq}, \ref{th-low-freq}, \ref{th-low-freq-bis} and \ref{th-high-freq}.

\section{Resolvent estimates for an abstract dissipative operator} \label{sec-Mourre}

\subsection{Multiple commutators method in the dissipative setting}

In this paragraph we first recall Mourre's commutators method in an abstract dissipative context, and then we derive uniform estimates for the powers of the resolvent, generalizing the results of \cite{jensenmp84,jensen85}. We will show in Paragraph \ref{sec-insert-factors} how to deal with the inserted factors mentioned after Theorem \ref{th-inter-freq}.\\ 

Let $\Hc$ be a Hilbert space. We consider on $\Hc$ a family of abstract operators $(H_\l)_{\l \in \L}$ (parametrized by any set $\L$) of the form $\Hl = \Hul - iV_\l$ where the operators $\Hul$ and $V_\l$ are self-adjoint and the domain $\Dom_H$ of $\Hul$ is independant of $\l$. Moreover $V_\l$ is non-negative and uniformly $\Hul$-bounded with bound less than 1: there exist $a \in [0,1[$ and $b \geq 0$ such that
\[
 \forall \l \in \L, \forall \f \in \Dom_H, \quad \nr{V_\l \f} _\Hc \leq a \nr{\Hul \f}_\Hc + b \nr \f _\Hc.
\]
In particular for all $\l \in \L$ the operator $\Hl$ is maximal dissipative, and its adjoint is $\Hul +iV_\l$ with domain $\Dom(\Hl^*) = \Dom(\Hl) =  \Dom_H$ (see Lemma 2.1 in \cite{art-mourre}).\\

Let $J$ be an open interval of $\R$. Let $(A_\l)_{\l \in \L}$ be a family of self-adjoint operators on $\Hc$ and $(\a_\l)_{\l \in \L} \in ]0,1]^\L$.

\begin{definition} \label{defconjunif}
The operators $(A_\l)_{\l \in \L}$ are said to be uniformly conjugate to $(\Hl)_{\l \in \L}$ on $J$ with lower bounds $(\a_\l)_{\l \in \L}$ if there exists a bounded family of non-negative numbers $(\b_\l)_{\l \in \L}$ such that
\begin{enumerate}[($a$)]
\item 
The domain of $A_\l$ does not depend on $\l \in \L$ (it is denoted by $\Dom_A$), and $\Dom_H \cap \Dom_A$ is dense in $\Dom_H$ endowed with the graph norm $\f \mapsto \nr{(\Hul -i) \f}$ for all $\l \in \L$.
\item 
For all $t \in \R$ and $\l \in \L$ we have $e^{itA_\l}\Dom_H \subset \Dom_H$, and 
\begin{equation*}
\forall \l \in \L, \forall \f \in \Dom_H,\quad  \sup_{\abs t  \leq 1} \nr{\Hul   e^{it A_\l} \f}<\infty.
\end{equation*}
\item 
For all $\l \in \L$ the quadratic forms $\ad_{iA_\l}(\Hul)$ and $\ad_{iA_\l}(V_\l)$ defined on $\Dom_H\cap \Dom_A$ are bounded from below and closable. If we still denote by $\ad_{iA_\l}(\Hul)$ and $\ad_{iA_\l}(V_\l)$ the associated self-adjoint operators, then $\Dom_H \subset \Dom\big(\ad_{iA_\l}(\Hul)\big) \cap\Dom\big(\ad_{iA_\l}(V_\l))$ and there exists $c \geq 0$ such that for $\l \in \L$ and $\f,\p \in\Dom_H$ we have
\begin{equation*}
\nr{\ad_{iA_\l}(\Hul) \f} +   \nr{\ad_{iA_\l}(V_\l) \f}  \leq c \sqrt {\a_\l} \nr{(\Hul-i)\f}
\end{equation*}
and 
\[
\b_\l \nr{V_\l \f} \nr{\ad_{iA_\l}(\Hul) \p} \leq c \, \a_\l \nr{(\Hul-i) \f} \nr{(\Hul -i) \p}.
\]

\item
There exists $c\geq 0$ such that for $\l\in\L$ and $\f,\p \in \Dom_H \cap \Dom_A$ we have
\[
\abs{\innp{\ad_{iA_\l}(\Hul) \f }{A_\l\p} - \innp{A_\l\f}{\ad_{iA_\l}(\Hul) \p}} \leq c \, \a _\l \nr{(\Hul-i) \f} \nr{(\Hul-i)  \p }.
\]
Moreover we have similar estimates if we replace $\ad_{iA_\l}(\Hul)$ by $\b_\l V_\l$ or $\ad_{iA_\l}(V_\l)$.
\item
For all $\l \in \L$ we have
\begin{equation} \label{hyp-mourre}
\1 J (\Hul)\big(\ad_{iA_\l}(\Hul) + \b_\l V_\l \big) \1 J (\Hul)  \geq \a_\l \1 J(\Hul).
\end{equation}
\end{enumerate}
\end{definition}

Up to the parameter dependence, these conditions are the standard ones when there is no dissipative perturbation (see \cite{mourre81}). 
In the applications, introducing a parameter will be particularly useful to handle the low and high frequency regimes. At low frequency, $ \lambda $ will correspond to the rescaling factor $ |z|  \ll 1 $, $ A_\l = A$ to the standard generator of dilations and $ H_{\lambda} $ to the rescaled operator $\Ptau$ (see \eqref{def-Ptau}). At high frequency, $ \lambda $ will be the semiclassical parameter $h$, $H_{\lambda} = \hhs$ (see \eqref{def-Hhs}) and $ A_h $ will be a modification of the conjugate operator of G\'erard-Martinez (to handle the possible bounded trajectories), see \eqref{op-Fh}. We refer a reader non familiar with the Mourre Theory to the Appendix \ref{AppB} where we consider a simple example with $ H_1 = - \Delta $.\\

In \cite{art-mourre}, the following extension to Mourre's result has been proved:

\begin{theorem} \label{th-mourre}
Suppose $(A_\l)_{\l\in\L}$ is uniformly conjugate to $(\Hl)_{\l\in\L}$ on $J$ with bounds $(\a_\l)_{\l \in \L}$. Let $\d \in \big] \frac 12 , 1]$ and let $I$ be a closed subinterval of $J$.
\begin{enumerate}[(i)]
\item There exists $c \geq 0$ such that for $\l \in \L$ and $z \in \C_{I,+}$ we have
\[
\nr{\pppg {A_\l}^{-\d} (\Hl -z)\inv \pppg {A_\l}^{-\d}}_{\Lc(\Hc)} \leq \frac c {\a_\l}.
\]
\item There exists $c \geq 0$ such that for $\l \in \L$ and $z,z' \in \C_{I,+}$ we have
\[
\nr{\pppg {A_\l}^{-\d} \left((\Hl -z)\inv - (\Hl -z')\inv \right) \pppg {A_\l}^{-\d}}_{\Lc(\Hc)} \leq c {\a_\l}^{-\frac {4\d}{2\d+1}} \abs{z-z'}^{\frac {2\d-1}{2\d+1}}.
\]
\item For any $\l \in \L$ and $\t \in I$ the limit
\[
\pppg {A_\l}^{-\d} (\Hl -(\t+i0))\inv \pppg {A_\l}^{-\d} = \lim_{\m \to 0^+} \pppg {A_\l}^{-\d} (\Hl -(\t+i\m))\inv \pppg {A_\l}^{-\d}
\]
exists in $\Lc(\Hc)$ and define a H\"older-continuous function of $\t$.
\end{enumerate}

\end{theorem}

We now extend this result to multiple resolvent estimates, following the ideas of \cite{jensenmp84} and \cite{jensen85}. To this purpose we need more commutators estimates. \\

Let us assume that $(A_\l)_{\l \in \L}$ is uniformly conjugate to $(\Hl)_{\l \in \L}$ on $J$ for lower bounds $(\a_\l)_{\l \in \L}$.
For $n \geq 2$, as long as the forms $[\ad_{iA_\l}^{n-1}(\Hul),iA_\l]$ and $[\ad_{iA_\l}^{n-1}(V_\l),iA_\l]$ are closable and semi-bounded we consider the corresponding operators denoted by $\ad_{iA_\l}^{n}(\Hul)$ and $\ad_{iA_\l}^{n}(V_\l)$. Then we set $\ad_{iA_\l}^{n}(\Hl) = \ad_{iA_\l}^{n}(\Hul) -i \ad_{iA_\l}^{n}(V_\l)$.
In general, defining such iterated commutators and their closures has to be handle with caution (see e.g. \cite{jensenmp84}). In our applications, however, we will only consider (pseudo)differential operators on $ \R^n $ for which there will be no problem.

Now we can define $N$-smoothness as in the self-adjoint case, taking into account the commutators of the dissipative part with the conjugate operator:

\begin{definition} \label{defconjunifdissn} 
Let $N \geq 2$. Then $H_\l$ is said to be uniformly $N$-smooth with respect to $A_\l$ if in addition to the assumptions of Definition \ref{defconjunif} the following conditions are satisfied:
\begin{enumerate} [($a_n$)]
\setcounter{enumi}{2}
\item 
For all $\l \in \L$ and $n \in \Ii 2 N$ the operators  $\ad_{iA_\l}^{n}(\Hul)$ and $\ad_{iA_\l}^{n}(V_\l)$ are well-defined and their domains contain $\Dom_H$.
Moreover there exists $c \geq 0$ such that
\[
\forall n \in \Ii 2 N, \forall \l \in \L, \forall \f \in \Dom_H, \quad \nr{\ad_{iA_\l}^{n}(\Hl) \f} \leq  c \, \a_\l   \nr {(\Hul-i) \f}.
\]
\item 
There exists $c \geq 0$ such that for all $\l \in \L$ and $\f,\p \in \Dom_H \cap \Dom_A$ we have
\[
\abs{\innp{\ad_{iA_\l}^{N}(\Hl) \f }{A_\l  \p} - \innp{A_\l \f}{\ad_{iA_\l}^{N}(\Hl) \p}} \leq c \, \a_\l \nr {(\Hul-i) \f} \nr {(\Hul-i) \p}.
\]
\end{enumerate}
Moreover we say that $\Hl$ is uniformly $\infty$-smooth with respect to $A_\l$ if it is $N$-smooth for all $N \geq 2$.
\end{definition}

Under such a smoothness assumption, we can prove uniform estimates for the powers of the resolvent of $\Hl$. To this purpose we follow the ideas of \cite{jensen85}, to which we refer for the proofs. The strategy consists of reducing estimates for the powers of the resolvent to estimates of the resolvent involving spectral projections of the conjugate operator $A_\l$ on $\R_+$ and $\R_-$. The following abstract statement summarizes and generalizes Lemmas 2.1 and 2.2 of \cite{jensen85}:

\begin{lemma} \label{lemjen2.1}
Let $Q$, $P_+$ and $P_-$ be bounded operators on $\Hc$ such that $P_+ + P_- = 1$ and $Q$ is self-adjoint, bounded, with $Q\inv \geq 1$. Let $l\in\N^*$ and let $R_1,\dots,R_l$ be bounded operators on $\Hc$. Suppose that there exist $\s \in ]0,1]$, $N>\s$ and $c_1,\dots,c_l > 0$ such that for any $j\in\Ii 1 l$ we have
\begin{enumerate}
\item[$(a_1)$] $\nr{Q^\s R_j Q ^\s} \leq c_j$,
\item[$(b_1)$] $\nr{Q^{1-\d} P_- R_j Q^\d} \leq c_j$ for all $\d \in [\s,N[$,
\item[$(c_1)$] $\nr{Q^\d R_j P_+ Q^{1-\d}} \leq c_j$ for all $\d \in [\s,N[$,
\item[$(d_1)$] $\nr{Q^{-\d_1} P_- R_j P_+ Q^{-\d_2}} \leq c_j$ for all $\d_1,\d_2 \geq 0$ such that $\d_1 + \d_2 < N -1$.
\end{enumerate}
Then for all $n \in \N^*$ such that $n < N + 1 - \s$ and $j_1,\dots,j_n \in \Ii 1 l$ we have
\begin{enumerate}
\item[$(a_n)$] $\nr{Q^{n-1+\s} R_{j_1}\dots R_{j_n} Q ^{n-1+\s}} \leq 2^{n-1} c_{j_1}\dots c_{j_n}$,
\item[$(b_n)$] $\nr{Q^{n-\d} P_- R_{j_1}\dots R_{j_n} Q^\d} \leq 2^{n-1} c_{j_1}\dots c_{j_n}$ for all $\d \in [n-1+\s,N[$,
\item[$(c_n)$] $\nr{Q^\d R_{j_1}\dots R_{j_n} P_+ Q^{n-\d}} \leq 2^{n-1} c_{j_1}\dots c_{j_n}$ for all $\d \in [n-1+\s,N[$,
\item[$(d_n)$] $\nr{Q^{-\d_1} P_- R_{j_1}\dots R_{j_n} P_+ Q^{-\d_2}} \leq 2^{n-1} c_{j_1}\dots c_{j_n}$ for all $\d_1,\d_2 \geq 0$ such that $\d_1 + \d_2 < N - n$.
\end{enumerate}
\end{lemma}

\begin{proof}
This lemma can be proved as the original one by induction on $n \in \N^*$ by inserting $Q^{s_1} Q^{-s_1} P_- + P_+ Q^{-s_2}Q^{s_2} = 1$ between the the factors $R_{j_1}$ and $R_{j_2}$ or $R_{j_{n-1}}$ and $R_{j_n}$, for suitable $s_1,s_2 \in \R$. We omit the details.
\end{proof}

We want to apply this lemma with $R_j = (\Hl-z)\inv$, $Q= \pppg {A_\l}\inv$ and $P_\pm = \1{\R_\pm}(A_\l)$.
We already have ($a_1$) by Theorem \ref{th-mourre}, so in order to obtain $(a_n)$, $(b_n)$, $(c_n)$ and $(d_n)$, it remains to prove $(b_1)$, $(c_1)$ and $(d_1)$. The following proposition states that $(d_1)$ holds true:


\begin{proposition} \label{thjen3.2}
Let $N \geq 2$. Suppose $(A_\l)_{\l \in \L}$ is uniformly conjugate to $(\Hl)_{\l \in \L}$ on $J$ with lower bound $(\a_\l)_{\l \in \L}$ and $\Hl$ is $N$-smooth with respect to $A_\l$. Let $I$ be a compact subinterval of $J$ and $\d_1,\d_2 \geq 0$ such that $\d_1 + \d_2< N-1$.

\begin{enumerate}[(i)]
\item There exists $c\geq 0$ such that
\begin{equation*}
\forall \l \in \L, \forall z\in \C_{I,+},\quad  \nr{\pppg{A_\l}^{\d_1} \1{\R_-} (A_\l) (\Hl - z)\inv \1{\R_+}(A_\l) \pppg{A_\l}^{\d_2}} \leq \frac c {\a_\l}.
\end{equation*}

\item For $\t \in I$ the limit
\begin{eqnarray*}
\lefteqn{\pppg{A_\l}^{\d_1}\1{\R_-}(A_\l) (\Hl - (\t+i0))\inv \1{\R_+}(A_\l) \pppg{A_\l}^{\d_2}}\\
&& \hspace{2cm} = \lim_{\m\to 0^+} \pppg{A_\l}^{\d_1}\1{\R_-}(A_\l) (\Hl - (\t+i\m))\inv \1{\R_+} (A_\l) \pppg{A_\l}^{\d_2}
\end{eqnarray*}
exists in $\Lc(\Hc)$ and define a H\"older-continuous function of $\t$.
\end{enumerate}
\end{proposition}

The next proposition shows that $(b_1)$ and $(c_1)$ hold true as well.

\begin{proposition} \label{thjen3.5}
Let $N \geq 2$. Suppose $(A_\l)_{\l \in \L}$ is uniformly conjugate to $(\Hl)_{\l \in \L}$ on $J$ with lower bound $(\a_\l)_{\l \in \L}$ and $\Hl$ is $N$-smooth with respect to $A_\l$. Let $I$ be a compact subinterval of $J$ and $\d \in \left] \frac 12 , N \right[$.

\begin{enumerate}[(i)]
\item There exists $c \geq 0$ such that
\begin{equation} \label{thjen3.5a}
\forall \l \in \L,  \forall z\in \C_{I,+}, \quad  \nr{\pppg{A_\l}^{-\d} (\Hl - z)\inv \1{\R_+}(A_\l) \pppg{A_\l}^{\d-1}}_{\Lc(\Hc)}  \leq \frac c {\a_\l}.
\end{equation}


\item For $\t \in I$ the limit
\begin{eqnarray*}
\lefteqn{\pppg{A_\l}^{-\d} (\Hl - (\t+i0))\inv \1{\R_+}(A_\l) \pppg{A_\l}^{\d'-1}}\\
&& \hspace{2cm} = \lim_{\m\to 0^+} \pppg{A_\l}^{-\d} (\Hl - (\t+i\m))\inv  \1{\R_+}(A_\t) \pppg{A_\l}^{\d'-1}
\end{eqnarray*}
exists in $\Lc(\Hc)$ and define a H\"older-continuous function of $\t$.
\item We have similar results for the operator 
\[
 \pppg{A_\l}^{-\d} (\Hl^* - \bar z)\inv \1{\R_-}(A_\l) \pppg{A_\l}^{\d-1}
\]
and hence, taking the adjoint, for
\[
 \pppg{A_\l}^{\d-1} \1{\R_-}(A_\l)  (\Hl -  z)\inv \pppg{A_\l}^{-\d}.
\]
\end{enumerate}
\end{proposition}

To prove these two results we follow word for word the proofs of the self-adjoint analogues given in \cite{jensen85} (they are also rewritten with full details for a family of dissipative operators in \cite{these}). Now using Lemma \ref{lemjen2.1} we get uniform estimates for the powers of the resolvent, which gives regularity for the limit $(\Hl-(\t+i0))\inv$ with respect to $\t$. These conclusions are summarized in the next two theorems.

\begin{theorem} \label{thordren}
Let $N\geq 2$. Suppose $(A_\l)_{\l \in \L}$ is uniformly conjugate to $(\Hl)_{\l \in \L}$ on $J$ with lower bound $(\a_\l)_{\l \in \L}$ and $\Hl$ is $N$-smooth with respect to $A_\l$. Let $I$ be a compact subinterval of $J$ and $n \in \Ii 1 N$.
\begin{enumerate}[(i)]
\item If $\d > n-\frac 12$ there exists $c\geq 0$ such that
\begin{equation*}
\forall \l \in \L, \forall z \in \C_{I,+}, \quad \nr{ \pppg{A_\l}^{-\d} (\Hl-z)^{-n} \pppg{A_\l}^{-\d}} \leq \frac c {\a_\l^{n}}.
\end{equation*}
\item If $\d \in \left]n - \frac 12,N\right[$ there exists $c\geq 0$ such that
\begin{equation*}
\forall \l \in \L, \forall z \in \C_{I,+}, \quad   \nr{ \pppg{A_\l}^{\d-n} \1{\R_-}(A_\l)(\Hl -z)^{-n}  \pppg{A_\l}^{-\d}} \leq \frac c {\a_\l^n}.
\end{equation*}
\item If $\d \in \left] n - \frac 12 , N \right[$ there exists $c\geq 0$ such that
\begin{equation*}
\forall \l \in \L, \forall z \in \C_{I,+}, \quad   \nr{ \pppg{A_\l}^{-\d} (\Hl -z)^{-n} \1{\R_+}(A_\l) \pppg{A_\l}^{\d-n}} \leq \frac c {\a_\l^n}.
\end{equation*}
\item If $\d_1,\d_2 \geq 0$ satisfy $\d_1 + \d_2 < N - n$ there exists $c\geq 0$ such that
\begin{equation*}
\forall \l \in \L, \forall z \in \C_{I,+}, \quad \nr{ \pppg{A_\l}^{\d_1} \1{\R_-} (A_\l) (\Hl-z)^{-n} \1{\R_+} (A_\l) \pppg{A_\l}^{\d_2}} \leq \frac c {\a_\l^n}.
\end{equation*}
\item Moreover these operators have a limit in $\Lc(\Hc)$ when $\Im z \to 0^+$ if $\pppg{A_\l}^{\d-n}$ is replaced by $\pppg{A_\l}^{\d'-n}$ for some $\d' \in ]-\infty,\d[$ in (ii) and (iii). These limits define H\"older-continuous fonctions on $I$.
\end{enumerate}
\end{theorem}

\begin{theorem} \label{th-mourre-puiss}
Let $N\geq 2$. Suppose $(A_\l)_{\l \in \L}$ is uniformly conjugate to $(\Hl)_{\l \in \L}$ on $J$ with lower bound $(\a_\l)_{\l \in \L}$ and $\Hl$ is $N$-smooth with respect to $A_\l$. Let $I$ be a compact subinterval of $J$ and $n \in\Ii 1 N$.
\begin{enumerate}[(i)]
\item If $\d > n - \frac 12$ then for all $\l \in \L$ and $\t \in I$ the limit
\begin{equation*}
 \pppg{A_\l}^{-\d} (\Hl-(\t+i0))\inv \pppg{A_\l}^{-\d}
\end{equation*}
exists in $\Lc(\Hc)$ and defines a function of class $C^{n-1}$ on $I$.
\item If $\d \in \left] n - \frac 12,  N \right[$ and $\d'<\d$ then for all $\l \in \L$ and $\t \in I$ the limits
\begin{equation*}
\pppg{A_\l}^{-\d} (\Hl -(\t+i0))\inv \1{\R_+}(A_\l) \pppg{A_\l}^{\d'-n}
\end{equation*}
and
\begin{equation*}
\pppg{A_\l}^{\d'-n}    \1{\R_-}(A_\l)   (\Hl -(\t+i0))\inv \pppg{A_\l}^{-\d} 
\end{equation*}
exist in $\Lc(\Hc)$ and define functions of class $C^{n-1}$ on $I$.
\item If $\d_1,\d_2 \geq 0$ satisfy $\d_1 + \d_2 < N - n$ then for all $\l \in \L$ and $\t \in I$ the limit
\begin{equation*}
\pppg{A_\l}^{\d_1} \1{\R_-} (A_\l) (\Hl -(\t+i0))\inv \1{\R_+} (A_\l) \pppg{A_\l}^{\d_2}
\end{equation*}
exists in $\Lc(\Hc)$ and defines a function of class $C^{n-1}$ on $I$.
\end{enumerate}
\end{theorem}

\subsection{Inserted factors}  \label{sec-insert-factors}

As mentioned in the introduction, the derivatives of our resolvent $R(z)$ defined in \eqref{def-Rz} are not given by linear combinations of its powers. We have to take into account that multiplications by the absorption index $a$ are inserted between some of the factors $R(z)$. More precisely we have the following result:

\begin{proposition} \label{prop-der-R2}
For all $n \in \N^*$ the derivative $R^{(n)}(z)$ is a linear combination of terms of the form 
\[
z^k R(z) a(x)^{j_1} R(z)  a(x)^{j_2} \dots  a(x)^{j_m} R(z)
\]
where $m \in \Ii 0 n$ (there are $m+1$ factors $R(z)$), $k\in\N$, $j_1,\dots,j_m \in \{ 0,1\}$ and $n = 2m - k - (j_1+\dots +j_m)$.
\end{proposition}

\begin{proof}
This is an easy induction which we omit (use \eqref{expr-der-R}).
\end{proof}

%
%
%
%
%

%
%

Thus we cannot apply directly Theorem \ref{thordren} to estimate the derivatives of $R(z)$ uniformly in $z$. To include such a situation, we prove that a certain class of operators can be inserted between the resolvents in Theorem \ref{thordren}.

\begin{definition}
Let $A$ be an operator on the Hilbert space $\Hc$ and $N\in\N$. We say that the operator $\Phi \in \Lc(\Hc)$ belongs to $\opinsert_{N} (A)$ if the commutators $\ad^n_{A}(\Phi)$ for $n \in \Ii 0 N$, defined iteratively in the sense of forms on $\Dom(A)$, extend to bounded operators on $\Hc$. In this case we set
\[
\nr{\Phi}_{A,N} = \sum_{n=0}^N \nr{\ad^n_A(\Phi)}_{\Lc(\Hc)}.
\]
If $\Phi = (\Phi_1, \Phi_2) \in \opinsert_N(A)^2$, we also denote by $\nr{\Phi}_{A,N}$ the product $\nr{\Phi_1}_{A,N}\nr{\Phi_2}_{A,N}$.
\end{definition}

This definition implies that $\Phi \in \opinsert_N(A)$ preserves the domain of $A^n$ for all $n\in\Ii 0 N$. More generally, $\ad_{A}^m (\Phi)$ preserves the domain of $A^n$ as long as $n+m \leq N$. This remark will be implicitly used below to justify the standard algebraic computations such as \eqref{dev-Phi12}.

\begin{proposition} \label{prop-insert-prod}
Let $A$ be a self-adjoint operator on $\Hc$ and $N\in\N$. Let $\Phi_1,\Phi_2 \in \opinsert_N(A)$. Then $\Phi_1\Phi_2 \in \opinsert_N(A)$ and 
\[
\nr{\Phi_1 \Phi_2}_{A,N} \leq 2^N \nr{\Phi_1}_{A,N}\nr{\Phi_2}_{A,N}.
\]
\end{proposition}

\begin{proof}
This comes from the equality
\begin{equation} \label{dev-Phi12}
\ad _A^N (\Phi_1\Phi_2) = \sum_{n=0}^N C_N^n \ad _A^n (\Phi_1)\ad _A^{N-n}(\Phi_2). 
\end{equation}
\end{proof}

\begin{proposition}  \label{prop-phi-Ad}
Let $A$ be a self-adjoint operator on $\Hc$, $N\in\N$ and $\d \in [-N,N]$. Then there exists $C \geq 0$ such that for all $\Phi \in \opinsert_{N}(A)$ we have 
\[
 \nr{\pppg A ^{\d} \Phi \pppg A ^{-\d}}_{\Lc(\Hc)} \leq C \nr \Phi_{A,N}.
\]
\end{proposition}

\begin{proof}
First assume that $\d \in \Ii 0 N$. For all $\Phi \in \opinsert_N(A)$ we have
\[
(A-i)^\d \Phi = \sum_{m=0}^\d C_\d^k (-1)^k \ad_A^k (\Phi) (A-i)^{\d-k},
\]
from which the estimate easily follows. We proceed similarly if $\d$ is a negative integer, and we conclude by interpolation.
\end{proof}

Let $f \in \symb^{-\rho}(\R)$ for some $\rho > 0$. We consider an almost analytic extension $\tilde f$ of $f$ (see \cite{dimassis, davies95}):
\begin{equation*} 
 \tilde f (x+iy) = \p \left( \frac y {\pppg x} \right) \sum_{k=0}^m f^{(k)} (x) \frac {(iy)^k}{k!}
\end{equation*}
where $m \geq 2$ and $\p \in C_0^\infty(\R,[0,1])$ is supported on $[-2,2]$ and equal to 1 on $[-1,1]$. Writing $\z=x+iy$ we have
\begin{align*}
 \frac {\partial \tilde f} {\partial \bar \z} (\z) = \frac  {\partial_x  +i \partial_y  }2 \p \left( \frac y {\pppg x} \right) \sum_{k=0}^m f^{(k)} (x) \frac {(iy)^k}{k!} + \frac 1 2 \p \left( \frac y {\pppg x} \right)  f^{(m+1)} (x) \frac {(iy)^m}{m!} ,
\end{align*}
and in particular
\begin{equation} \label{estim-tildef}
\abs{\frac {\partial \tilde f}{\partial \bar \z} (\z)} \leq \1 {\singl{\pppg {\Re \z} \leq \abs{\Im \z} \leq 2 \pppg {\Re \z}} } (\z) \pppg {\Re \z}^{-1-\rho} + \1 {\singl{ \abs{\Im \z} \leq 2 \pppg{\Re \z}}} (\z) \abs{\Im \z}^m \pppg {\Re \z}^{-m-1-\rho}.
\end{equation}
%
%
Thus for any self-adjoint operator $A$ we can write the Helffer-Sj\"ostrand formula for $f(A)$:
\begin{equation}\label{helffer-sjostrand}
f(A) = - \frac 1 \pi \int_{\z = x + i y \in \C} \frac {\partial \tilde f}{\partial \bar \z} (\z) (A-\z)\inv \, dx \, dy
\end{equation}
(see \cite{dimassis, davies95} for more details).

\begin{proposition} \label{prop-comm-hplus}
Let $N \in \N^*$ and $\Phi \in \opinsert_{N}(A)$. Let $\d_1,\d_2 \in \R_+$ be such that $\d_1 + \d_2 < N$. Let $g_- \in \symb^{\d_1}(\R)$ and $g_+ \in \symb^{\d_2}(\R)$ be such that $\supp g_-\cap \supp g_+ = \emptyset$. Then there exists $C \geq 0$ such that for any self-adjoint operator $A$ and any $\Phi \in \opinsert _N(A)$ we have
\[
\nr{g_-(A) \Phi g_+(A)}_{\Lc(L^2)} \leq C  \nr \Phi _{A,N}.
\]
\end{proposition}

\begin{proof}
Let $\e = N-\d_1 - \d_2 > 0$. We can write 
\[
 g_+(A) = \pppg A ^{\d_2 + \e} \prod_{j=1}^N g_j(A)
\]
where for all $j \in \Ii 1 N$ we have $g_j \in \symb^{-\e / N}(\R)$ and $\supp g_- \cap \supp g_j = \emptyset$. Then we have
\[
 g_-(A) \Phi g_+(A) = g_-(A) \left(\ad_{g_1(A)} \dots  \ad_{g_N(A)} \Phi \right) \pppg A ^{\d_2 + \e}.
\]
According to \eqref{helffer-sjostrand} we have
\begin{eqnarray*}
\lefteqn{\left(\prod _{j=1}^N \ad_{g_j(A)}\right) \Phi}\\
&& = \frac {1}{\pi^N} \int_{\z_1} \dots \int_{\z_N} \left(\prod_{j=1}^N \frac {\partial \tilde g_j}{\partial \bar \z_j}(\z_j)(A-\z_j)\inv\right) \ad^N_A \Phi  \prod_{j=1}^N (A-\z_j)\inv \, dx_N\, dy_N \dots dx_1\,dy_1.
\end{eqnarray*}
For $\g \in [0,1]$ and $\z \in \C\setminus \R$ we have
\[
\nr{\pppg A^\g (A-\z)\inv }_{\Lc(\Hc)} \lesssim \frac {\pppg {\Re \z}^\g} {\abs{\Im \z}},
\]
so according to \eqref{estim-tildef} (with $\rho = \e / N$) there exists $C \geq 0$ such that for $j \in \Ii 1 N$ and $\z_j \in \C \setminus \R$ we have
\[
\abs{\frac {\partial \tilde g_j}{\partial \bar \z_j}(\z_j)} \nr{\pppg A^{\frac {\d_1} N} (A-\z_j)\inv}  \nr{ (A-\z_j)\inv\pppg A^{\frac {\d_2+\e} N}} \leq C \pppg {\z_j}^{-2-\frac \e N}.
\]
%
%
The result follows after integration.
\end{proof}

We can now state the main result of this section. Let $N\geq 2$. The families of operators $(A_\l)_{\l \in \L}$ and $(\Hl)_{\l \in \L}$ are as before. We also consider for all $k \in \Ii 1 N$ and $\l \in \L$ a pair $\Phi_{k,\l} = (\Phi_{k,\l,1},\Phi_{k,\l,2})$ of operators in $\opinsert_N(A_\l)$. For $k \in \Ii 1 N$ and $\l \in \L$ we set
\[
R^\Phi_{k,\l}(z) = \Phi_{k,\l,1} (\Hl -z)\inv  \Phi_{k,\l,2} .
\]

\begin{theorem} \label{th-estim-insert}
Suppose $(A_\l)_{\l \in \L}$ is uniformly conjugate to $(\Hl)_{\l \in \L}$ on $J$ with lower bound $(\a_\l)_{\l \in \L}$ and $\Hl$ is $N$-smooth with respect to $A_\l$. Let $I$ be a compact subinterval of $J$ and $n \in \Ii 1 N$. 
\begin{enumerate}[(i)]
\item If $\d > n-\frac 12$ there exists $c\geq 0$ such that for all $\l \in \L$ and $z \in \C_{I,+}$ we have
\begin{equation*}
\nr{ \pppg{A_\l}^{-\d} R^\Phi_{1,\l}(z) \dots R^\Phi_{n,\l}(z) \pppg{A_\l}^{-\d}} \leq \frac c {\a_\l^{n}} \prod_{k=1}^n \nr{\Phi_{k,\l}}_{A_\l,N}.
\end{equation*}
\item If $\d \in \left]n - \frac 12,N\right[$ there exists $c\geq 0$ such that for all $\l \in \L$ and $z \in \C_{I,+}$ we have
\begin{equation*}
\nr{ \pppg{A_\l}^{\d-n} \1{\R_-}(A_\l)R^\Phi_{1,\l}(z) \dots R^\Phi_{n,\l}(z)  \pppg{A_\l}^{-\d}} \leq \frac c {\a_\l^n} \prod_{k=1}^n \nr{\Phi_{k,\l}}_{A_\l,N}.
\end{equation*}
\item If $\d \in \left] n - \frac 12 , N \right[$ there exists $c\geq 0$ such that for all $\l \in \L$ and $z \in \C_{I,+}$ we have
\begin{equation*}
\nr{ \pppg{A_\l}^{-\d} R^\Phi_{1,\l}(z) \dots R^\Phi_{n,\l}(z) \1{\R_+}(A_\l) \pppg{A_\l}^{\d-n}} \leq \frac c {\a_\l^n}  \prod_{k=1}^n \nr{\Phi_{k,\l}}_{A_\l,N}.
\end{equation*}
\item If $\d_1,\d_2 \geq 0$ satisfy $\d_1 + \d_2 < N - n$ there exists $c\geq 0$ such that for all $\l \in \L$ and $z \in \C_{I,+}$ we have
\begin{equation*}
\nr{ \pppg{A_\l}^{\d_1} \1{\R_-} (A_\l) R^\Phi_{1,\l}(z) \dots R^\Phi_{n,\l}(z) \1{\R_+} (A_\l) \pppg{A_\l}^{\d_2}} \leq \frac c {\a_\l^n}  \prod_{k=1}^n \nr{\Phi_{k,\l}}_{A_\l,N}.
\end{equation*}
\end{enumerate}
\end{theorem}

Since the identity operator always belongs to $\opinsert_N(A_\l)$ with norm 1, Theorem \ref{th-estim-insert} can be seen as a generalization of Theorem \ref{thordren}.

Note that in all these estimates the constants $c$ may depend on $(\Hl)_{\l \in \L}$ and $(A_\l)_{\l\in\L}$ (in fact they depend on the constants which appears in properties $(c_n)$ and $(d_n)$ of Definition \ref{defconjunifdissn}), but they do not depend on the inserted factors $\Phi_{k,\l}$.

\begin{proof}
To prove this theorem we apply Lemma \ref{lemjen2.1} with $R_j = R^\Phi_{j,\l}(z)$ instead of ${(\Hl-z)\inv}$. Let $\d \in {\left] \frac 12, N \right[}$. According to Theorem \ref{thordren} and Proposition \ref{prop-phi-Ad} we have for all $ j\in\Ii 1 N$
\[
\nr{\pppg {A_\l}^{-\d} R^\Phi_{j,\l}(z) \pppg {A_\l}^{-\d}}_{\Lc(\Hc)} \lesssim \frac {\nr{\Phi_{j,\l}}_{A,N}} {\a_\l} .
\]
Now let $\h_0 \in C_0^\infty(\R,[0,1])$ be equal to 1 in a neighborhood of 0 and $\h_\pm = (1 - \h_0) \1{\R_\pm} \in C^\infty(\R,[0,1])$. For $* \in \{ - , 0 , +\}$, $\l \in \L$, $z \in \C_{I,+}$ and $j\in \Ii 1 N$ we set
\[
\Th^*_{j,\l}(z) = \pppg {A_\l} ^{\d-1} \1 {\R_-}({A_\l})  \Phi_{j,\l,1} \h_*({A_\l})  (\Hl -z)\inv  \Phi_{j,\l,2} \pppg {A_\l} ^{-\d}.
\]
As above we have
\[
\nr{\Th^0_{j,\l}(z)}_{\Lc(\Hc)} \lesssim \frac {\nr{\Phi_{j,\l}}_{A,N}} {\a_\l} .
\]
According to Theorem \ref{thordren}, the operator $\pppg {A_\l} ^{\d-1} \h_-({A_\l}) (\Hl-z)\inv \pppg {A_\l}^{-\d}$ is bounded uniformly in $\l \in \L$ and $z \in \C_{I,+}$ and hence, according to Proposition \ref{prop-phi-Ad}, $\Th^-_{j,\l}(z)$ is estimated as $\Th^0_{j,\l}(z)$. 
Let us now turn to $\Th^+_{j,\l}(z)$. Let $\tilde \h_- \in C^\infty(\R,[0,1])$ be equal to 1 on $\R_-$ and equal to 0 on $\supp(\h_+)$. According to Theorem \ref{thordren} and Propositions \ref{prop-phi-Ad} and \ref{prop-comm-hplus} we have
\begin{align*}
\nr{\Th^+_{j,\l}(z)}_{\Lc(\Hc)}
& \lesssim \frac {\nr{\Phi_{j,\l,2}}_{A,N}} {\a_\l}  \nr{\pppg {A_\l} ^{\d-1} \1 {\R_-}({A_\l})  \Phi_{j,\l,1} \h_+({A_\l}) \pppg {A_\l}}\\
& \lesssim \frac {\nr{\Phi_{j,\l,2}}_{A,N}} {\a_\l} \nr{\pppg {A_\l} ^{\d-1} \tilde \h_- ({A_\l})  \Phi_{j,\l,1} \h_+({A_\l}) \pppg {A_\l}}\\
& \lesssim \frac {\nr{\Phi_{j,\l,2}}_{A,N}} {\a_\l} \nr{\Phi_{j,\l,1}}_{A,N} .
\end{align*}
Finally we have proved that there exists $C \geq 0$ such that for all $\l \in \L$, $z \in \C_{I,+}$ and $j\in \Ii 1 N$ we have
\[
\nr{ \pppg {A_\l} ^{\d-1} \1 {\R_-}({A_\l})  R^\Phi_{j,\l}(z)  \pppg {A_\l} ^{-\d}}\leq \frac C {\a_\l} \nr{\Phi_{j,\l}}_{A,N}.
\]
The operators
$
\pppg {A_\l} ^{-\d}  R^\Phi_{j,\l}(z)  \1 {\R_+}({A_\l})  \pppg {A_\l} ^{\d-1}
$
and
$
 \pppg {A_\l} ^{\d_1} \1 {\R_-}({A_\l}) R^\Phi_{j,\l}(z) \1 {\R_+} \pppg {A_\l} ^{\d_2}
$
are treated similarly, and we conclude with Lemma \ref{lemjen2.1}.
\end{proof}

\begin{definition} \label{def-415}
Let $H$ be a maximal dissipative operator on $\Hc$. Let $A$ be an operator on $\Hc$ and $N\in\N$. We say that the pair of (non necessarily bounded) operators $\Phi = (\Phi_1,\Phi_2)$ belongs to $\opinsert_{N} (H,A)$ if the operators $\Phi_1(H-i)\inv$, $(H-i)\inv\Phi_2$ and $\Phi_1 (H-i)\inv\Phi_2$ belong to $\opinsert_N(A)$. In this case we set
\[
\nr{\Phi}_{H,A,N} = \nr{\Phi_1 (H-i)\inv \Phi_2}_{A,N} + \nr{\Phi_1 (H-i)\inv}_{A,N} \nr{ (H-i)\inv\Phi_2}_{A,N}.
\]
\end{definition}

This is an abstract condition which will be fulfilled (and standard to check) for the kind of differential operators we consider in this paper.

\begin{theorem} \label{th-estim-insert2}
Under the assumptions of Theorem \ref{th-estim-insert}, if for all $\l \in \L$ the pairs of operators $\Phi_{1,\l},\dots,\Phi_{N,\l}$ belong to $\opinsert_N(\Hl,A_\l)$ then in all the estimates we can replace $\prod_{k=1}^n \nr{\Phi_{k,\l}}_{A_\l,N}$ by $\prod_{k=1}^n \nr{\Phi_{k,\l}}_{\Hl,A_\l,N}$ if $z$ stays in a bounded subset of $\C_{I,+}$.
\end{theorem}

This new version of Theorem \ref{th-estim-insert} allows unbounded operators $\Phi_{k,\l}$. Moreover, even for inserted factors which belong to $\opinsert_N(A_\l)$, the estimate of Theorem \ref{th-estim-insert2} may be better than that of Theorem \ref{th-estim-insert}. In the application this will be useful to see the inserted factors of Proposition \ref{prop-der-R2} not only as bounded operators on $L^2$ but also as operators acting on suitable Sobolev spaces. This will be crucial at low frequency when these factors will be rescaled (see Proposition \ref{prop-estim-res-amort} below).


\begin{proof}
Let $\l \in \L$, $z \in \C_{I,+}$ and $k \in \Ii 1 N$. According to the resolvent equality we have
\begin{align*}
R^\Phi_{k,\l}(z)
& = (z-i) \Phi_{k,\l,1} (\Hl-i)\inv \Phi_{k,\l,2} + (z-i)^2 \Phi_{k,\l,1} (\Hl-i)^{-2} \Phi_{k,\l,2} \\
& \qquad + (z-i)^2 \Phi_{k,\l,1} (\Hl-i)\inv(\Hl-z)\inv(\Hl-i)\inv \Phi_{k,\l,2}.
\end{align*}
Replacing each factor $R^\Phi_{k,\l}(z)$ by any of the term in the right hand side and applying Theorem \ref{th-estim-insert} with Proposition \ref{prop-insert-prod} we obtain the estimates of Theorem \ref{th-estim-insert2}.
\end{proof}

Theorems \ref{th-estim-insert} and \ref{th-estim-insert2} give estimates for products of resolvents with inserted factors. In particular we are now able to prove results for the derivatives of the ``resolvent'' $R$ defined in \eqref{def-Rz}. 
In the rest of this paper we use these two results to prove Theorems \ref{th-inter-freq}, \ref{th-low-freq}, \ref{th-low-freq-bis} and \ref{th-high-freq}.

\section{Intermediate frequency estimates}  \label{sec-inter-freq}

The main goal of this section is to prove Proposition \ref{prop-inter-freq} below. Theorem \ref{th-inter-freq} will be a direct consequence of this statement.\\

We denote by $A$ the (self-adjoint realization of the) generator of $L^2$ dilations, namely
\begin{equation} \label{def-A}
 A = - \frac i2 ( x \cdot \nabla + \nabla \cdot x) = - i \, (x \cdot \nabla) - \frac {id}2. 
\end{equation}
Let us record the properties of $A$ we need in this paper: 
\begin{proposition} \label{prop-A}
\begin{enumerate}[(i)]
\item For $\th \in \R$, $u \in \Sc(\R^d)$ and $x \in \R^d$ we have
\begin{equation*} 
 (e^{i\th A} u) (x) = e^{\frac {d\th}2} u ( e^\th x).
\end{equation*}
\item For $j \in \Ii 1 d$ and $\g \in C^\infty(\R^d)$ we have on $\Sc(\R^d)$:
\begin{equation*} 
 [\partial_j , i A] = \partial_j \quad \text{and} \quad [\g,iA] = -(x \cdot \nabla) \g
\end{equation*}
 \item For $p \in [1,+\infty]$, $\th \in \R$ and $u \in \Sc(\R^d)$ we have
\[
 \nr{e^{i\th A} u}_{L^p} = e^{\th \left( \frac d 2 - \frac d p\right)} \nr u_{L^p}.
\]
\end{enumerate}
\end{proposition}

Let $(a_j)_{j\in\N^*}$ be a sequence of functions in $\symbor (\R^d)$. For $j,k \in \N$ with $j \leq k$ we set 
\begin{equation*} 
\Rc_{j,k} (z) = R(z) a_{j+1}(x) R(z) a_{j+2}(x) \dots R(z) a_{k}(z)R(z).
\end{equation*}
Given $\a_{j},\dots , \a_{k} \in \N$ we also define
\[
 \Th_{j ; \a_{j},\dots , \a_{k} }(z) = \big( \Hz + 1\big)^{-\a_{j}}  a_{j+1}(x) \big(\Hz + 1\big)^{-\a_{j+1}} \dots {a_{k}(x)} \big(\Hz + 1\big)^{-\a_{k}}. 
\]
In the self-adjoint case the resolvent $(H_0 -z^2)\inv$ can be replaced by a spectrally localized version $(H_0 -z^2)\inv \h(H_0)$ with $\h \in C_0^\infty(\R)$, which is technically useful since $\h(H_0)$ is a smoothing operator. Here the non self-adjointness of $H_z$ prevents us from using such a localization. The operators $ \Th_{j ; \a_{j},\dots , \a_{k} }(z)$ will be a suitable replacement of $\h(H_0)$ in the end of the proof of Proposition \ref{prop-inter-freq}.

\begin{proposition} \label{prop-inter-freq}
Let $K$ be a compact subset of $\C\setminus \singl 0$, $n \in \N$, $\d > n + \frac 12$ and $\a \in \N^d$ with $\abs \a \leq 1$. Then there exists $C \geq 0$ such that for all $z \in  K \cap \C_{+}$ we have
\[
\nr { \pppg x ^{-\d} D^\a \Rc_{0,n}(z) \pppg x ^{-\d}}_{\Lc(L^2)}\leq C.
\]
\end{proposition}

Note that according to Proposition \ref{prop-der-R2}, the derivatives $R^{(n)}(z)$ which appear in Theorem \ref{th-inter-freq} are linear combinations of terms of the form $\Rc_{0,m}(z)$ with $a_j$ equal to 1 or $a$.\\

Before giving a full proof of this proposition, we briefly outline its main ideas. We check that $A$ is conjugate to the dissipative operator $\Hz$ up to any order uniformly in $z \in K\cap \C_{+,+}$ and then we prove that multiplication by $a_k$ belongs to $\opinsert_N(A)$ for any $N \in \N$. According to Theorem \ref{th-estim-insert} there exists $C \geq 0$ such that for all $z \in K_{+,+}$ we have
\begin{equation} \label{estim-ARc}
 \nr{ \pppg A^{-\d} \Rc_{j,k}(z) \pppg {A}^{-\d}}_{\Lc(L^2)} \leq C.
\end{equation}
Finally, as usual, it remains to replace the weight $\pppg A ^{-\d}$ by $\pppg x ^{-\d}$.

\begin{proof}
\stepp 
Let $\m_0 > 0$. According to Proposition \ref{prop-R-diss} we can assume without loss of generality that $\Im z \leq \m_0$ for all $z \in K\cap \C_{+}$. We prove the result for $z \in K \cap \C_{+,+}$. Reversing the order of the inserted factors and taking the adjoint give the result for $z \in K \cap \C_{-,+}$. If $\m_0$ is small enough, there exists a compact subinterval $J$ of $\R_+^*$ such that $\Re (z^2) \in J$ for all $z \in K \cap \C_{+,+}$.
So let us prove that the operator $A$ is conjugate to $\Hz$ on $J$ and up to order $N$ uniformly in $z \in K\cap \C_{+,+}$. In fact it is enough to prove that for any $E > 0$ there exists a neighborhood $J_E$ of $E$ in $\R$ such that $A$ is conjugate to $\Hz$ on $J_E$.

\stepp We first check that the generator of dilations $A$ satisfies the first four assumptions ($a$)-($d$) to be uniformly conjugate to the family of dissipative operators $(\Hz)_{z \in K\cap \C_{+,+}}$. The real part of $\Hz$ is $\rehz = \Ho + \Im (z) a$. Its domain is $H^2(\R^d)$. For ($a$) we only have to remark that $\Sc(\R^d) \subset H^2(\R^d) \cap \Dom(A)$ is dense in $H^2(\R^d)$. Assumption ($b$) easily follows from Proposition \ref{prop-A}, and all the commutators involved in assumptions ($c_n$) and ($d_n$) can be computed explicitly using again Proposition \ref{prop-A}.

\stepp For the positive commutator assumption ($e$) we apply the usual trick. According to Proposition \ref{prop-A} we can write 
\[
\left[\rehz ,iA \right] = 2 \rehz + W 
\]
where
\[
W = -2 \Im (z)a - \Im (z) (x\cdot \nabla) a + \sum_{j,k= 1}^d \g_{j,k}(x) D_j D_k + \sum_{k= 1}^d b_{k}(x) D_k
\]
for some $\g_{j,k}, b_k \in \symb^{-\rho}(\R^d)$, $j,k \in\Ii 1 d$.
Let $E > 0$, $\s > 0$ and $E_\s = [E-\s,E+\s]$. Composing with the projection $\1 {E_\s}\big(\rehz\big)$ on both sides of the commutator and using, as usual for Mourre estimates, that $ {\mathds 1}_{E_{\sigma}} (B)B \geq (E-\sigma) {\mathds 1}_{E_{\sigma}}(B) $ for any self-adjoint operator $ B $, we get
\begin{multline}
\1 {E_\s}\big(\rehz\big) \left[\rehz , iA \right] \1 {E_\s}\big(\rehz\big)\\
\geq 2(E-\s)\1 {E_\s}\big(\rehz\big) + \h_{E,\s} \big(\rehz\big) Q\,  \h_{E,\s} \big(\rehz\big).
\end{multline}
Here $\h_{E,\s} \in C_0^\infty(\R, [0,1])$ is supported in $E_{2\s}$ and is equal to 1 on $E_\s$, and 
\[
 Q = \1 {E_\s}\big(\rehz\big)  W \1 {E_\s}\big(\rehz\big)
\]
is compact. According to the Helffer-Sj\"ostrand formula \eqref{helffer-sjostrand} and the resolvent equation between $\Re(\Hz)$ and $\Ho$ we have in $\Lc(L^2)$:
\begin{align*}
\h_{E,\s}\big(\rehz\big)
& = \h_{E,\s}(\Ho) + \frac {\Im (z)}{\pi}  \int_{\z = x+iy \in \C} \frac {\partial \tilde  \h_{E,\s}}{\partial \bar \z} (\z) (\Ho-\z)\inv a (\rehz -\z)\inv \, dx \, dy \\
& = \h_{E,\s}(\Ho) + \bigo {\Im z} 0 (\Im (z)).
\end{align*}
We know from \cite{kocht06} that $E$ is not an eigenvalue of $\Ho$. If we choose $\s > 0$ and $\m_0 > 0$ small enough, then for $z \in \C_{+,+} \cap K$ with $\Im z \leq \m_0$ we have
\[
\1 {E_\s}(\rehz) [\rehz , iA] \1 {E_\s}(\rehz) \geq (E-\s)\1 {E_\s}(\rehz).
\]
This follows from the usual trick that $ \chi (B) $ goes weakly to $ 0 $ as the support of $ \chi $ shrinks to a point which is not an eigenvalue of $B$ and also from the compactness of $Q$. This proves that assumption (e) holds on $E_\s$.

\stepp Since multiplication by $a$ belongs to $\opinsert_N(A)$ uniformly in $z \in K \cap \C_{+,+}$ for all $N \in \N$ we can apply Theorem \ref{th-estim-insert2} to obtain \eqref{estim-ARc} for $z \in K \cap \C_{+,+}$.

\stepp
We prove by induction on $m \in \N$ that $\Rc_{0,n}$ can be written as a sum of terms either of the form 
\begin{equation*} 
\big( 1 +  z ^2 \big)^\b   \Th_{0 ; \a_0,\dots,\a_n}(z) 
\end{equation*}
with $\b \in \N$, $\a_0,\dots, \a_n \in \N^*$, or 
\begin{equation} \label{terme2}
\big( 1 +  z ^2 \big)^\b  \Th_{0 ; \a_0,\dots,\a_j} (z) \Rc_{j,k}(z)\Th_{k ; \a_{k},\dots , \a_n}(z) 
\end{equation}
where $\b \in \N$,  $j,k\in\Ii 0 n$, $j\leq k$, $\a_0,\dots, \a_{j-1} , \a_{k+1}, \dots , \a_n \in \N^*$, $\a_j,\a_{k} \in\N$, 
\[
\sum_{l=0}^j \a_l \geq m \quad \text{and} \quad \sum_{l=k}^{n} \a_l \geq m. 
\]
It is clear when $m = 0$, and for the inductive step we only have to consider a term like \eqref{terme2}. If $j = k$ we only have to write 
\[
R(z) = \big( \Hz +1 \big) \inv + \big(1+  z ^2 \big) \big( \Hz +1 \big) ^{-2} + \big(1+  z ^2 \big)^2 \big( \Hz +1 \big) \inv  R(z)   \big( \Hz +1 \big) \inv.
\]
if $j+1=k$ we have 
\begin{align*}
\Rc_{j,k}(z)
& = (\Hz + 1)\inv  \left( 1 + (1 +  z ^2 ) R(z) \right) a_{k}(x)  \left( 1 + (1 + z ^2 ) R(z) \right) (\Hz + 1)\inv 
\end{align*}
and finally if $j + 2 \leq  k$:
\begin{align*}
 \Rc_{j,k}(z)
& =  (\Hz + 1)\inv  \left( 1 + (1 +  z ^2 ) R(z) \right) a_{j+1} (x)  \Rc _{j+1,k-1} (z) a_{k}(x)  \\
& \qquad \times \left( 1 + (1 + z ^2 ) R(z) \right) (\Hz + 1)\inv .
\end{align*}
For any $\a_0,\dots, \a_n \in \N^*$ it is clear that $\pppg x ^{-\d}  \Th_{0 ; \a_0,\dots,\a_n}(z) \pppg x ^{-\d} $ is bounded on $L^2$ uniformly in $z \in K \cap \C_{+,+}$. For a term of the form \eqref{terme2} we remark that for $m$ large enough the operators 
\[
\pppg x ^{-\d} D^\a \Th_{0 ; \a_0,\dots,\a_j} (z) \pppg A ^\d \quad \text{and} \quad  \pppg A ^\d \Th_{k ; \a_{k},\dots , \a_n}(z)  \pppg x ^{-\d}
\]
are bounded uniformly in $z \in K\cap \C_{+}$. For instance for the first one we use on one hand that $\pppg x ^{-\d} D^\a \Th_{0 ; \a_0,\dots,\a_j} (z) \pppg D^{2m - \abs \a} \pppg x ^{\d}$ is uniformly bounded, and on the other hand that $\pppg x ^{-\d} \pppg D^{-2m + \abs \a} \pppg A^\d$ is bounded, which follows from an interpolation argument. Then we conclude with \eqref{estim-ARc}.
\end{proof}

\section{Low frequency estimates} \label{sec-small}

In this section we prove Theorems \ref{th-low-freq} and \ref{th-low-freq-bis}. We will first consider a globally small perturbation of $-\D$, and then use a perturbation argument to deal with the general case.\\


Let $\h \in C_0^\infty(\R^d,[0,1])$. We set 
\[
\Poo = -\divg  \big(\h I_d + (1-\h) G \big) \nabla \quad \text{and} \quad  K_0 = \Ho - \Poo = -\divg \big(\h(G-I_d)\big) \nabla .
\]
These two operators can be written
%
%
%
\[
 \Poo = \sum_{j,k=1}^d  D_j \g_{j,k} D_k, 
\qquad K_0 = \sum_{j,k=1}^d  D_j\g_{j,k}^0 D_k ,
\]
where for $j,k \in \Ii 1 d$ the coefficients $\g_{j,k}^0$ are compactly supported. Moreover $\h$ can be chosen in such a way that the coefficients $\g_{j,k}- \d_{j,k}$ are small in $\symb^{-\rho}(\R^d)$, in a sense to be made precise in Theorem \ref{th-low-freq-R0} below.\\

Let $a_\iota = (1-\h) a$, $a_0 = \h a$, and for $z \in \C_{+}$:
\begin{eqnarray}
\Pol = \Poo - iza_\iota ,\qquad K_z = K_0 -iz a_0. \label{defKz}
\end{eqnarray}
The operator $\Poo$ is self-adjoint and non-negative on $L^2$ with domain $H^2$. $\Pol$ is maximal dissipative on $H^2$ for all $z \in \C_{+,+}$, self-adjoint for $z \in i\R_+^*$ and $\Pol = P_{-\bar z}^*$ when $z \in \C_{-,+}$. Thus for $z \in \C_{+}$ we can define
\[
\RPz = \big( P_z -z^2\big) \inv \in \Lc(L^2).
\]
On the other hand $K_0$ is symmetric and $K_z$ is dissipative for $z \in \C_{+,+}$. These operators are at least defined on $H^2$. \\

Let us fix an integer $\bar d$ greater than $\frac d 2$.
For $\n \in \big[0,\frac d 2 \big[$, $N\in\N$ and $\vf \in \symb^{-\n-\rho}(\R^d)$ we set
\begin{equation} \label{def-nr-nu-N}
\nr{\vf}_{\n,N} =  \sup_{\abs \a \leq \bar d} \, \sup_{0\leq m \leq N} \sup_{x \in \R^d}\abs {\pppg x ^{ \n + \rho +\abs \a} \big(\partial^\a (x\cdot \nabla)^m \vf \big)(x)}.
\end{equation}
Then for $N \in \N$ we put
\begin{equation} \label{def-Nc}
\Nc_N = \sum_{j,k=1}^d \nr{\g_{j,k}-\d_{j,k}}_{0,N}  +  \nr{a_\iota}_{1,N}.
\end{equation}

In order to obtain Theorem \ref{th-low-freq}, we first prove an analogous result for $\RP$:

\begin{theorem} \label{th-low-freq-R0}
Let $n \in \N$ and $\e > 0$. 
\begin{enumerate} [(i)]
\item Let $\d$ be greater than $n + \frac 12$ if $n \geq \frac d 2$ and greater than $n + 1$ otherwise. Then if $\Nc_1$ is small enough there exists $C \geq 0$ such that for all $\a \in \N^d$ with $\abs \a \leq 1$ and $z \in \C_{+}$ we have
%
\begin{equation}\label{eq-low-freq3}
\nr{\pppg x ^{-\d}  D^\a \RP^{(n)}(z) \pppg x ^{-\d} }_{\Lc(L^2)}  \leq C \left( 1 + \abs z ^{d -2 + \abs \a - n - \e} \right).
\end{equation}
\item 
Assume that $d$ is odd or $n\neq d-2$. Let $\d_1,\d_2 > n + \frac 12$ be such that $\d_1 + \d_2 > \min(2(n+1), d)$. Then there exists $C \geq 0$ such that for all $z \in \C_ {+}$ we have
\[
\nr{\pppg x ^{-\d_1} \RP^{(n)}(z) (z) \pppg x ^{-\d_2}}_{\Lc(L^2)} \leq C \left( 1 + \abs{z}^{d-2-n}\right).
\]
\end{enumerate}
\end{theorem}

The proof of Theorem \ref{th-low-freq-R0} will be given after Proposition \ref{prop-estim-res-amort}. For this, we are going to use a scaling argument. For $\vf \in C^\infty(\R^d)$, $\l > 0$ and $x \in\R^d$ we set
\begin{equation} \label{not-tau}
\vf_\newl (x) = \vf\left( \frac  x \l \right).
\end{equation}
For $z \in \C_{+}$ we set
\begin{align*} 
\Ptauo =  \frac 1 {\abs z^2} e^{-iA \ln \abs z} \Poo e^{iA \ln \abs z} = \sum_{j,k=1}^d  D_j \g_{j,k, \abs z} D_ k 
\end{align*}
and 
\begin{equation}\label{def-Ptau}
\Ptau =  \frac 1 {\abs z^2} e^{-iA \ln \abs z} \Pol e^{iA \ln \abs z} = \Ptauo - i \frac {\hat z} { \abs z} a_{\iota, \abs z}
\end{equation}
(where $\hat z$ stands for $z / \abs z$, we recall that $e^{iA\ln \abs z}u(x) = |z|^{d/2} u (\abs z x)$).
For $z \in \C_{+,+}$ the operator $\Ptau$ is maximal dissipative on $H^2$, and as before we can consider for all $z \in \C_{+}$
\begin{equation*} 
 \tRtau = \big(\Ptau- \hat z^2\big) \inv,
\end{equation*}
so that
\begin{equation*} 
\RPz  = \frac 1 { \abs z^2} e^{i A \ln \abs z} \tRtau e^{-i A \ln  \abs z}.
\end{equation*}
We are going to use the Mourre method to obtain uniform estimates of $\tRtau$ when $z \in \C_+$ is close to 1. To this end we first give some properties for operators of multiplication by functions of the form $\vf_\l$ when $\l > 0$ goes to 0.\\

Before going further, we introduce some notation we shall use in this section.
Let $\seq \n j \in \{0,1\}^\N$. For $j \in \N^*$ and $z \in \C_+$ we define the operator $\Phi_j$ as the multiplication by $a_\iota^{\n_j}$ and $\Phi_0$ is of the form $D^\a$ with $\abs \a = \n_0$. Then we set
\[
 \tilde \Phi_j(z) = e^{-i A \ln  \abs z} \Phi_j e^{i A \ln  \abs z} = \begin{cases} \abs z^{\n_0} D^{\a} & \text{if } j = 0 , \\ a_{\iota , \abs z}^{\n_j}& \text{if } j \in \N^* \end{cases} 
\]
(here again and everywhere in the sequel, the index $\abs z$ corresponds to \eqref{not-tau} with $\l = \abs z$).
These operators have a very particular form, but the only properties we need are that of Corollary \ref{cor-Phi-nu} below.
For $z \in \C_{+}$ and $j,k \in \N$ such that $j \leq k$ we set
\begin{equation*} 
\Rc_{j,k}^\iota (z) =  \RPz \Phi_{j+1}  \RPz  \dots  \Phi_k  \RPz
\end{equation*}
and
\begin{equation*} 
\tilde \Rc^\iota _{j,k}(z) =   \tRtau  \tilde \Phi_{j+1}(z) \tRtau \dots \tilde \Phi_{k}(z) \tRtau.
\end{equation*}
For $\a_{j},\dots , \a_{k} \in \N$ we also define
\begin{equation} \label{def-Th}
 \tThiota_{j ; \a_{j},\dots , \a_{k} }(z) =  \big( \Ptau + 1\big)^{-\a_{j}}  \tilde \Phi_{j+1}(z)  \big( \Ptau + 1\big)^{-\a_{j+1}}  \dots\tilde \Phi_{k}(z)  \big( \Ptau + 1\big)^{-\a_{k}} 
\end{equation}
(this is well-defined since $\Ptau$ is maximal dissipative with non-negative real part). Finally, for all $j,k \in \N$ with $j\leq k$ we set 
\begin{equation}
\Vc_{j,k} = \sum_{l=j+1}^{k} \n_l. \label{matcalV}
\end{equation}
We recall that these notations are recorded in Appendix \ref{section-notations}.

\subsection{Some properties of the rescaled operators}  \label{sec-dec-functions}

 In this paragraph we derive some properties of the rescaled operators $\Ptau$ and $\tilde \Phi_j$. Most of them rely on the following proposition, in which we show that the spatial decay of \eqref{dec-metric-a} induces some differentiation-like properties for the rescaled coefficients $\g_{j,k,\abs z}$ and $a_{\iota, \abs z}$.

\begin{proposition} \label{prop-dec-sob}
Let $\n \in \big[ 0, \frac d 2 \big[$ and $s \in \big] -\frac d 2, \frac d 2\big[$ be such that $s -\n \in \big] -\frac d 2, \frac d 2\big[$. Then there exists $C \geq 0$ such that for $\vf \in \symb^{-\n-\rho}(\R^d)$, $u \in H^s$ and $\l > 0$ we have
\[
 \nr{\vf_\l u}_{\dot H^{s-\n}} \leq C \l ^\n \nr \vf_{\n,0} \nr u _{\dot H^s}
\]
and
\[
 \nr{\vf_\l u}_{H^{s-\n}} \leq C \l ^\n \nr \vf_{\n,0} \nr u _{H^s}.
\]
\end{proposition}

\begin{corollary} \label{cor-dec-sob}
Let $s \in \big] -\frac d 2, \frac d 2\big[$. Then there exists $C \geq 0$ such that for $\vf \in C^\infty(\R^d)$ which satisfies $\vf-1 \in \symb^{-\rho}(\R^d)$, $u \in H^s$ and $\l > 0$ we have
\[
 \nr{\vf_\l u}_{H^{s}} \leq C \left(1+\nr {\vf-1} _{0,0}\right) \nr u _{H^s}.
\]
\end{corollary}

\begin{proof}
For $u \in H^s$ and $\l > 0$ we have 
\[
  \nr{\vf_\l u}_{H^{s}} \leq   \nr{(\vf-1)_\l u}_{H^{s}} +   \nr{ u}_{H^{s}},
\]
so we only have to apply Proposition \ref{prop-dec-sob} to $\vf-1$.
\end{proof}

Proposition \ref{prop-dec-sob} mostly relies on the following result (see for instance \cite{lemarieg06}, Lemmas 1 and 5):

\begin{lemma} \label{lem-multipliers}
\begin{enumerate}[(i)]
\item Let $s \in \big] -\frac d 2, \frac d 2\big[$ and $\n \in \big]0, \frac d 2\big[$ be such that $s -\n \in \big] -\frac d 2, \frac d 2\big[$. Then there exists $C \geq 0$ such that for all $\vf \in \dot H^{\frac d 2 - \n}$ and $u \in \dot H^s$ we have
\[
 \nr{\vf u}_{\dot H^{s-\n}} \leq C \nr{\vf}_{\dot H^{\frac d 2 - \n}} \nr {u}_{\dot H^s}.
\]
\item Let $s \in \big]0, \frac d 2\big[$. Then there exists $C \geq 0$ such that for all $\vf \in \dot H^{\frac d 2} \cap L^\infty$ and $u \in \dot H^s$ we have
\[
 \nr{\vf u}_{\dot H^{s}} \leq C \left(\nr{\vf}_{\dot H^{\frac d 2 }} + \nr \vf _{L^\infty} \right) \nr {u}_{\dot H^s}.
\]
\end{enumerate}
\end{lemma}

\begin{proof}[Proof of Proposition \ref{prop-dec-sob}]
\stepp First assume that $s - \n \leq 0 \leq s $. Then according to Sobolev embeddings and H\"older inequality we have
\begin{align*}
 \nr{\vf_\l u} _{H^{s-\nu}}
\lesssim \nr {\vf_\l u} _{L^{\frac {2d}{d + 2(\n-s)}}} \lesssim  \nr {\vf_\l} _{L^{\frac d \n}} \nr u _ {L^{\frac {2d}{d - 2s}}} \lesssim  \l^\n \nr \vf _{\n,0} \nr u_{H^s}.
\end{align*}
We proceed similarly in homogeneous Sobolev spaces.

\stepp Now we assume that $s  \geq s-\n \geq 0$. The case $s-\n \leq s \leq 0$ will follow by taking the adjoint.
We first remark that for all $\l > 0$ we have
\begin{equation} \label{scale-Hd2}
\nr{\vf_\l}_{\dot H^{\frac d 2 - \n}} = \l ^\n \nr{\vf}_{\dot H^{\frac d 2 - \n}}.
\end{equation}
Let $\h_0 \in C_0^\infty(\R^d,[0,1])$ and $\h \in C_0^\infty(\R^d\setminus \singl 0,[0,1])$ be such that $\h_0 + \sum_{k=1}^\infty \h_k = 1$, where for $x \in \R^d$ and $k\in\N^*$ we have set $\h_k(x) = \h\big( 2^{-k}x\big)$.
According to \eqref{scale-Hd2} applied to $\h_k \vf$ we have for all $k \in \N^*$
\[
\nr{ \h_k \vf}_{\dot H^{\frac d 2 -\n}} = 2^{k\n} \nr{ \h \vf (2^k \cdot)}_{\dot H^{\frac d 2 -\n}} \leq 2^{k\n} \nr{ \h \vf(2^k \cdot)}_{H^{\bar d}}
\]
(we recall that $\bar d$ was fixed before \eqref{def-nr-nu-N}). And since $\h$ is compactly supported we have
\begin{align*}
\nr{ \h \vf(2^k \cdot)}_{H^{\bar d}} 
& \lesssim \nr{ \vf(2^k \cdot)}_{H^{\bar d}(\supp \h)} \lesssim \sum_{\abs \a \leq \bar d} 2^{k \abs \a} \nr{(\partial ^\a \vf)(2^k \cdot)}_{L^2(\supp \h)} \\
& \lesssim \sum_{\abs \a \leq \bar d} 2^{k \abs \a} \nr{(\partial ^\a \vf)(2^k \cdot)}_{L^\infty(\supp \h)} \lesssim 2^{-k(\n+\rho)} \nr{\vf}_{\n,0}.
\end{align*}
And finally:
\[
 \nr { \vf}_{\dot H^{\frac d 2 -\n}} \leq  \nr {\h_0 \vf}_{\dot H^{\frac d 2 -\n}} + \sum_{k=1}^\infty \nr {\h_k \vf}_{\dot H^{\frac d 2 -\n}} \lesssim \nr{\vf}_{\n,0}.
\]
Moreover the norm of $\vf_\l$ in $L^\infty$ does not depend on $\l > 0$, so we get the first assertion according to Lemma \ref{lem-multipliers}.

\stepp To prove the second assertion, we only have to estimate the $L^2$-norm of $\vf_\l u$. The first statement of the proposition applied with $s = \n$ gives
\[
 \nr{\vf_\l u}_{L^2}  \lesssim\l^\n \nr{\vf}_{\n,0}\nr{ u}_{\dot H^\n} .
\]
Since $\nr u_{\dot H^\n} \leq \nr{u}_{H^s}$, this concludes the proof. 
\end{proof}

Given $\m \in \N^d$, we set 
\[
\ad_x^\m := \ad_{x_{1}}^{\m_1} \dots \ad_{x_{d}}^{\m_d}.
\]


As a direct consequence of Proposition \ref{prop-dec-sob} we obtain the following properties on the operators $\tilde \Phi_j(z)$:

\begin{corollary} \label{cor-Phi-nu} Let $\m \in \N^d$ and $n \in \N$.
\begin{enumerate}[(i)]
\item There exists $C\geq 0$ such that for all $s \in \big]- \frac d 2 ,\frac d 2 - 1 \big[$, $j \in \N^*$ and $z \in \C_{+}$ we have
\[
 \nr{\ad_x^\m \ad_{iA}^n \tilde \Phi_j(z)}_{\Lc(H^{s+1},H^s)} \leq C \abs z ^{\n_j}.
\]
\item There exists $C\geq 0$ such that for all $s \in \big]- \frac d 2 ,\frac d 2\big[$ and $z \in \C_{+}$ we have
\[
 \nr{\ad_x^\m \ad_{iA}^n  \tilde \Phi_0(z)}_{\Lc(H^{s+1},H^s)} \leq C \abs z ^{\n_0}.
\]
\end{enumerate}
\end{corollary}

This corollary records all the properties we need about the operators $ \tilde{\Phi}_j(z) $. From now on, we will only use that they behave
 like differentiations, keeping  in mind the restriction on the Sobolev index $s$. Notice that this restriction is slightly weaker on $\tilde \Phi_0$. Here this simply comes from the fact that $\tilde \Phi_0$ is really a derivative and maps $H^{s+1}$ to $H^s$ for any $s \in \R$. The good behavior of the commutators with  $A$ and $x$ will be useful to apply Mourre theory and in the proof of Proposition \ref{prop-Th3} below.\\

 In the following proposition we estimate the difference between the commutators of $\Ptau$ and $-\D$ with $A$ and $x$. We will see that it only depends on the semi-norms $\Nc$ introduced in \eqref{def-Nc}. This will be used to apply Mourre theory and in the proof of Proposition \ref{prop-Ptau}.

\begin{proposition} \label{prop-Ptau-comm-Hs}
Let $ n \in \N$, $\m \in \N^d$ and $s \in \big]-\frac d 2, \frac d 2 \big[$. Then there exists $C \geq 0$ such that for all $z \in \C_{+}$ we have
\[
\nr{\ad_x^\m \, \ad_{iA}^n \big(\Ptauo+\D\big)}_{\Lc(H^{s+1},H^{s-1})} 
+ \nr{\ad_x^\m \ad_{iA}^n (\hat z \abs z\inv a_{\iota,\abs z})}_{\Lc(H^{s+1},H^{s-1})} 
\leq C \Nc_n.
\]
\end{proposition}

\begin{proof}
For $j,k \in \Ii 1 d$ and $n \in \N$ we set 
\begin{equation}\label{gamma-m}
\g_{j,k}^{(n)} = (2 - x \cdot \nabla)^n \g_{j,k}
\quad  \text{and} \quad a^{(n)}_\iota = (- x \cdot \nabla)^n a_\iota.
\end{equation}
Using Proposition \ref{prop-A} we can check by induction on $n \in \N$ that for all $z \in \C_{+}$ we have
\[
 \ad_{iA}^n (\Ptauo) = \sum_{j,k}   D_j \g_{j,k,\abs z}^{(n)} D_k 
\quad 
\text{and}
\quad
 \ad_{iA}^n \left(a_{\iota, \abs z}\right) =   a_{ \iota , \abs z}^{(n)}.
\] 
According to Proposition \ref{prop-dec-sob}, for $s \in \big]-\frac d 2 ,\frac d 2[$, $z \in\C_{+}$ and $u \in \Sc(\R^d)$ we have
\[
\nr{\ad_{iA}^n(\Ptauo+\D) u }_{H^{s-1}} \lesssim \sum_{j,k} \nr{\big(\g_{j,k}^{(n)}- 2^n \d_{j,k} \big)D_k u}_{H^s} \lesssim \sum_k \Nc_{n} \nr{D_k u}_{H^{s}}  \lesssim \Nc_{n} \nr{u}_{H^{s+1}}
\]
(note that $\g_{j,k}^{(n)}- 2^n \d_{j,k}=(2 - x \cdot \nabla)^n (\g_{j,k}-\d_{j,k})$). On the other hand, if $s \in \left]-\frac d 2 + 1,\frac d 2\right[$ we have 
\[
\nr{ \ad^n_{iA} (\hat z \abs z\inv a_\iota) u}_{H^{s-1}} \lesssim \Nc_{n} \nr{u}_{H^{s}}\leq \Nc_{n} \nr{u}_{H^{s+1}},
\]
and if $s \in \left]-\frac d 2 ,\frac d 2-1\right[$: 
\[
\nr{ \ad^n_{iA} (\hat z \abs z\inv a_\iota) u}_{H^{s-1}} \leq \nr{ \ad^n_{iA} (\hat z \abs z\inv a_\iota) u}_{H^{s}}\lesssim \Nc_{n} \nr{u}_{H^{s+1}}.
\]
This proves the proposition when $\m=0$. Now let $l,p \in \Ii 1 d$. We have 
\[
\big[\ad_{iA}^n \big(\Ptauo + \D \big) ,x_l\big] = -i \sum_{j} D_j \g^{(n)}_{j,l,\abs z}  -i \sum_{k} \g^{(n)}_{l,k,\abs z} D_k + i 2^{n+1}  D_l,
\]
\[
\big[\big[\ad_{iA}^n \big(\Ptauo + \D\big), x_l\big],x_p\big] = -2 \g^{(n)}_{l,p,\abs z},
\]
and hence $\ad_x^\m \ad_{iA}^n (\Ptauo) = 0$ if $\abs \m \geq 3$. These operators can be estimated as $\ad_{iA}^n (\Ptauo + \D)$. Since $\ad_x^\m \ad^n_{iA} (\hat z \abs z\inv a_\iota) = 0$ if $\abs \m \geq 1$, this concludes the proof.
\end{proof}

\begin{remark} \label{rem-estim-tPtau}
Using Corollary \ref{cor-dec-sob} we can similarly prove that there exists $C \geq 0$ such that for all $z \in \C_{+}$ we have
\[
\nr{\ad_x^\m \, \ad_{iA}^n \Ptauo}_{\Lc(H^{s+1},H^{s-1})} \leq C \Nc_n.
\]
\end{remark}

\begin{remark} \label{rem-lap-PO}
With the same proof we can prove that if $\Nc_0$ is small enough then for all $u \in \Sc(\R^d)$ we have
\[
\nr{(\Poo + \D)u} \leq \frac 12 || \Delta u ||_{L^2} 
\]
and in particular 
\[
|| \Delta u||_{L^2} \lesssim \nr{\Poo u}_{L^2}.
\]

\end{remark}

\begin{proposition} \label{prop-Ptau}
Let $n \in \N$, $\m \in \N^d$ and $s \in \big]-\frac d 2 ,\frac d 2[$. If $\Nc_0$ is small enough there exists $C \geq 0$ such that for all $z \in \C_{+}$ and $u \in \Sc(\R^d)$ we have
\[
 \nr{\ad_x^\m \ad_{iA}^n \left( \big( \Ptau +1 \big) \inv\right) u}_{H^{s+1}} \leq C \nr{u}_{H^{s-1}}.
\]
\end{proposition}

This statement means that if $\Ptau$ is close to $-\D$ then the resolvent $(\Ptau + 1)\inv$ has (uniformly) the same elliptic property as $(-\D + 1)\inv$. Of course this holds with the same restriction on the Sobolev index $s$ as in Proposition \ref{prop-Ptau-comm-Hs}.

\begin{proof}
According to Proposition \ref{prop-Ptau-comm-Hs}, if $\Nc_{0}$ is small enough then $(\Ptau +1)$ is close to $(-\D+1)$ in $\Lc (H^{s+1},H^{s-1})$ uniformly in $z \in \C_{+}$, which gives the result if $\m = 0$ and $n = 0$. Now the operator $\ad_x^\m \ad_{iA}^n \left( \big( \Ptau +1 \big) \inv\right)$ can be written as a linear combination of terms of the form
\[
 \big( \Ptau +1 \big) \inv\ad_x^{\m_1} \ad_{iA}^{n_1}  \big( \Ptau \big)   \big( \Ptau +1 \big) \inv \dots  \ad_x^{\m_k} \ad_{iA}^{n_k} \big( \Ptau \big)   \big( \Ptau +1 \big) \inv,
\]
where $k \in \N$, $n_1,\dots,n_k \in \N$ and $\m_1,\dots ,\m_k \in \N^d$ are such that $n_1 + \dots + n_k = n$ and $\m_1 + \dots + \m_k = \m$. We have proved that $\big( \Ptau +1 \big) \inv$ is uniformly bounded in $\Lc(H^{s-1},H^{s+1})$ and according to Remark \ref{rem-estim-tPtau} the operator $\ad_x^{\m_l} \ad_{iA}^{n_l}  \big( \Ptau \big) $ is uniformly bounded in $\Lc(H^{s+1},H^{s-1})$ for all $l \in \Ii 1 k$.
\end{proof}

We now give uniform estimates for the operators $\tilde \Th_{j;\alpha_j, \ldots , \alpha_k}(z)$ defined in \eqref{def-Th}. Because of the restrictions in the Sobolev spaces in Corollary \ref{cor-Phi-nu} and Proposition \ref{prop-Ptau}, this is not a simple count of gains and losses of regularity. However we can take advantage of the fact that we have at least one resolvent between each differentiation $\tilde \Phi_j$.

\begin{proposition} \label{prop-Th}
Let $\rho>0$ be given by \eqref{dec-metric-a}. Let $s \in \big[0,\frac d 2+1[$, $s^* \in \big]-\frac d 2-1,0\big]$, $n\in\N$ and $\m \in \N^d$.  Let $j,k \in \N$ be such that $j \leq k$, and $\a_j,\dots , \a_{k} \in \N^*$.
Then there exist $\s_0 \in ]0,\min(1,\rho)/2[$ and $C \geq 0$ such that for $\s \in [0,\s_0]$ which satisfies $2 \sum_{l=j}^k \a_{l} - (1+\s) \Vc_{j,k}  \geq (s-s^*)$ and $z \in \C_{+}$ we have
\[
 \nr{\ad_x^\m \ad_{A}^n   \tThiota_{j ; \a_{j},\dots , \a_{k} } (z) }_{\Lc(H^{s^*}, H^{s})} \leq C \abs z ^{(1+\s) \Vc_{j,k}}.
\]
Moreover this also holds when $\a_j = 0$ if $s < \frac d 2 -1$, and when $\a_k = 0$ if $s^* > -\frac d 2 +1$.
\end{proposition}

\begin{proof}
We have $s - 2 \d_{0,\a_j} < \frac d 2 - 1$ and $s^* + 2 \d_{0,\a_k} > - \frac d 2 +1$, so we can consider $\s_0 \in ]0,\min(1,\rho)/2[$ such that 
\[
S : = \max\left( 0, s - 2 \d_{0,\a_j} + 1 + \s_0 \right) < \frac d 2 \quad \text{and} \quad s^* + 2 \d_{0,\a_k} - 1 - \s_0 > - \frac d 2 .
\]
Let $\s \in [0,\s_0]$. For $p \in \Ii j k$ we set
\[
\tilde s_p = s^* + 2 \sum_{l=p}^k \a_l - (1+\s) \Vc_{p,k} \quad \text{and} \quad s_p = \min(S,\tilde s_p).
\]
Note that the sequence $(s_p)_{j\leq p \leq k}$ is non-increasing and $s_p > -\frac d 2 + (1+\s_0)$ for all $p \in \Ii j k$. The operator $\ad_x^\m \ad_{A}^n \tThiota_{j ; \a_{j},\dots , \a_{k} } (z)$ can be written as a linear combination of terms of the form
\[
 \prod_{l=j}^{k-1}  \left(\ad_x^{ \m_l} \ad_{A}^{ n_l} \big(\Ptau+1\big)^{-\a_l}  \ad_x^{\tilde \m_l} \ad_{A}^{\tilde n_l} \tilde \Phi_{l+1}(z) \right) \times \ad_x^{ \m_k} \ad_{A}^{ n_k} \big(\Ptau+1\big)^{-\a_k}
\]
where $\m_k + \sum_{l=j}^{k-1} (\tilde \m_l+\m_l) = \m$ and $n_k + \sum_{l=j}^{k-1} (\tilde n_l+n_l) = n$. 
According to Proposition \ref{prop-Ptau} we have
\[
\nr{ \ad_x^{ \m_k} \ad_{A}^{n_k} \big(\Ptau+1\big)^{-\a_k}}_{\Lc(H^{s^*},H^{s_k})} \lesssim 1.
\]
Let $p \in \Ii {j+1}{k-1}$. Since $s_{p+1}$ and $s_{p+1} - (1+\s)\n_{p+1}$ belong to $\big]-\frac d 2, \frac d 2\big[$, we have according to Corollary \ref{cor-Phi-nu} and Proposition \ref{prop-Ptau}
\begin{eqnarray*}
\lefteqn{\nr{\ad_x^{ \m_p} \ad_{A}^{ n_p} \big(\Ptau+1\big)^{-\a_p}\ad_x^{\tilde \m_{p}} \ad_{A}^{\tilde n_{p}} \tilde \Phi_{p+1}(z)}_{\Lc(H^{s_{p+1}},H^{s_{p}})}}\\ 
&& \leq 
\nr{\ad_x^{ \m_p} \ad_{A}^{ n_p} \big(\Ptau+1\big)^{-\a_p}}_{\Lc(H^{s_{p+1} - (1+\s)\n_{p+1}},H^{s_{p}})}
\nr{\ad_x^{\tilde \m_{p}} \ad_{A}^{\tilde n_{p}} \tilde \Phi_{p+1}(z)}_{\Lc(H^{s_{p+1}},H^{s_{p+1} - (1+\s)\n_{p+1}})}\\
&& 
\lesssim \abs z^{(1+\s) \n_{p+1}}.
\end{eqnarray*}
Similarly we have 
\[
\nr{\ad_x^{ \m_j} \ad_{A}^{ n_j} \big(\Ptau+1\big)^{-\a_j}\ad_x^{\tilde \m_{j}} \ad_{A}^{\tilde n_{j}} \tilde \Phi_{j+1}(z)}_{\Lc(H^{s_{j+1}},H^{s})} \lesssim \abs z^{(1+\s) \n_{j+1}}.
\]
Here we have used the assumption that $s_j \geq s$.
\end{proof}

Proposition \ref{prop-Th} will be used to estimate the terms that we called $E_z$ in Section \ref{sec-outline}. Now we go slightly further and estimate the terms called $W_z$. The main difference with the previous proposition is that we now have to add a factor $\pppg A^\d$. This is necessary to compensate the weight $\pppg A ^{-\d}$ which we need to use Mourre theory.

\begin{proposition} \label{prop-Th3}
 Let $s \in \big[\n_0,\frac d 2 +1 \big[$ be such that $s-\n_0 \neq \frac d 2$ and $\d > s$. Let $j,k\in\Ii 0 N$ and $\a_j,\dots,\a_k \in \N$ be such that $m:=\sum_{l=j}^k \a_l$ is large enough (say $m \geq \d + 2 + s$).
\begin{enumerate} [(i)]
\item If $\a_j,\dots,\a_{k-1} \in \N^*$ there exists $C \geq 0$ such that for all $z \in \C_{+}$ we have 
\[
 \nr{\pppg x ^{-\d} e^{iA \ln  \abs z} \tilde \Phi_0(z)  \tThiota_{j ; \a_j,\dots,\a_k} (z) \pppg { A}^{\d}}_{\Lc(L^2)} \leq C \abs z^{\min \left(s - \n_0 , \frac d 2\right)  + \n_0 + \Vc_{j,k}}.
\]
\item If $\a_{j+1} ,\dots,\a_{k} \in \N^*$ there exists $C \geq 0$ such that for all $z \in \C_{+}$ we have 
\[
 \nr{\pppg { A}^{\d}  \tThiota_{j ; \a_j,\dots,\a_k} (z) e^{-iA \ln  \abs z} \pppg x ^{-\d} }_{\Lc(L^2)} \leq C \abs z^{\min \left(s , \frac d 2\right) +  \Vc_{j,k}}.
\]
\end{enumerate}
\end{proposition}

\begin{proof}
\stepp 
We prove the first estimate. The second is proved similarly. Let $\tilde s =\min \left(s - \n_0 , \frac d 2\right)$. Since $\d > \tilde s$, multiplication by $\pppg x ^{-\d}$ is bounded from $L^{\frac {2d}{d-2\tilde s}}$ to $L^2$. According to Sobolev embeddings and Proposition \ref{prop-A} we have
\begin{equation} \label{eq-xAx}
\begin{aligned}
\nr{\pppg x ^{-\d} e^{iA \ln \abs z} \big(1+ \abs x^\d \big)}_{\Lc(H^{s-\n_0}, L^2)}
& \leq \nr{\pppg x ^{-\d} e^{iA \ln \abs z}}_{\Lc\big(L^{\frac {2d}{d-2\tilde s}}, L^2 \big)} +  \nr{\pppg x ^{-\d}  e^{iA \ln \abs z } \abs x^\d }_{\Lc(L^2)}\\
& \lesssim  \nr{ e^{iA \ln \abs z}}_{\Lc\big(L^{\frac {2d}{d-2\tilde s}}\big)} +  \abs z ^\d  \nr{\pppg x ^{-\d}  \abs x^\d e^{iA \ln \abs z }  }_{\Lc(L^2)}\\
& \lesssim \abs z ^{\tilde s}.
\end{aligned}
\end{equation}

\stepp For $\tilde \Phi_0(z)$ we have
\[
\nr{\pppg x^{-\d} \tilde \Phi_0(z) \pppg x ^\d}_{\Lc(H^s,H^{s-\n_0})} \lesssim \abs z ^{\n_0}.
\]

\stepp Now we prove that for all $\d \geq 0$ (in this analysis we no longer use the assumption $\d > s$):
\begin{equation} \label{stepp-Th3}
\nr{\pppg x^{-\d}   \tThiota_{j ; \a_j,\dots,\a_k} (z) \pppg { A}^{\d}}_{\Lc(L^2,H^{s})} \lesssim \abs z^{\Vc_{j,k}}.
\end{equation}
It is enough to prove this when $\d$ is an integer, and then the general case will follow by interpolation, using the following argument: if $ p $ is a fixed integer, the estimate (\ref{stepp-Th3}) implies that for some $ N \geq 0 $, we have
\begin{multline}
\nr{ \big( \langle D \rangle^{s - \nu_0} \langle x \rangle^{- i \rm{Im}(\d)}  \langle D \rangle^{ \nu_0 - s} \big)
 \langle D \rangle^{s - \nu_0} \langle x \rangle^{-  \rm{Re}(\d)}     \tThiota_{j ; \a_j,\dots,\a_k} (z) \pppg { A}^{\d}}_{\Lc(L^2)}\\
 \lesssim  (1 + |{\rm Im}(\d)|)^{N} |z|^{\Vc_{j,k}},
\end{multline}
when $ {\rm Re}(\d) = p $. Indeed, by the Calder\`on-Vaillancourt Theorem
$$ \nr{  \langle D \rangle^{s - \nu_0} \langle x \rangle^{- i \rm{Im}(\d)}  \langle D \rangle^{\nu_0-s}  }_{{\mathcal L}(L^2)} \lesssim (1 + |{\rm Im}(\d)|)^N $$
since $  \langle x \rangle^{- i \rm{Im}(\d)} $ is a zero order symbol with seminorms growing polynomially in $ {\rm Im}(\delta) $.
  Then, the result will follow from a routine argument using the Hadamard three lines theorem and the estimate (\ref{stepp-Th3}) when $ \d $ is an integer which we assume from now on.  

\stepp The operator $ \tThiota_{j ; \a_j,\dots,\a_k} (z)$ can be rewritten as
\[
\tThiota^b_{ j_1,\dots,j_m} (z) : =  (\Ptau +1)\inv  b_{j_1,\abs z}(x)  (\Ptau +1)\inv  b_{j_2,\abs z}(x) \dots  (\Ptau +1)\inv  b_{j_m,\abs z}(x) 
\]
where $j_1,\dots,j_m \in \{0,1\}$, $b_0(x) = 1$, $b_1(x) = a_\iota (x)$ and $\sum_{l=1}^m j_l = \Vc_{j,k}$. We prove by induction on $\d \in \N$ that if $m \geq \d + 1 + s$ then for all $\m \in \N^d$ we have 
\begin{equation} \label{stepp-Th3-mu}
\nr{\pppg x^{-\d} \ad_x^\m \left(  \tThiota_{j ; \a_j,\dots,\a_k} (z) \pppg { A}^{\d} \right)}_{\Lc(L^2,H^{s-\n_0})} \lesssim \abs z^{ j_1+\dots + j_m}.
\end{equation}
Note that \eqref{stepp-Th3-mu} gives \eqref{stepp-Th3} when $\m = 0$. 
%
%
%
%
%
Statement \eqref{stepp-Th3-mu} is a consequence of Corollary \ref{cor-Phi-nu} and Proposition \ref{prop-Th} when $\d = 0$. Let us consider the general case.
We have
\[
\tThiota^b_{ j_1,\dots,j_m} (z) A^{\d} = \sum_{l=0}^{{\d}} C_{{\d}}^l \tThiota^b_{ j_1,\dots,j_{m-1}} A^{{\d}-l} \ad_{A}^l  \big( (\Ptau +1)\inv b_{j_m,\abs z}(x) \big).
\]
When $l\neq 0$ we can apply the inductive assumption to $ \tThiota^b_{ j_1,\dots,j_{m-1}} A^{{\d}-l}$. With Proposition \ref{prop-Th} this proves that the corresponding term in the right-hand side satisfies estimate \eqref{stepp-Th3-mu}.
We now consider the term corresponding to $l = 0$. More precisely it is enough to consider
\begin{equation} \label{stepp-Th3-xq}
\tThiota^b_{ j_1,\dots,j_{m-1}} A^{{\d}-1} x_q D_q  (\Ptau +1)\inv b_{j_m,\abs z}(x) 
\end{equation}
for some $q \in \Ii 1 d$. The operator $ D_q (\Ptau +1)\inv b_{j_m}(x) $ and its commutators with $x$ are of size $\abs z^{j_m}$ as operators on $L^2$. On the other hand
\[
\tThiota^b_{ j_1,\dots,j_{m-1}} A^{{\d}-1} x_q = x_q  \tThiota^b_{ j_1,\dots,j_{m-1}} A^{{\d}-1} + \ad_{x_q} \left( \tThiota^b_{ j_1,\dots,j_{m-1}} A^{{\d}-1} \right)
\]
and hence, according to the inductive assumption, the term \eqref{stepp-Th3-xq} also satisfies \eqref{stepp-Th3-mu}.
%
%
This gives the inductive step, proves \eqref{stepp-Th3-mu}, and concludes the proof of the proposition.
\end{proof}

Concerning the weight we need in Proposition \ref{prop-Th3} and hence in Theorem \ref{th-low-freq}, we emphasize that in one case we need a weight $\pppg x ^{-\d}$ to map $L^{\frac {2d}{d - 2\tilde s}}$ to $L^2$ and in the second case we need it to compensate the factor $\abs x ^\d$ which comes from $A^\d$, but we do not have to do both at the same time.

\subsection{Low frequency estimates for a small perturbation of the Laplacian}  \label{sec-low-freq}

After the preliminary work of Section \ref{sec-dec-functions}, we are able to apply the general theory of Section \ref{sec-Mourre} to the rescaled resolvent $\tRtau$ and prove the estimates of Theorem \ref{th-low-freq-R0}.

\begin{proposition} \label{prop-estim-res-amort}
Let $n\in\N$ and $\d > n + \frac 12$. If $\Nc_1$ is small enough there exists $C \geq 0$ such that for all $j,k \in \N$ with $j \leq k$ and $k-j \leq n+1$ we have
\[
\forall z \in \C_+, \quad  \nr{\pppg A^{-\d} \tilde \Phi_0(z) \tilde \Rc^\iota_{j,k}(z) \pppg A^{-\d}}_{\Lc(L^2)} \leq C \abs z ^{\n_0 + \Vc_{j,k}}.
\]
\end{proposition}

\begin{proof} 
Let $J\subset \R_+^*$ be a compact neighborhood of 1. We have to prove that $A$ is uniformly conjugate to $\Ptau$ on $J$ with constant lower bound, that $\Ptau$ is uniformly $(n+1)$-smooth with respect to $A$ and that for all $k \in \N$ there exists $C\geq 0$ such that for all $z \in \C_{+}$ the pair $(\tilde \Phi_{k}(z),\Id)$ belongs to $\opinsert_{n+1}(\Ptau , A)$ with
\[
\nr{\tilde \Phi_k(z)}_{\Ptau , A,n+1} \leq C \abs z ^{\n_k}
\]
(see Definition \ref{def-415}). This last statement is a direct consequence of Corollary \ref{cor-Phi-nu} and Proposition \ref{prop-Ptau}.
Assumptions ($a$) and ($b$) of  Definition \ref{defconjunif} can be checked as for intermediate frequencies, and assumptions $(c_N)$ and $(d_N)$ of Definition \ref{defconjunifdissn} are consequences of Proposition \ref{prop-Ptau-comm-Hs}.
Let us now check the main assumption of Definition \ref{defconjunif}. Let $\Ptaur = \Ptauo + \Im(\hat z) \abs z \inv a_\iota$ be the real part of $\Ptau$. 
Using notation \eqref{gamma-m} we have for $m \in \{0,1\}$ and $u \in \Sc(\R^d)$
\begin{align*}
\innp{ \ad_{iA}^m\big(\Ptauo\big)u}{u}_{L^2}- 2^m\nr{\nabla u}_{L^2}^2 =  \sum_{j,k} \innp{\big(\g_{j,k, \abs z}^{(m)}- 2^m \d_{j,k}\big) D_k u}{D_j u}  \gtrsim  - \Nc_m  \nr{\nabla u}^2.
\end{align*}
This proves that if $\Nc_1$ is small enough we have
\begin{align*}
 \ad_{iA}(\Ptauo)  \geq -\D \geq \frac 12 \Ptauo
\end{align*}
in the sense of quadratic forms on $\Sc(\R^d)$. It remains true on $H^2$ which is the domain of the closures of these operators.
On the other hand for $m \in \{0,1\}$ we have
$
\Im(\hat z) \abs z \inv \nr{a^{(m)}_{\iota, \abs z} u}_{L^2} \leq C \Nc_m \nr u_{H^1},
$
so if $\Nc_1$ is small enough we obtain
\begin{align*}
\1 J \big(\Ptaur\big) \ad_{iA}\big(\Ptaur\big) \1 J \big(\Ptaur\big)
& \geq \frac 14 \1 J \big(\Ptaur\big)  \Ptauo  \1 J \big(\Ptaur\big) \\
& \geq \frac 18 \1 J \big(\Ptaur\big)  \Ptaur \1 J \big(\Ptaur\big) \\
& \geq \frac {\inf J} 8 \1 J \big(\Ptaur\big).
\end{align*}
Thus we can apply Theorem \ref{th-estim-insert2} to obtain the result for $z \in \C_{+,+}$. We can similarly prove an analogous result with inserted factors in reversed order and take the adjoint to get the result for $z \in \C_{-,+}$. It only remains to remark that the result is clear when $\Re (z) = 0$.
\end{proof}

\begin{proposition} \label{prop-estim-tildeRc2}
Let $n \in \N$ and $\e > 0$. 
\begin{enumerate} [(i)]
\item Let $\d$ be greater than $n + \frac 12$ if $n \geq \frac d 2$ and greater than $n + 1$ otherwise. Then there exists $C \geq 0$ such that for all $z \in \C_ {+}$ we have
\[
\nr{\pppg x ^{-\d}   \Phi_0\Rc^\iota_{0,n} (z) \pppg x ^{-\d}}_{\Lc(L^2)} \leq C \left( 1 + \abs{z}^{d-\e-2(n+1) + \n_0 + \Vc_{0,n}}\right). 
\]
\item Assume that $2n \neq d-2$ or $\Vc_{0,n} \neq 0$. Let $\d_1,\d_2 > n + \frac 12$ be such that $\d_1 + \d_2 > \min(2(n+1), d)$. Then there exists $C \geq 0$ such that for all $z \in \C_ {+}$ we have
\[
\nr{\pppg x ^{-\d_1} \Rc^\iota_{0,n} (z) \pppg x ^{-\d_2}}_{\Lc(L^2)} \leq C \left( 1 + \abs{z}^{d-2(n+1) + \Vc_{0,n}}\right).
\]
\end{enumerate}
\end{proposition}

%

\begin{proof}
We prove both statements at the same time, using the notation $\d_1 = \d_2 = \d$ for the first case. We have $\n_0 = 0$ and $\Phi_0 = 1$ in the second case.
We recall that for all $z \in \C_{+}$ we have
\begin{equation} \label{Rc-tRc}
\pppg x ^{-\d_1}  \Phi_0 \Rc^\iota_{0,n} (z)\pppg x ^{-\d_2} = \abs z^{-2(n+1)} \pppg x ^{-\d_1} e^{iA \ln \abs z}  \tilde \Phi_0(z) \tilde \Rc^\iota_{0,n} (z) e^{-iA \ln \abs z}  \pppg x ^{-\d_2}.
\end{equation}
 As we did for $\Rc_{0,n}(z)$ in the proof of Proposition \ref{prop-inter-freq}, we can check by induction on $m \in \N$ that $  \tilde \Phi_0(z)\tilde \Rc^\iota_{0,n}(z)$ can be written as a sum of terms either of the form $\big( 1 + \hat z ^2 \big)^\b  \tilde \Phi_0(z)  \tThiota_{0 ; \a_0,\dots,\a_n}(z)$ with $\b \in \N$, $\a_0,\dots, \a_n \in \N^*$, or 
\begin{equation} \label{terme2p}
\big( 1 + \hat z ^2 \big)^\b  \tilde \Phi_0(z)  \tThiota_{0 ; \a_0,\dots,\a_j} (z)  \tilde \Rc^\iota_{j,k}(z)  \tThiota_{k ; \a_{k},\dots , \a_n}(z) 
\end{equation}
where $\b \in \N$,  $j,k\in\Ii 0 n$, $j\leq k$, $\a_0,\dots, \a_{j-1} , \a_{k+1}, \dots , \a_n \in \N^*$, $\a_j,\a_{k} \in\N$, 
$
\sum_{l=0}^j \a_l \geq m$ and $\sum_{l=k}^{n} \a_l \geq m. 
$

\stepp
Let $s_1 \in \big[0, \frac d 2 + 1 \big[ \setminus \big\{ \frac d 2 \big\}$, $\tilde s_1 = \min \big( s_1-\n_0 , \frac d 2 \big)$, $s_2 = \tilde s_2 \in \big[0, \frac d 2 \big[$, and assume that $\d_1 > \tilde s_1$ and $\d_2 > \tilde s_2$. We have $H^{s_1-\n_0} \subset L^{\frac {2d}{d-2\tilde s_1}}$ and $L^{\frac {2d}{d+2\tilde s_2}} \subset H^{- s_2}$ with continuous embeddings. Moreover multiplication by $\pppg x ^{-\d_1}$ is bounded from $L^{\frac{2d}{d-2 \tilde s_1}}$ to $L^2$ and multiplication by $\pppg x ^{-\d_2}$ is bounded from $L^2$ to $L^{\frac {2d}{d+2 \tilde s_2}}$.
Let $\s_0$ be given by Proposition \ref{prop-Th} and $\s \in [0,\s_0]$ (we take $\s = 0$ if $\n_0 \neq 0$). Assume that $s_1+s_2 \leq 2(n+1) - (1+\s) \Vc_{0,n}$. According to Proposition \ref{prop-A} and Proposition \ref{prop-Th} we have
\begin{eqnarray*}
\lefteqn{ \nr{\pppg x ^{-\d_1} e^{iA \ln \abs z}  \tilde \Phi_0(z)  \tThiota_{0 ; \a_0,\dots,\a_n}(z) e^{-iA \ln  \abs z}  \pppg x ^{-\d_2}}_{\Lc(L^2)} }\\
&& \leq 
\nr{e^{iA \ln \abs z}}_{\Lc\big(L^{\frac {2d}{d-2 \tilde s_1}}\big)} 
\nr{ \tilde \Phi_0(z)}_{\Lc(H^{s_1},H^{s_1-\n_0})} \nr{\tThiota_{0 ; \a_0,\dots,\a_n}(z)}_{\Lc(H^{-s_2},H^{s_1})} 
\nr{ e^{-iA \ln \abs z}}_{\Lc\big(L^{\frac {2d}{d+2 \tilde s_2}}\big)} 
\\
&& \leq C \abs z ^{ \tilde s_1 + \tilde s_2  + \n_0 + (1+\s) \Vc_{0,n}}.
\end{eqnarray*}

\stepp 
We now consider a term of the form \eqref{terme2p} with $m$ large enough. 
According to Proposition \ref{prop-estim-res-amort} we have
\[
 \nr{\pppg A ^{-\d_1}  \tilde \Rc^\iota_{j,k}(z)   \pppg A ^{-\d_2}}_{\Lc(L^2)} \lesssim  \abs z ^{\Vc_{j,k}}.
\]
Then we apply Proposition \ref{prop-Th3} and obtain 
\begin{multline*}
{ \nr{\pppg x ^{-\d_1} e^{iA \ln  \abs z}  \tilde \Phi_0(z)  \tThiota_{0 ; \a_0,\dots,\a_j} (z)  \tilde \Rc^\iota_{j,k}(z)  \tThiota_{k ; \a_{k},\dots , \a_n}(z) e^{-iA \ln  \abs z}  \pppg x ^{-\d_2}}_{\Lc(L^2)}}\\
 \lesssim\abs z^{\tilde s_1  +\n_0 + \Vc_{0,j}} \abs z ^{\Vc_{j,k}} \abs z ^{\tilde s_2  + \Vc_{k,n}} \lesssim \abs {z}^{\tilde s_1 + \tilde s_2 + \n_0 +  \Vc_{0,n}}.
\end{multline*}

\stepp
Let us now choose $\s$, $s_1$ and $s_2$ more precisely. In case (i) we set $\s = 0$,
\[
s_2 = \min \left(n +1 - \frac {\n_0 + \Vc_{0,n} }2, \frac {d-\e} 2 \right), \quad \text{and} \quad s_1 = s_2 + \n_0.
\]
Then the conditions on $s_1$ and $s_2$ are satisfied and we have 
\[
\tilde s_1 + \tilde s_2 + \n_0 +  \Vc_{0,n} = \min\big(2(n+1), d-\e +\n_0 +\Vc_{0,n}\big)
\]
which, together with \eqref{Rc-tRc}, gives the first statement of the proposition.

\stepp
In case (ii) we set $\s = 0$ if $2n+2 - \Vc_{0,n} \neq d$ and we choose any $\s \in ] 0, \s_0]$ otherwise, so that $2n+2 - (1+\s) \Vc_{0,n} \neq  d$ in any case. Then, if $2n+2 - (1+\s) \Vc_{0,n} < d$ we choose $s_1 \in \big[0,\min\big(\d_1, \frac d 2 \big) \big[$ and $s_2 \in \big[0,\min\big(\d_2, \frac d 2 \big) \big[$ such that $s_1 + s_2 = 2n+2 - (1+\s) \Vc_{0,n}$. If $2n+2 - (1+\s) \Vc_{0,n} > d$ we consider 
\[
s_1 = s_2 \in \left] \frac d 2 , \min \left( \d_1,\d_2, \frac d 2 +1 , \frac {2n + 2 - (1+\s) \Vc_{0,n}} 2 \right) \right[.
\]
Thus we have $\tilde s_1 + \tilde s_2 + (1+\s) \Vc_{0,n} = \min\big(2(n+1), d +\Vc_{0,n}\big)$ and statement (ii) is proved.
\end{proof}

Finally we can prove Theorem \ref{th-low-freq-R0}:

\begin{proof} [Proof of Theorem \ref{th-low-freq-R0}]
Let us write $\Phi_0 R^{(m)}(z)$ as a linear combination of terms as given by Proposition \ref{prop-der-R2}. We consider such a term $\Phi_0  T(z)$ and use the notation of Proposition \ref{prop-der-R2}. $\Phi_0 T(z)$ is then of the form $z^k \Phi_0 \Rc_{0,n}(z)$ with $\Vc_{0,n} = j_1+\dots+j_n$. According to Proposition \ref{prop-estim-tildeRc2} we have
\begin{align*}
\nr{ \pppg x ^{-\d}  \Phi_0 T(z) \pppg x ^{-\d} }_{\Lc(L^2)}  \lesssim \abs z^k \left( 1 + \abs z^{d-\e - 2(n+1) +\n_0 + \Vc_{0,n}} \right) \lesssim 1 + \abs z^{m-2+\n_0-\e}.
\end{align*}
For the second part we have $\n_0 = 0$. If $2(n+1) \neq  d$ or $\Vc_{0,n} \neq 0$, then we apply the second part of Proposition \ref{prop-estim-tildeRc2} to conclude. If $2(n+1) = d$ and $\Vc_{0,n} = 0$ then $m \neq d-2$ (since $d$ is even) and hence $k > 0$. Then as above
\begin{align*}
\nr{ \pppg x ^{-\d}  T(z) \pppg x ^{-\d} }_{\Lc(L^2)}  \lesssim \abs z^k \left( 1 + \abs z^{d-\e - 2(n+1)} \right) \lesssim \abs z^k \left( 1 + \abs z^{-\e} \right) \lesssim \abs z^{k-\e} \lesssim 1.
\end{align*}
%
%
%
%
%
%
%
%
%
\end{proof}

\subsection{General long-range perturbations} \label{sec-non-small}

We now use Theorem \ref{th-low-freq-R0} to prove Theorems \ref{th-low-freq} and \ref{th-low-freq-bis}. Here we basically follow the same strategy as in \cite{bouclet11} but this approach has to be modified since we are dealing with non self-adjoint operators.\\

For $z \in \C_{+}$ and $\p \in C_0^\infty(\R)$ we set
\[
S_\p (z) =  K_z \RPz \p(\Ho ).
\]
According to the resolvent equation 
\[
R(z) = R_\iota(z)  - R(z) K_z R_\iota(z)
\]
we have
\begin{equation} \label{res-eq-S}
R(z) \p(\Ho) = \RPz \p(\Ho) - R(z) S_\p(z).
\end{equation}

\begin{proposition} \label{prop-BS}
Let $\p \in C_0^\infty(\R,[0,1])$, $n \in \N$ and $\e > 0$.
\begin{enumerate}[(i)]
\item
Let $\d$ be greater than $n + \frac 12$ if $n \geq \frac d 2$ and greater than $n + 1$ otherwise, and $M \geq 0$. Then there exists $C \geq 0$ such that for all $z \in  \C_{+}$ with $\abs z \leq 1$ we have
\begin{eqnarray}
 \nr{ \pppg x ^M S_\p^{(n)}(z) \pppg x ^{-\d}} _{\Lc(L^2,H^{-1})} \leq C \left( 1 + \abs z ^{d-1-n-\e} \right). \label{borne-importante}
\end{eqnarray}
\item
Assume that $d$ is odd or $n\neq d-2$. Let $\d_2 > n + \frac 12$. Then there exists $C \geq 0$ such that for all $z \in \C_ {+}$ with $\abs z \leq 1$ we have
\[
\nr{ \pppg x ^M S_\p^{(n)}(z) \pppg x ^{-\d_2}} _{\Lc(L^2,H^{-1})} \leq C \left( 1 + \abs{z}^{d-2-n}\right).
\]
\end{enumerate}
\end{proposition}

We will see later the importance in (\ref{borne-importante}) of having the estimates  in term of $ \abs z ^{d-1-n-\e} $ rather than $ \abs z ^{d-2-n-\e} $. The basic reason why we have such bounds is that, in the expression of $ K_z $ defined in (\ref{defKz}), one term carries an additional power of $z$ and the other one carries derivatives $ D_k $ which allows to use the estimates \eqref{eq-low-freq3} with $\abs \a = 1$.


\begin{proof}
The result follows from Theorem \ref{th-low-freq-R0} and the boundedness of $\p(\Ho)$ as an operator on $L^{2,\d}$ for all $\d \in \R$. The reason why we obtain estimates in the $ {\mathcal L} (L^2,H^{-1}) $ topology is due to the fact we see $ \langle x \rangle^{M} D_j \gamma_{j,k}^0 D_k R_{\iota} (z) \psi (H_0) $ as bounded from $ L^{2,-\delta} $ to $ H^{-1} $ (rather than $ L^2 $) because of the derivatives $ D_j $ in the expression of $ K_0 $ (see (\ref{defKz})).
\end{proof}

Then we prove that if $\p$ is well-chosen then $S_\p(z)$ is in fact uniformly small (for $\abs z$ small) in some suitable sense:

\begin{proposition} \label{prop-b589}
Let $\e_1 > 0$, $\d > 2$ and $M \geq 0$. There exist a bounded neighborhood $\Uc$ of $0$ in $\C$ and $\p  \in C_0^\infty(\R)$ equal to 1 in a neighborhood of $0$ such that for all $z \in \Uc\cap \C_{+} $ we have
\[
\nr{ \pppg x ^M  S_\p (z) \pppg x ^{-\d}}_{\Lc(L^2)} \leq \e_1.
\]
\end{proposition}

\begin{proof} 
\stepp
Let $\t \in \R$, $\m > 0$ and $z = \t + i \m$. We can write
\begin{equation} \label{decomp-S}
\begin{aligned}
\pppg x ^M  S_\p (z) \pppg x ^{-\d}
& = \pppg x ^M  K_0 \RP(i\m)  \p(\Ho )  \pppg x ^{-\d}
   +\pppg x ^M  K_0 \big( \RPz - \RP( i\m)  \big)\p(\Ho ) \pppg x ^{-\d} \\
& \quad -i z \pppg x ^M  a_0 \RPz  \p(\Ho )  \pppg x ^{-\d}
\end{aligned}
\end{equation}
and estimate each term of the right-hand side.
%
%

\stepp
Let us estimate the first term. According to the Hardy inequality we have for $u \in \Sc(\R^d)$
\[
\nr{\pppg x ^M K_0 u }_{L^2} \lesssim \sum_{j,k=1}^d \nr{\pppg x ^M (D_j \g^0_{j,k}) D_k u}_{L^2} + \nr{\pppg x ^M \g^0_{j,k} D_j  D_k u}_{L^2} \lesssim \nr{u }_{\dot H ^2}.  
\]
Thanks to Remark \ref{rem-lap-PO} we get
\[
\nr{\pppg x ^M K_0 u }_{L^2} \lesssim  \nr{ \Poo u}_{L^2}.
\]
Since for all $\m \in ]0,1]$ we also have
\[
 \Poo \RP(i\m) = 1 - \m a_\iota  \RP(i\m) - \m^2 \RP(i\m)
\]
we obtain
\begin{eqnarray*}
 \lefteqn{ \nr{\pppg x ^M K_0 \RP(i\m)  \p(\Ho) \pppg x ^{-\d}}_{\Lc(L^2)}  \lesssim \nr{\Poo \RP(i\m)  \p(\Ho) \pppg x ^{-\d}} }\\
&& \lesssim \nr{\p(\Ho) \pppg x ^{-\d}} +   \m  \nr{ a_\iota \RP(i\m)  \p(\Ho) \pppg x ^{-\d}} +  \m^2 \nr{ \RP(i\m)  \p(\Ho) \pppg x ^{-\d}}\\
&& \lesssim \nr{\p(\Ho) \pppg x ^{-\d}} \left( 1  + \m^2 \nr{\RP(i\m)} \right) + \m  \nr{ \sqrt {a_\iota} \RP(i\m) \pppg x^{-\d}}
%
\end{eqnarray*}
We have
\[
 \m^2 \nr{\RP(i\m)} \leq 1,
\]
and according to Proposition \ref{prop-res-diss-am} applied to $\RP$ and Theorem \ref{th-low-freq-R0}:
\[
 \m \nr{ \sqrt {a_\iota} \RP(i\m) \pppg x ^{-\d}} \leq  \sqrt \m \nr{\pppg x ^{-\d}\RP(i\m)\pppg x ^{-\d}}^{\frac 12} \lesssim \sqrt \m,
\]
so
\[
  \nr{\pppg x ^M K_0 \RP(i\m)  \p(\Ho) \pppg x ^{-\d}}_{\Lc(L^2)} \lesssim \nr{\p(\Ho) \pppg x ^{-\d}} _{\Lc(L^2)} + \sqrt \m.
\]
Since $0$ is not an eigenvalue of $\Ho$, the right-hand side goes to 0 if the support of $\p$ shrinks to $\singl {0}$, and hence the left-hand side is less than $\e_1 / 3$ if $\p$ is well-chosen and $\m$ is small enough.

\stepp According to Theorem \ref{th-low-freq-R0}, and since $a_0$ is compactly supported, the third term can indeed be made as small as we wish if $\abs z$ is chosen small.

\stepp
It remains to estimate the second term of \eqref{decomp-S}. Since the operators $\pppg x ^{\d} \p(\Ho ) \pppg x ^{-\d}$ and $\pppg x ^M K_0 (P_{i\m} + 1)\inv \pppg x ^\d$ have bounded closures on $L^2$ (whose norms are uniform in $\m \in ]0,1]$), we only have to estimate
\[
\pppg x ^{-\d} (P_{i\m} +1) \big( \RPz - \RP(i\m)  \big) \pppg x^{-\d}
\]
in $\Lc(L^2)$ to conclude. We have
\begin{eqnarray*}
\lefteqn{(P_{i\m} +1) \big( \RPz - \RP(i\m)  \big)}\\
& \hspace {1cm} & = 1  + (1+z^2) \RPz +  i\t  a_\iota  \RPz - 1 - (1 - \m^2 ) \RP(i\m) \\
&& = (1 + z^2) \int _0^\t \frac d {ds }  \RP(s+i\m) \, ds +   (z^2 + \m^2) \RP(i\m) + i \t a_\iota   \RPz ,
\end{eqnarray*}
and hence we can conclude with Theorem \ref{th-low-freq-R0}.
\end{proof}

Now we can prove low frequency resolvent estimates in the general setting:

\begin{proof} [Proof of Theorems \ref{th-low-freq} and \ref{th-low-freq-bis}]
For $z \in \C_{+}$ we set $B(z) = \big( \Ho - z^2\big)\inv (1-\p)(\Ho)$. For any $\s \in \R$, the function $B$ and all its derivatives are bounded on $L^{2,\s}(\R^d)$ uniformly in $z \in \C_{+}$ close enough to 0.
Let $n\in\N$. If $n >0$ we assume that the first estimate of Theorem \ref{th-low-freq} is proved for all $m \in \Ii 0 {n-1}$ and we proceed by induction. Let $\s > n+ \frac 12$ if $n \geq \frac d 2$ and $\s > \max \big(n+  1 ,2 \big)$ otherwise. Let $\e_1 >0$. Let $\Uc$ and $\p$ be given by Proposition \ref{prop-b589} applied with $M= \s$. We can assume that $\bar \Uc \cap \supp (1-\p) = \emptyset$. 
According to \eqref{res-eq-S} and the resolvent identity between $\Ho$ and $\Hz$ we have
\begin{equation} \label{eq-R-RPz}
\begin{aligned}
R^{(n)}(z) 
& = \frac {d^n} {dz^n}  \big( \RPz \p(\Ho) - R(z) S_\p(z)  + B(z)+ iz R(z) a B(z) \big) \\
& = R_\iota ^{(n)}(z) \p(\Ho) - R^{(n)}(z) S_\p(z) + iz R^{(n)}(z) a B(z) + B^{(n)}(z) \\
& \quad + \sum_{k=0}^{n-1} C_n^k R^{(k)}(z) \frac {d^{n-k}}{dz^{n-k}} \big( S_\p(z) + izaB(z)\big)
\end{aligned}
\end{equation}
Here we observe that bounds on 
$$  R^{(k)}(z) \frac {d^{n-k}}{dz^{n-k}}  S_\p(z) =  R^{(k)}(z) \langle x \rangle^{-\delta} \langle x \rangle^{\delta} \frac {d^{n-k}}{dz^{n-k}}  S_\p(z) $$
 will rest on Proposition \ref{prop-BS} after the simple observation that $ L^2 \rightarrow L^2 $ estimates on $  R^{(k)}(z) \langle x \rangle^{-\delta} $, as $ z \rightarrow 0 $, can be easily be converted in $ H^{-1} \rightarrow L^2 $ estimates by using the resolvent identity to write
$$ R(z) \langle x \rangle^{-\delta} = R(i) \langle x \rangle^{-\delta} + R(z) \langle x \rangle^{-\delta} (1+a+z^2 + i z a) \left( \langle  x \rangle^{\delta} R(i) \langle x  \rangle^{-\delta} \right) $$
where the first term and the last bracket in the right hand side are bounded from $ H^{-1} $ to $ L^2 $ (and even $ H^1 $) by standard elliptic regularity.
Then, according to Theorem \ref{th-low-freq-R0}, Proposition \ref{prop-b589}, Proposition \ref{prop-BS} and the inductive assumption, there exists $C \geq 0$ such that for all $z \in \Uc \cap \C_{+}$ we have
\[
\nr{\pppg x ^{-\d}  R^{(n)}(z) \pppg x ^{-\s} } \leq C \left( 1 + \abs z ^{d-2- n- \e} \right) + C (\e_1 + \abs z) \nr{\pppg x ^{-\d}  R^{(n)}(z) \pppg x ^{-\s} }  . 
\]
If $\e_1$ and $\abs z$ are small enough we get
\[
\nr{\pppg x ^{-\d}  R^{(n)}(z) \pppg x ^{-\s}} \leq  C\left( 1 + \abs z ^{d-2- n- \e} \right) . 
\]
Then, according to \eqref{res-eq-S} and Proposition \ref{prop-BS}:
\begin{eqnarray*}
\lefteqn{\nr{\pppg x ^{-\d} R^{(n)}(z) \pppg x ^{-\d}} (1 - C \abs z) }\\
&& \leq C \left( 1 + \abs z ^{d-2- n- \e} \right) + C \nr{\pppg x ^{-\d}  R^{(n)}(z) \pppg x ^{-\s}} \nr{\pppg x^{\s} S_\p(z) \pppg x ^{-\d}}\\
&& \leq C\left( 1 + \abs z ^{d-2- n- \e} \right),
\end{eqnarray*}
which gives the result for $\abs z$ small enough. Theorem \ref{th-low-freq-bis} and the second estimate of Theorem \ref{th-low-freq} are proved similarly by using crucially in the latter case the bounds (\ref{borne-importante}).
\end{proof}

\section{High frequency estimates} \label{sec-high-freq}

Let us now discuss high frequency estimates. For $h > 0$ and $\s \in \C_{+}$ we set 
\begin{equation} \label{def-Hhs}
\hhs = h^2 \Ho -ih \s a(x) .
\end{equation}
For $z \in \C_+$, $h = \abs z \inv$ and $\s = hz$ we have
\begin{equation} \label{R-hhs}
 R(z) = h^2 (\hhs - \s^2)\inv.
\end{equation}

%
%
%
%
%
%

To prove Theorem \ref{th-high-freq} we use again the uniform and dissipative version of Mourre's commutators method developed in Section \ref{sec-Mourre}. We are now in a semiclassical setting, and the proof relies on semiclassical pseudo-differential calculus. We recall that for a suitable symbol $q$ on the phase space $\R^{2d} \simeq T^* \R^d$ the pseudo-differential operator $\Opw(q)$ is defined for $u \in \Sc(\R^d)$ and $x \in \R^d$ by
\[
\Opw(q) u (x) = \frac 1 {(2 \pi h)^d} \int_{\R^d} \int_{\R^d} e^{-\frac ih \innp{x-y}\x} q\left( \frac {x+y} 2, \x \right) u(y) \, dy\,d\x.
\]
In particular the semiclassical generator of dilations
\[
 A_h = -\frac {i h}2 \big(x \cdot \nabla + \nabla \cdot x\big)
\]
is the quantization of the symbol $f_0 : (x,\x) \mapsto  \innp x  \x _{\R^d}$ and the principal symbol of $h^2 \Ho$ is
\[
 p : (x,\x) \mapsto \innp{G(x) \x } \x_{\R^d}.
\]
We refer to \cite{zworski, robert, martinez,dimassis} for detailed presentations of semiclassical analysis.\\

According to Proposition \ref{prop-R-diss} it is sufficient to prove Theorem \ref{th-high-freq} for $\Im z \in ]0,1]$, and as before (see the proof of Proposition \ref{prop-inter-freq}) it suffices to consider the case where $z$ (and hence $\s$) belongs to $\C_{+,+}$. According to Proposition \ref{prop-der-R2}, it will be a consequence of Theorem \ref{th-estim-insert2} if we prove that $A_h$ is uniformly conjugated to $\hhs$ on a neighborhood of 1 with lower bound of size $c_0 h$ for some $c_0 > 0$, if $\hhs$ is uniformly $N$-smooth with respect to $A_h$ for any $N\in\N$, and if moreover the multiplication by $a$ is in $\opinsert_N(A_h)$ uniformly in $h$ (see Definition \ref{def-415}).\\

For $w \in \R^{2d}$ we denote by $\vf^t(w) = \big(X (t,w), \Xi (t,w)\big)$, $t\in\R$, the solution of the Hamiltonian equations generated by the symbol $p$ with initial condition $w$:
\[
 \vf^0(w) = w, \quad \partial_t X(t,w) =  \nabla_\x p \big(\vf^t(w)\big)  \quad \text{and} \quad \partial_t \Xi(t,w) = - \nabla_x p \big(\vf^t(w)\big).
\]
In this particular case $\singl{\vf^t(w), t\in\R}$ is also the geodesic for the metric $G(x)\inv$ starting from $w$. We also recall that $p$ is preserved by the flow. Moreover for any $q \in C^\infty(\R^{2d})$ the Poisson bracket $\{p,q\} = \nabla_\x p\cdot \nabla_x q - \nabla_x p \cdot \nabla _\x q$ is the derivative $\restr{\partial_t (q \circ \vf^t)}{t=0}$ of $q$ along the flow.\\

Let us also introduce the forward and backward trapped sets    
\[
 \O_b^\pm = \singl{ w \in \R^{2d} \st \sup_{t \geq 0} \abs{X(\pm t,w)} < \infty    }
\]
and the forward and backward non-trapped sets
\[
\O_\infty^\pm = \singl{ w \in \R^{2d} \st  \abs{ X(\pm t,w)} \limt {t} {+\infty} +\infty}.
\]
Then the trapped and non-trapped sets are respectively defined by
\[
 \O_b =\O_b^- \cap \O_b^+ \quad \text{and} \quad \O_\infty =\O_\infty^- \cap \O_\infty^+.
\]
For $I \subset \R$ we also define $\O_b^\pm(I) = \O_b^\pm \cap p\inv(I)$. The sets $\O_b(I), \O_\infty^\pm(I)$ and $\O_\infty(I)$ are defined similarly. Although it is not clear from the definitions, it turns out that a classical trajectory is either trapped (bounded) or non-trapped in the future. The same holds for negative times. This will be the meaning of Proposition \ref{prop-flow} below.\\

 We recall that the geodesic flow is said to be non-trapping if $\O_b(\R_+^*) = \emptyset$, and we say that every bounded geodesic goes through the damping region (or that we have geometric control) if
\begin{equation} \label{hyp-amort}
\forall w \in \O_b (\R_+^*) , \exists t \in \R, \quad a (X (t,w)) > 0.
\end{equation}

\bigskip

The idea to prove Theorem \ref{th-high-freq} is close to that of Theorem 4.2 in \cite{art-mourre}. We review the proof since we consider here a geodesic flow, everything has to be uniform in $\s$ and there are the factors $a(x)$ between the resolvents. We also correct a mistake of the first proof (about trajectories in $\O_b^\pm \setminus \O_b$, see the proof of Lemma \ref{lem-esc-non-trapped}). \\

As in \cite{gerardm88} the proof of Theorem \ref{th-high-freq} relies on the construction of an escape function, whose quantization provides a conjugate operator for the Schr\"odinger operator. For high frequencies we really use the generalized version of Mourre estimate \eqref{hyp-mourre}: we only need a symbol which is increasing along the flow outside the damping region.

\begin{proposition}[Construction of an escape function] \label{prop-escape}
Let $I$ be a compact subset of $\R_+^*$. Then there exist $c_0 > 0$, $f_c \in C_0^\infty(\R^{2d})$ and $\b \geq 0$ such that 
\[
 \{ p, f_0 + f_c \} + \b a \geq 4 c_0 \quad \text{on }p\inv(I).
\]
\end{proposition}

We recall that $f_0$ is the symbol of the generator of dilations. As in \cite{gerardm88} we can check that $f_0$ is an escape function far from the origin. We can also use the idea of Ch. G\'erard and A. Martinez to construct a symbol which is an escape function on any compact subset of $\O_\infty$. However we may have problems at the boundary of $\O_\infty$, where some non-trapped trajectories may escape very slowly. We circumvent this difficulty by constructing a generalized escape function on a neighborhood of any compact subset of $\O_b^+ \cup \O_b^-$. For this we use Proposition \ref{prop-semi-amort}. More precisely for any $w \in \O_b^+ \cup \O_b^-$ we construct a function which is increasing along the flow around $w$ and which is non-decreasing outside the damping region. Adding a suitable multiple of $a$ we obtain the required positivity. 
The proof of Proposition \ref{prop-escape} is based on several lemmas.\\


\begin{lemma} \label{lem-flow}
There exists $\Rc_G > 0$ such that if
\begin{equation}\label{hyp-flow}
 \abs{X (\pm t_0,x,\x)} \geq \max \big(\Rc_G , \abs{x} + \g \big)
\end{equation}
for some $(x,\x) \in \R^{2d}$, $\g > 0$ and $t_0 > 0$, then this holds for $t_0$ replaced by any $t \geq t_0$ and moreover
\[
\abs{X(\pm t,x,\x)} \limt t {+\infty} +\infty.
\]
\end{lemma}

\begin{proof} 
Since $G(x)$ is close to $\Id$ for large $\abs x$ (in the sense of \eqref{dec-metric-a}), we can check that there exists $C \geq 0$ such that for $w \in \R^{2d}$ and $t \in \R$ we have
\begin{align*}
\frac {\partial^2} {\partial t^2} \abs{X(\pm t,w)}^2
& = 8 \abs{G(X(\pm t,w) ) \Xi (\pm t,w)}^2  - C \pppg {X(\pm t,w)}^{-\rho} \abs {\Xi(\pm t,w) }^2,
\end{align*}
where $\rho > 0$ is given by \eqref{dec-metric-a}. This is greater than $c_0 \abs{\Xi(\pm t,w)}^2$ for some $c_0 > 0$ if $\abs{X (\pm t,w)}^2 \geq \Rc_g^2$ with $\Rc_g$ large enough. Now let $\Rc_g$ be fixed. The assumption implies that there exists $\th \in [0,t]$ such that $\abs{ X (\pm \th, w)}^2 > \Rc_g^2$ and $\restr{\partial_s \abs{ X (\pm s,w)}^2}{s= \th} >0$ (and in particular $ \abs { \Xi (\pm \th,w) } > 0$). With the property on the second derivative we obtain that $s \mapsto \abs{ X (\pm s ,w)}^2$ is increasing for $s \geq \th$ and goes to infinity when $s$ goes to infinity.  
\end{proof}

Together with continuity of the flow, Lemma \ref{lem-flow} has important consequences which we shall use in the proof of Theorem \ref{th-high-freq}:

\begin{proposition} \label{prop-flow}
\begin{enumerate}[(i)]
\item We have 
\[
 \R^{2d} = \O_b^+ \sqcup \O_\infty^+ = \O_b^- \sqcup \O_\infty^- ,
\]
and in particular 
\[
 \R^{2d}  = \O_b^+ \cup \O_b^- \cup \O_\infty.
\]

\item $\O_\infty^+$, $\O_\infty^-$ and $\O_\infty$ are open in $\R^{2d}$, and $\O_b^+$, $\O_b^-$ and $\O_b$ are closed.
\item If $K$ is a compact subset of $\O_\infty$, then for all $R \geq 0$ we can find $T \geq 0$ such that $\abs{X (t,v)} \geq R$ for all $\abs t \geq T$ and $v \in K$. Moreover $\bigcup_{t\in\R} \vf^{-t}(K)$ is closed in $\R^{2d}$.
\end{enumerate}
\end{proposition}

\begin{proof} 

Let $\Rc_g$ be given by Lemma \ref{lem-flow}.

\noindent
(i) \quad  If $w \in \R^{2d}$ does not belong to $\O_b^\pm$ then for any $\g > 0$ there exists $t_0$ such that \eqref{hyp-flow} holds, and hence $w \in \O_\infty^\pm$ according to Lemma \ref{lem-flow}. The second statement easily follows.

\noindent 
(ii) \quad Let $w \in \O_\infty^\pm$. We can find $t_0 \geq 0$ such that Assumption ${\eqref{hyp-flow}}_\pm$ holds with $\g = 2$. By continuity of the flow there exists a neighborhood $\Vc$ of $w$ in $\R^{2d}$ such that it holds for all $v \in \Vc$ with the same $t_0$ and $\g = 1$, and hence $\Vc \subset \O_\infty^\pm$. This proves that $\O_\infty^\pm$ is open in $\R^{2d}$. Then we use (i) to conclude.

\noindent 
(iii) \quad Let $R \geq 0$. We can assume without loss of generality that $R \geq \Rc_g$ and $K \subset \singl{\abs x < R}$. For all $w \in K$ we can find $t_w \geq 0$ and a neighborhood $\Vc_w$ of $w$ such that $\abs{ X (\pm t_w,v)} > R$ for all $v \in \Vc_w$. According to Lemma \ref{lem-flow} this holds for any $t\geq t_w$. Since $K$ is compact we can find $T \geq 0$ such that for all $w \in K$ we have $\abs{X(\pm t,w)} > R$ for $t = T$ and hence for any $t \geq T$.
This proves the first claim. In particular for any $R \geq 0$ there exists $T \geq 0$ such that $\bigcup_{t\in\R} \vf^{t}(K) \cap \singl{\abs x \leq R} = \bigcup_{t\in [-T,T]} \vf^{t}(K) \cap\singl{\abs x \leq R}$. By continuity of the flow this set is compact for all $R \geq 0$, which implies that $\bigcup_{t\in\R} \vf^{t}(K)$ is closed.
\end{proof}

The damping condition \eqref{hyp-amort} has been stated for trapped trajectories. We now claim that it automatically holds for semi-trapped trajectories:

\begin{proposition} \label{prop-semi-amort}
 If \eqref{hyp-amort} holds then
\[
 \forall w \in \O_b^\pm(\R_+^*), \exists t \geq 0, \quad a (X (\pm t,w)) > 0.
\]
\end{proposition}

\begin{proof} 
  Let $w \in \O_b^\pm (\R_+^*)$ and $R = \sup \singl{\abs {X (\pm t,w)}, t \geq 0}$. We can find $w_\infty \in \R^{2d}$ and a sequence $\seq t n$ such that $t_n \to +\infty$ and $\vf^{\pm t_n}(w) \to w_\infty$. Let $t \in \R$. Since $t \pm t_n \geq 0$ for $n$ large enough we have
\[
 \abs{X(t,w_\infty)} = \lim_{n \to \infty}  \abs{X (t \pm t_n,w)} \leq R ,
\]
and hence $w_\infty \in \O_b(\R_+^*)$. Then, according to \eqref{hyp-amort}, there exists $t_\infty \in \R$ such that $a(X( \pm t_\infty ,w_\infty)) > 0$. Finally, since $a$ is continuous, we can find $n \in \N$ such that $t_n + t_\infty \geq 0$ and ${a ( X (\pm (t_n + t_\infty),w))} > 0$.
\end{proof}

\begin{lemma}[Escape function at infinity] \label{lem-esc-inf}
There exist $\Rc > 0$ and $C \geq 0$ such that we have on $\R^{2d}$: 
\[
\{ p ,f_0\}  \geq p \left( 1 -  C  \1 {\singl{\abs x \leq \Rc}}\right).
\]
\end{lemma}

\begin{proof} 
For any $(x,\x) \in \R^{2d}$ we have
 \[
 \{ p , f_0 \} (x,\x)  = 2 p(x,\x) - \innp { (x \cdot \nabla_x G(x)) \x }\x.
\]
Moreover there exists $c_1 > 0$ such that for all $(x,\x) \in \R^{2d}$ we have
\begin{equation} \label{abs-p}
 c_1 \inv \abs \x ^2 \leq p(x,\x) \leq c_1 \abs \x ^2.
\end{equation}
Thus for all $(x,\x) \in \R^{2d}$ we have
\[
 \{ p , f_0 \} (x,\x) \geq 2p(x,\x) - \nr{(x \cdot \nabla_x) G(x)} \abs \x^2 \geq (2-\nr{(x \cdot \nabla_x) G(x)} c_1)p(x,\x),
\]
and since $\nr{(x \cdot \nabla_x) G(x)}$ goes to 0 as $\abs x$ goes to $+\infty$, we have  $\{ p , f_0 \} \geq p$ if $\abs x$ is large enough.
\end{proof}

The following lemma uses Assumption \eqref{hyp-amort}:

\begin{lemma} [Escape function on semi-bounded geodesics] \label{lem-esc-trapped}
Let $I$ be a compact subset of $\R_+^*$, $C_b > 0$, $\Rc > 0$ and 
\[
 K_b = \big(\O_b^+(I) \cup \O_b^-(I)\big) \cap \singl{\abs x \leq \Rc}.
\]
Then $K_b$ is compact and there exist $f_b \in C_0^\infty(\R^{2d})$, $\b \geq 0$ and an open neighborhood $\Uc_b$ of $K_b$ such that on $\R^{2d}$:
\[
 \{ p,f_b\} + \b a \geq C_b \1 {\Uc_b}.
\]
\end{lemma}

\begin{proof} 
The set $K_b$ is bounded according to \eqref{abs-p} and closed according to Proposition \ref{prop-flow} (ii). Let $w \in K_b$. According to Proposition \ref{prop-semi-amort} there exists $t_w \in \R$ such that $a(X (t_w,w)) > 0$. Since $\vf^{t_w}$ is continuous we can find a neighborhood $\Vc_w$ of $w$ such that $a(X (t_w,v)) > 0$ for all $v \in \Vc_w$. Let $g_w \in C_0^\infty(\R^{2d},[0,1])$ be supported in $\Vc_w$ and equal to 1 on a neighborhood $\Uc_w$ of $w$, and consider
\[
 f_w = \int_0^{t_w} g_w \circ \vf^{-t}\, dt.
\]
The symbol $f_w$ is compactly supported and
\[
\{ p , f_w\} = g_w - g_w \circ \vf^{-t_w}.
\]
Since $g_w \circ \vf^{-t_w}$ is compactly supported in $\singl{(x,\x) \st a(x) > 0}$, there exists $\b_w \geq 0$ such that $\{p,f_w\} + \b_w a$ is non-negative on $\R^{2d}$ and at least equal to $1$ on $\Uc_w$.
Since $K_b$ is compact we can find $n\in\N^*$ and $w_1, \dots ,w_n \in K_b$ such that $K_b \subset \Uc_b := \bigcup_{j=1}^n \Uc_{w_j}$. Setting $f_b = C_b \sum_{j=1}^n f_{w_j}$ and $\b = C_b \sum_{j=1}^n \b_{w_j}$ we obtain a compactly supported symbol $f_b$ such that $\{ p,f_b \} + \b a$ is non-negative and at least equal to $C_b$ on the neighborhood $\Uc_b$ of $K_b$.
\end{proof}

\begin{lemma} [Escape function on a bounded set of non-trapped geodesics] \label{lem-esc-non-trapped}
 Let $K_\infty$ be a compact subset of $\O_\infty$, $C_\infty \geq 0$ and $\e > 0$. Then there exists $f_\infty \in C_0^\infty(\R^{2d})$ such that we have on $\R^{2d}$:
\[
\{ p , f_\infty \} \geq C_\infty \1 {K_\infty} - \e.
\]
\end{lemma}

\begin{proof} 
Since $\O_\infty$ is open (see Proposition \ref{prop-flow}), we can consider $g_\infty \in C_0^\infty(\R^{2d},[0,1])$ supported in $\O_\infty$ and equal to 1 on $K_\infty$. Let $\Vc$ be an open bounded neighborhood of $\supp g_\infty$ such that $ \bar \Vc \subset \O_\infty$, and let $T$ be given by Proposition \ref{prop-flow}(iii) applied with $K = \bar \Vc$ and $R$ so large that $\supp g_\infty \subset \{\abs x < R \}$.
We claim that for any $w \in \R^{2d}$ there exists a neighborhood $\Wc_w$ of $w$ and $\t_w \geq 0$ such that 
\begin{equation} \label{tauw-ginf}
 \forall v \in \Wc_w, \forall t \in \R_+ \setminus [\t_w, \t_w + T], \quad g_\infty\big( \vf^t(v)\big) = 0.
\end{equation}
It is clear if $w$ does not belong to $\Tc = \bigcup_{t \in \R} \vf^{-t} (\supp (g_\infty))$, which is closed, or if $w \in \supp (g_\infty)$ (with $\Wc_w = \Vc$ and $\t_w = 0$). Finally, let $w \in \Tc \setminus \supp(g_\infty)$ and $\t_w \geq 0$ such that $\vf^{\t_w}(w) \in \Vc$ but $\vf^{t}(w) \notin \supp (g_\infty)$ if $t \in [0,\t_w]$. Then there exists a neighborhood $\Wc_w$ of $w$ such that for $v \in \Wc_w$ we have $\vf^{\t_w} (v) \in \Vc$ but $\vf^{t}(v) \notin \supp (g_\infty)$ if $t \in [0,\t_w]$, and \eqref{tauw-ginf} holds true. As a consequence the function 
\[
\tilde  f_\infty = - \int_0^{+\infty} g_\infty \circ \vf^t \, dt 
\]
is well-defined and belongs to $\symbor(\R^{2d})$. Moreover $\{ p ,\tilde  f_\infty\}  = g_\infty$ is non-negative and equal to 1 on $K_\infty$. However $\tilde f_\infty$ is not compactly supported.
So let $\z \in C_0^\infty(\R^d,[0,1])$ equal to 1 on $\{\abs{x}< R\}$. Maybe after replacing $\z$ by $x \mapsto \z (\n x )$ for $\n > 0$ small enough, we can assume that $\abs{(C_\infty \tilde f_\infty \{ p, \z \})(w)} \leq \e$ for all $w \in p\inv(J_2)$. Then $f_\infty = C_\infty \tilde f_\infty \z \in C_0^\infty(\R^{2d})$ satisfies the conditions of the proposition. 
\end{proof}

Now we can prove Proposition \ref{prop-escape}:

\begin{proof} [Proof of Proposition \ref{prop-escape}]
We write $I = [E_1,E_2]$ with $E_2 \geq E_1 > 0$. Let $\Rc$ and $C$ be given by Lemma \ref{lem-esc-inf}. We apply Lemma \ref{lem-esc-trapped} with $C_b = C E_2  + E_1$. Let $f_b$, $\b$ and $\Uc_b$ be given by this lemma. Then we apply Lemma \ref{lem-esc-non-trapped} with $C_\infty = C E_2  + E_1$, $\e = E_1 / 2$ and $K_\infty = \O_\infty(I) \cap \singl {\abs x \leq \Rc} \setminus \Uc_b$. Setting $f_c = f_b + f_\infty$ and $c_0 = E_1 / 8$ finally gives the result.
\end{proof}

With Proposition \ref{prop-escape} and Theorem \ref{th-estim-insert2} we can finally prove Theorem \ref{th-high-freq}:

\begin{proof} [Proof of Theorem \ref{th-high-freq}]
\stepp
As already mentioned, we only have to prove the result for $\Im z \in ]0,1]$ and $\Re z \gg 1$. Thus we have to prove estimates on powers of the resolvent $(\hhs-\s^2)\inv$ (with inserted multiplications by $a$) for $h > 0$ small enough and $\s \in \C_{+,+}$ close to 1.
Let $I \subset \R_+^*$ be a compact neighborhood of $1$. Let $f_c \in C_0^\infty(\R^{2d})$, $\b \geq 0$ and $c_0 > 0$ be given by Proposition \ref{prop-escape}. We check that the self-adjoint operator
\begin{equation} \label{op-Fh}
F_h := \Opw(f_0 + f_c)
\end{equation}
satisfies for all $n \in \N$ the assumptions of Definitions \ref{defconjunif} and \ref{defconjunifdissn} (with $\a_h = c_0 h$ in \eqref{hyp-mourre}) to be conjugate to $\hhs$ uniformly in $\l = (h,\s)$ where $h \in ]0,h_0]$ for some $h_0 > 0$ and 
\[
\s \in [\g\inv,\g] + i]0,1]
\]
for some $\g>1$ close to 1.
For ($a$) we remark as usual that $\Sc(\R^d) \subset \Dom(\Re (\hhs)) \cap \Dom(F_h)$. We know that $A_h$ satisfies assumption ($b$), and for $\f \in H^2(\R^d)$ we can write 
\begin{align*}
\sup_{\abs t \leq 1} \nr{e^{-itA_h} \f- e^{-itF_h}\f}_{H^2}
& = \sup_{\abs t \leq 1}  \nr{\int_0^t e^{-is A_h} \Opw(f_c) e^{-i(t-s)F_h}\f\,ds}_{H^2} \\ 
& \leq \nr{\Opw(f_c)}_{\Lc(L^2,H^2)} \nr \f _{L^2} < \infty.
\end{align*}
Properties about commutators can be checked using pseudo-differential calculus.

\stepp
Let $\h \in C_0^\infty(\R, [0,1])$ be supported in $I$ and equal to 1 on a neighborhood $J$ of $1$. According to Proposition \ref{prop-escape} and the (easy) G\aa rding inequality (see Theorem 4.26 in \cite{zworski}) we have
\[
 \Opw \left( (\h \circ p)^2 (\{ p, f_0 + f_c\} + \b a ) + 4c_0  ((1-\h) \circ p)^2 \right) \geq 2c_0
\]
for $h > 0$ small enough, so there exists $C \geq 0$ such that 
\[
 \frac 1 h \h(\Re (\hhs)) \big(\left[ \Re (\hhs) , i F_h \right]  +  {\b} {\g} h \Re (\s) a(x) \big)  \h(\Re (\hhs)) + 4c_0 \left((1 -\h(\Re (\hhs))\right)^2  \geq 2 c_0 - C h,
\]
uniformly in $\s$ with $\Re \s \geq \g\inv $ and $\abs \s$ bounded. Multiplication by $h$ and composition by $\1 J (\Re (\hhs))$ on both sides gives
\[
\1 J (\Re (\hhs)) \big( \left[ \Re (\hhs) , i F_h \right] + \b \g  h  \Re (\s) a(x) \big) \1 J (\Re (\hhs))  \geq c_0 h \1 J (\Re (\hhs))
\]
for $h > 0$ small enough, which is exactly assumption ($e$) on $J$.

\stepp 
It is easy to see that multiplication by $a$ or a derivation $D^\a$ with $\abs \a \leq 1$ belong to $\opinsert_N(\hhs, A_h)$ for any $N \in \N$ (see Definition \ref{def-415}). Moreover there exists $C \geq 0$ such that for all $h \in ]0,1]$ and $\s $ close to 1 we have
\[
 \nr{a}_{F_h,N} \leq C \quad \text{and} \quad  \nr{D^\a}_{\hhs, F_h,N} \leq \frac C {h^{\abs \a}}. 
\]
Now we can apply Theorem \ref{th-estim-insert2}. For all $n \in \N$, $\d \geq n + \frac 12$ and $\n_0,\dots \n_n \in \{0,1\}$ there exists $C \geq 0$ such that for all $h \in ]0,h_0]$ and $\s \in \C_{+,+}$ close to 1 we have
\[
\nr{\pppg {F_h}^{-\d} D^\a (\hhs - \s^2)\inv a(x)^{\n_1} (\hhs - \s^2)\inv \dots a(x)^{\n_n} (\hhs - \s^2)\inv \pppg {F_h}^{-\d}} \leq \frac C {h ^{1+\abs\a}}.
\]

\stepp
It remains to see that we can replace the weights $\pppg {F_h}^{-\d}$ given by the Mourre theory by weights $\pppg x ^{-\d}$. For this we can proceed as for intermediate frequencies (see also \cite{art-mourre}). Finally we use \eqref{R-hhs} and Proposition \ref{prop-der-R2} to conclude.
\end{proof}

\section{The case of a Laplace-Beltrami operator} \label{sec-laplacien}
\setcounter{equation}{0}
In this section, we explain how to prove the following analogue of Theorem \ref{th-loc-decay} when $ H_0 $ is replaced by a Laplace-Beltrami operator
$$  - \Delta_g = - \sum_{j,k=1}^d |g(x)|^{-1} \partial_j \big( |g(x)| g^{jk}(x) \partial_k \big), \qquad \text{where } |g (x)| = \mbox{det} (g_{jk}(x))^{1/2} , $$
for any metric $ g $ which is a long range perturbation of the flat one, in the sense that
$$  \big| \partial^{\alpha} \big(g_{jk}(x) - \delta_{jk} \big) \big| \leq C_{\alpha} \langle{ x \rangle}^{-\rho-|\alpha|} , $$
with $ \rho > 0 $. As usual, we have set $ (g^{jk}(x)) = (g_{jk}(x))^{-1} $. The analogue of Theorem \ref{th-loc-decay} is the following theorem.
\begin{theorem} \label{wave-Laplacien} Assume that every   bounded geodesic of $g$ goes through the damping region. Let $ \delta > d + \frac{1}{2} $ and $ \e > 0 $. Then there exists $ C > 0 $ such that, for all $ (u_0,u_1) \in H^{2,\delta} \times H^{1,\delta} $, the solution $u$ to
$$ \partial_t^2 u - \Delta_g u + a (x)  \partial_t u = 0, \qquad u_{|t=0} = u_0, \ \ \partial_{t}u_{|t=0} = u_1 , $$
satisfies
$$ || \nabla u (t) ||_{L^{2,-\delta}} + || \partial_t u (t) ||_{L^{2,-\delta}} \leq C \langle t \rangle^{- (d - \e)} \big( || u_0 ||_{H^{2,\delta}} + || u_1 ||_{H^{1,\delta}} \big) $$
for all $ t \geq 0 $.
\end{theorem}

The proof of this theorem follows exactly the same line as the one discribed in Section \ref{sec-time}, \ie it is a consequence of resolvent estimates analogue to those obtained for $H_0 $ in Theorems \ref{th-inter-freq}, \ref{th-low-freq} and \ref{th-high-freq}. The main issue is to obtain estimates at low frequencies, \ie when $ z \rightarrow 0 $. The estimates at intermediate and high frequencies do not depend on the precise structure of $ H$ (or $ H_0 $) to which we could add (symmetric) first order long range perturbations, so we do not consider this part. The relevant low frequency estimates used to prove Theorem \ref{wave-Laplacien} are given in Theorem \ref{low-freq-Laplacien} below.\\

We let 
$$ H = \mbox{ self-adjoint realization of }  - \Delta_g  \ \mbox{ on }  L^2_g= L^2 (\R^d , d {\rm vol}_g) , $$
where $ d {\rm vol}_g = |g(x)| dx  $. 
 By the general arguments given in Section \ref{sec-diss}, the resolvent
$$ R_g (z) = \big( H - i z a - z^2 \big)^{-1}, \qquad z \in {\mathbb C}_+ , $$
is well defined and analytic with respect to $ z \in {\mathbb C}_+ $. Since $ |g(x)| $ is bounded from above and from below, we have $  L^2 (\R^d , d {\rm vol}_g) =  L^2 (\R^d , dx)  $ and their norms are equivalent, so we can see $ R_g (z) $ (as well as its derivatives in $z$) as an operator on the standard $ L^2 = L^2 (\R^d,dx) $ space. The only difference, which will be irrelevant here, is that bounds of the form $ || R_{g} (z) ||_{{\mathcal L}(L^2_g)} \leq (\mbox{Im} z)^{-1}  $ are replaced by $ || R_{g} (z) ||_{{\mathcal L}(L^2)} \leq C (\mbox{Im} z)^{-1} $.\\

We next recall that we can choose coordinates on $ \R^d $ such that $ |g(x)| = 1 $ outside a compact set (see \cite{bouclet11}) which we assume from now on.
The interest of this remark is that
\begin{eqnarray} 
 H = H_0 + W , \qquad W = \sum_{j=1}^d b_j (x) D_j  , \label{pourW}
\end{eqnarray}
where $ b_1 , \ldots , b_n  \in C_0^{\infty} (\R^d) $ and $H_0$ is as in \eqref{wave-lap} with $ G (x) = (g^{jk}(x)) $. 
 We also let as before
$$ R (z) = (H_0 - i z a - z^2)^{-1}, \qquad z \in {\mathbb C}_+ . $$
Our strategy is to take advantage of the estimates on $ R (z) $ proved in Theorem \ref{th-inter-freq} and to use a perturbative argument {\it \`a la} Jensen-Kato \cite{jensenk79} to derive estimates on $ R_g (z) $. If we set as before $ z = \tau + i \mu $, we can write
\begin{eqnarray}
 R_g (z) & = & R (z) - R_g (z) W R (z) \\
         & = & R (z) - R_g (z) W R (i \mu) - R_g (z) W \big( R (z) - R (i \mu) \big) . \label{Rg2}
\end{eqnarray}
We record the main technical results of this section in  the next two propositions.
\begin{proposition} \label{continuite}  There exists an operator denoted by $ R (0) : \dot{H}^{-1} \rightarrow {\dot H}^1  $ such that, for all $ \delta > 1 $ and $ N \geq 0 $, 
$$  
\nr{ \pppg x ^N W \big(  R (i \mu) -  R (0) \big) \pppg x ^{-\delta}}_{\Lc(L^2)} \limt \m {0^+} 0. 
$$
\end{proposition}
Notice that $ \langle x \rangle^{-\delta} $ maps $ L^2 $ into $ \dot{H}^{-1} $ and that, since $ W $ is a first order differential operator with (smooth and) compactly supported coefficients, $ \langle x \rangle^N  W $ maps $ \dot{H}^1 $ into $ L^2 $.
\begin{proposition} \label{inversibilite} For all $ \delta > 1 $, the operator $ I + \langle x \rangle^\delta W R (0)  \langle x \rangle^{-\delta} $ is invertible on $ L^2 $.
\end{proposition}
This proposition is essentially a consequence of the fact that $ 0 $ is neither an eigenvalue nor a resonance for $ H $, \ie that there are no non trivial solution to $ \Delta_g u = 0 $  in $ \dot{H}^1 $. 

We postpone the proofs of Propositions \ref{continuite} and \ref{inversibilite} to the end of the present section and explain first how to use them to get resolvent estimates.
\begin{theorem} \label{low-freq-Laplacien} Let $ \e > 0 $ and $ n \leq d $. Let $ \delta $ be greater than $ n + \frac{1}{2} $ if $ n \geq \frac{d}{2} $ and greater than $ n + 1 $ otherwise. Then there exists a neighborhood $ {\mathcal U} $ of $ 0 $ in $ {\mathbb C} $ and $ C \geq 0 $ such that for all $ z \in {\mathcal U} \cap {\mathbb C}_+ $ we have
$$ \left| \left| \langle x \rangle^{- \delta} R_g^{(n)}(z) \langle x \rangle^{- \delta} \right| \right|_{{\mathcal L}(L^2)} \leq C \left( 1 + |z|^{d-2-n- \e} \right) $$
and 
$$  \left| \left| \langle x \rangle^{- \delta} \nabla R_g^{(n)}(z) \langle x \rangle^{- \delta} \right| \right|_{{\mathcal L}(L^2)} \leq C \left( 1 + |z|^{d-1-n- \e} \right) . $$
\end{theorem}

Here we consider derivatives of order $ n \leq d $ for this is only what we need to prove the result on the energy decay.

\begin{proof}
We first show that the estimates hold with $ \delta $ replaced by some $ N $ large enough. Using (\ref{Rg2}), we have
$$ R_g (z) \langle{ x \rangle}^{-N} \big( I + A (z)  \big)  =  R (z) \langle{ x \rangle}^{-N} , $$
where
$$ A (z) = \langle{ x \rangle}^N W R (0) \langle{ x \rangle}^{-N} + \langle{ x \rangle}^N W \big( R (i \mu) - R (0) + R (z) - R (i \mu)  \big) \langle{ x \rangle}^{-N} . $$
By Propositions \ref{continuite} and \ref{inversibilite}, $ I +  \langle{ x \rangle}^N W R (0) \langle{ x \rangle}^{-N} + \langle{ x \rangle}^N W \big( R (i \mu) - R (0) \big) \langle x \rangle^{-N} $ is invertible if $ \mu $ is small enough. Furthermore, by using Theorem \ref{th-low-freq} and proceeding as in the end of the proof of Proposition \ref{prop-b589}, we see that
 $  \langle{ x \rangle}^N W  \big( R (z) - R (i \mu)  \big) \langle{ x \rangle}^{-N} $ is small when $ z $ is close to zero (by writing the difference as an integral).
Therefore $ I + A (z) $ is invertible if $ z $ is small enough, so that we can write (for $ j = 0 $ or $1$)
\begin{eqnarray}
 \langle{ x \rangle}^{-N} \nabla^j R_g (z) \langle{ x \rangle}^{-N} = \langle{x \rangle}^{-N} \nabla^j R (z) \langle{ x \rangle}^{-N} \big( I + A (z) \big)^{-1} . \label{recurrence} 
\end{eqnarray} 
Using the form of $ W $ in (\ref{pourW}) and the estimates of Theorem \ref{th-low-freq}, it is not hard to check (using in particular that $ N  $ is large enough) that
$$ \left| \left| \partial_z^n \big( I + A (z) \big)^{-1} \right|\right|_{{\mathcal L}(L^2)} \lesssim \left( 1 + |z|^{d-1-n - \e} \right) .  $$
The form of $ W $ ensures that it suffices to use estimates  on $ \langle{ x \rangle}^{-\delta} \nabla R^{(n)}(z) \langle{ x \rangle}^{- \delta} $, which are better than those on $  \langle{ x \rangle}^{-\delta} R^{(n)}(z) \langle{ x \rangle}^{- \delta}  $.
These estimates, Theorem \ref{th-low-freq} and (\ref{recurrence}) yield easily the expected estimates with $ N $ instead of $ \delta $. To obtain the result with  $ \delta $, it then suffices to write
$$ R_g (z) = R (z) - R (z) W R (z) + R (z) W R_g (z) W R (z) , $$
and use the previous estimates on $ \langle{ x \rangle}^{-N} \nabla^j R^{(n)}_g (z) \langle{ x \rangle}^{-N} $ (since $ W $ has compactly supported coefficients) combined with those  given by Theorem \ref{th-low-freq}.
\end{proof} 




The rest of the section is devoted to  the proofs of Propositions \ref{continuite} and \ref{inversibilite}.

\begin{lemma} The operator $ H_0 $, acting in the distributions sense, is an isomorphism from $ \dot{H}^1 $ to $ \dot{H}^{-1} $.  We denote by $R(0)$ its inverse. Then for all $ f \in \dot{H}^{-1} $, $ u:= R (0)f $ is the unique solution in $\dot H^1$ to $ H_0 u = f $ in the temperate distributions sense. 
\end{lemma}

\begin{proof}
This lemma is basically a consequence of the fact that
\begin{eqnarray}
\big| \big| |D| u \big| \big|_{L^2} \lesssim \big| \big| H_0^{1/2} u \big| \big|_{L^2} \lesssim \big| \big| |D| u \big| \big|_{L^2}, \qquad u \in H^1 , \label{borneequivalence}
\end{eqnarray}
and the standard Lax-Milgram argument.
\end{proof}




\begin{lemma} \label{equivalenceHmoinsun} 
The space $ L^2 \cap \dot{H}^{-1} $ is contained in the domain $\Dom\big(H_0^{-1/2}\big)$ of $H_0^{-1/2}$ and there exists $ C > 1 $ such that
\begin{eqnarray}
 C^{-1} || \psi ||_{\dot{H}^{-1}} \leq \big| \big| H_0^{-1/2} \psi \big| \big|_{L^2} \leq C || \psi ||_{\dot{H}^{-1}} , \label{equivalence}
\end{eqnarray} 
for all $ \psi \in L^2 \cap \dot{H}^{-1} $.
\end{lemma}

\begin{proof} 
Let $ \psi \in L^2 \cap \dot{H}^{-1} $ and $ v \in \Dom \big(H_0^{-1/2}\big) $. Then
$$ \big( \psi , H_0^{-1/2} v \big) = \big( |D| |D|^{-1} \psi , H_0^{-1/2} v \big) = \big(  |D|^{-1} \psi , |D| H_0^{-1/2} v \big) ,
 $$
since $ |D|^{-1} \psi $ belongs to $ L^2 $ hence to $ H^1 $ and similarly $ H_0^{-1/2} v $ belongs to $ \mbox{Dom}(H_0^{1/2}) = H^1 $. Since $H_0^{1/2} H_0^{-1/2} v = v$ we obtain
 $$ \big| \big( \psi , H_0^{-1/2} v \big) \big| \leq \big| \big| |D|^{-1} \psi \big| \big|_{L^2} \big| \big| |D| H_0^{-1/2} v \big| \big|_{L^2} \lesssim || \psi ||_{\dot{H}^{-1}} || v ||_{L^2}  $$
 using (\ref{borneequivalence}). All this shows that $ \psi  $ belongs to $ \mbox{Dom}(H_0^{-1/2}) $ and that $ \big| \big| H_0^{-1/2} \psi \big| \big|_{L^2} \lesssim || \psi ||_{\dot{H}^{-1}}  $. To prove the reverse inequality, we observe that for all $ \phi \in {\mathcal S} ({\mathbb R}^d) $,
$$ \big( |D|^{-1} \psi ,  \phi \big) =  \big( \psi , |D|^{-1} \phi \big) = \big( H_0^{1/2} H_0^{-1/2} \psi , |D|^{-1} \phi \big) = \big(  H_0^{-1/2} \psi ,  H_0^{1/2} |D|^{-1} \phi \big) , $$
using, in the third equality that $ |D|^{-1} \phi $ belongs to $ H^1 $. Using (\ref{borneequivalence}) again, this implies that
$$ \big| \big( |D|^{-1} \psi ,  \phi \big) \big| \lesssim || H_0^{-1/2} \psi ||_{L^2} || \phi ||_{L^2} , $$
which yields $ \big| \big | |D|^{-1} \psi \big| \big|_{L^2} \lesssim || H_0^{-1/2} \psi ||_{L^2} $. This completes the proof.
\end{proof}






\begin{lemma} Let $ f \in L^2 \cap \dot{H}^{-1} $, then
\begin{eqnarray}
 ( H_0 + \mu^2 )^{-1} f \rightharpoonup R (0) f \qquad \mbox{in} \ \dot{H}^1 , \label{convergencefaible}
\end{eqnarray}
as $ \mu \rightarrow 0 $.
\end{lemma}

\begin{proof} Observe first that $$ 
|| ( H_0 + \mu^2 )^{-1} f ||_{\dot{H}^1} = \big| \big| |D| ( H_0 + \mu^2 )^{-1} f \big| \big|_{L^2} \lesssim || H_0^{1/2} (H_0 + \mu^2)^{-1} f ||_{L^2} . $$ 
Since $ f  $ belongs to $  L^2 \cap \dot{H}^{-1} $, $ |D|^{-1} f $ belongs to $ H^1 $ and, by writing for all $ \psi \in L^2 $,
$$  \big(  H_0^{1/2} (H_0 + \mu^2)^{-1} f , \psi \big) = \big(   |D|^{-1} f , |D| (H_0 + \mu^2)^{-1}H_0^{1/2}  \psi \big) , $$
we see that
$$ || H_0^{1/2} (H_0 + \mu^2)^{-1} f ||_{L^2} \lesssim || H_0^{1/2} (H_0 + \mu^2)^{-1} H_0^{1/2} ||_{{\mathcal L}(L^2)} \big| \big| |D|^{-1} f  \big| \big|_{L^2} \lesssim || f ||_{\dot{H}^{-1}} , $$
and thus
\begin{eqnarray}
|| ( H_0 + \mu^2 )^{-1} f ||_{\dot{H}^1} \lesssim || f ||_{\dot{H}^{-1}} , \qquad \mu > 0 , \ f \in L^2 \cap \dot{H}^{-1} . \label{borneuniforme}
\end{eqnarray}
The interest of the uniform bound (\ref{borneuniforme}) is that it suffices to prove (\ref{convergencefaible}) for $f$ in a dense subset of $ L^2 \cap \dot{H}^{-1} $. Let $ \chi \in C_0^{\infty} ({\mathbb R}) $ such that $ \chi \equiv 1 $ near $ 0 $. Then $ (1-\chi)(H_0 / \e) f \rightarrow f $ in $ L^2 \cap \dot{H}^{-1} $ as $ \e \rightarrow 0 $, by the Spectral Theorem\footnote{and the fact that $ 0 $ is not an eigenvalue of $ H_0 $} and (\ref{equivalence}). Therefore, it suffices to prove the result when $ f $ is replaced by $ (1-\chi)(H_0 / \e) f $. In this case, using the Spectral Theorem again
$$ \left| \left|  ( H_0 + \mu^2 )^{-1} (1-\chi)(H_0 / \e) f - \frac{(1-\chi)(H_0 / \e)}{H_0} f \right| \right|_{\dot{H}^1}   \rightarrow 0, \qquad \mu \rightarrow 0 , $$
since the left hand side is not greater than
$$  C \left| \left|  H_0^{1/2} ( H_0 + \mu^2 )^{-1} (1-\chi)(H_0 / \e) f - H_0^{1/2} \frac{(1-\chi)(H_0 / \e)}{H_0} f \right| \right|_{L^2} . $$
It remains to observe that 
$$ \frac{(1-\chi)(H_0 / \e)}{H_0} f = R (0) (1-\chi)(H_0 / \e) f $$
since the left hand side belongs to $ H^2 \subset \dot{H}^1 $ and solves $ H_0 u = (1-\chi)(H_0 / \e) f $ in the distributions sense.
\end{proof}

\bigskip

\begin{proof}[Proof of Proposition \ref{continuite}.]
 We prove first that
\begin{eqnarray}
  \left| \left| \langle x \rangle^N W \big(  (H_0 + \mu^2)^{-1} -  R (0) \big) \langle x \rangle^{-\delta} \right| \right|_{{\mathcal L}(L^2)} \rightarrow 0, \qquad \mu \rightarrow 0^+ . \label{sousconvergence} 
\end{eqnarray}
Let $ \chi \in C_0^{\infty}(\R) $ be equal to $1$ near $ 0  $. Write
$$   W (H_0 + \mu^2)^{-1} =  W \chi (H_0) (H_0 + \mu^2)^{-1} \chi (H_0) +  W (1 - \chi^2(H_0)) (H_0 + \mu^2)^{-1} . $$
The convergence of the second term of the right hand side is easy. The contribution of the first term is obtained by writing
$$  \left( \langle{ x \rangle}^N W \chi (H_0) \langle{ x \rangle} \right) \langle{ x \rangle}^{-1} (H_0 + \mu^2)^{-1} \langle{ x \rangle}^{-1} \big( \langle{ x \rangle} \chi (H_0) \langle{ x \rangle}^{-\delta} \big) $$
where  $\left( \langle{ x \rangle}^N W \chi (H_0) \langle{ x \rangle} \right)$ and  $\left( \langle{ x \rangle} \chi (H_0) \langle{ x \rangle}^{-\delta} \right)$ are compact
on $ L^2  $ and $ \langle{ x \rangle}^{-1} (H_0 + \mu^2)^{-1} \langle{ x \rangle}^{-1} $ converges weakly to  $  \langle{ x \rangle}^{-1} R (0) \langle{ x \rangle}^{-1}$ on $ L^2  $ by Lemma \ref{convergencefaible}
and the fact that $ \langle{ x \rangle}^{-1} $ is bounded from $ L^2  $ to $ \dot{H}^{-1} \cap L^2 $. Then, to complete the proof \ie replace $ (H_0 + \mu^2)^{-1} $ by $ R (i \mu) = ( H_0 + \mu a + \mu^2)^{-1} $ in (\ref{sousconvergence}), it suffices to prove that \begin{eqnarray}
 || (H_0 + \mu a + \mu^2)^{-1} - (H_0  + \mu^2)^{-1} ||_{\Lc(\dot{H}^{-1}, \dot{H}^1)} \rightarrow 0, \qquad \mu \rightarrow 0 . \label{convergenceoperateur} 
\end{eqnarray} 
Using the resolvent identity,  (\ref{convergenceoperateur}) follows from the fact that 
$$  \big| \big| |D| (H_0 + \mu a + \mu^2)^{-1} a \mu (H_0  + \mu^2)^{-1} |D| \big| \big|_{{\mathcal L}( L^2 )} \lesssim \mu^{\rho} . $$
Using the fact that $ H_0 + \mu a + \mu^2 \geq H_0 \geq c |D|^2 $ (when tested on $L^2 $ against any $ H^2 $ function), the above estimate is reduced to
$$  \big| \big|  (H_0 + \mu a + \mu^2)^{-1/2} a \mu (H_0  + \mu^2)^{-1/2}  \big| \big|_{\Lc(L^2)} \lesssim \mu^{\rho} . $$
By the Hardy inequality $ \langle{ x \rangle}^{-1} (H_0+\mu^2)^{-1/2} $ is uniformly bounded, therefore it suffices to consider
\begin{eqnarray*}
\lefteqn{ \big| \big| \mu  a  \langle{ x \rangle} (H_0 + \mu a + \mu^2)^{-1/2}  \big| \big|_{{\mathcal L}(L^2)}} \\
&& \lesssim  \mu \big| \big| \langle{ x \rangle}^{-\rho} (H_0 + \mu a + \mu^2)^{-1/2}  \big| \big|_{{\mathcal L}(L^2)} \\
&& \lesssim  \mu \big| \big| \langle{ x \rangle}^{-1} (H_0 + \mu a + \mu^2)^{-1/2}  \big| \big|_{{\mathcal L}(L^2)}^{\rho} \big| \big|  (H_0 + \mu a + \mu^2)^{-1/2}  \big| \big|_{{\mathcal L}( L^2 )}^{1 - \rho} \\
&& \lesssim  \mu^{\rho} ,
\end{eqnarray*}
 using a (simple) interpolation argument in the second line. The result follows. 
\end{proof}
 
 \bigskip

\begin{proof}[Proof of Proposition \ref{inversibilite}.] The proof of Proposition \ref{continuite} shows that $\langle{ x \rangle}^{\delta} W R (0) \langle{ x \rangle}^{-\delta}$ can be written as
$$\left( \langle{ x \rangle}^{\delta} W \chi (H_0) \langle{ x \rangle} \right) \langle{ x \rangle}^{-1}R(0) \langle{ x \rangle}^{-1} \left( \langle{ x \rangle} \chi (H_0) \langle{ x \rangle}^{-\delta} \right) +  \langle{ x \rangle}^{\delta} W \frac{1 - \chi^2(H_0)}{H_0} \langle{ x \rangle}^{-\delta} , $$ 
which implies that this is a compact operator on $ L^2 $. Therefore, $ I + \langle{ x \rangle}^{\delta} W R (0) \langle{ x \rangle}^{-\delta}  $ is a Fredholm operator of index $ 0 $ and thus is bijective if and only if it is injective. Assume that $ v \in L^2 $ satisfies
\begin{eqnarray}
 v + \langle{ x \rangle}^{\delta} W R (0) \langle{ x \rangle}^{-\delta} v = 0 . \label{valeurpropre}
\end{eqnarray}
Then, $ \langle{ x \rangle}^{-\delta} v $ belongs to $ \dot{H}^{-1} $ hence 
$$u := R (0) \langle{ x \rangle}^{-\delta} v \in \dot{H}^1 , $$
satisfies  $ \big( H_0 + W \big) u = - \Delta_g u = 0 $.
This implies that $ u = 0 $ and therefore $v = 0$ which completes the proof. We simply record that the claim $u = 0$ is obtained by showing
$$ || \nabla_g u ||_{L^2_g}^2 = -\lim_{\e \rightarrow 0} \big( \Delta_g  (\chi (\e x) u ), \chi (\e x) u \big)_{L^2_g} =0. $$
This in turn is proved by routine arguments, using on one hand
$$  \nabla_g \big( (\chi (\e) x) u \big) = \chi (\e x) \nabla_g u + {\mathcal O} \big( \e |x| (\nabla \chi) (\e x) \big) |x|^{-1} u \rightarrow \nabla_g u \qquad \mbox{in} \ L^2 , $$
and on the other hand
$$ \Delta_g  (\chi (\e x) u ) =  u \Delta_g  (\chi (\e x) ) + 2 g \big(\nabla_g u , \nabla_g \big( \chi (\e x) \big)\big) , $$
combined with the fact that $ |x|^2 \Delta_g  (\chi (\e x) )  $ and $ |x| \nabla_g \big( \chi (\e x) \big) $ are families of operators on $ L^2 $ going to zero in the strong sense (for the first family one uses that $ \partial_j (|g(x)| g^{jk}(x)) $ is a short range symbol).
\end{proof}

\appendix

\section{Notation} \label{section-notations}
{\noindent \bf General notation}
\begin{itemize}
\item{Sets of integers
$$ {\mathbb N} = \mbox{set of non negative integers}, \qquad \Ii j k = [j,k] \cap {\mathbb N} $$}
\item{Sets of complex numbers
\begin{eqnarray*}
 {\mathbb C}_+ & = & \{ z \in {\mathbb C} \ : \ {\rm Im} (z) > 0 \} \\
  {\mathbb C}_{\pm,+} &  = & \{z \in {\mathbb C} \ : \  {\rm Im} ( z) > 0, \ \pm {\rm Re} ( z ) > 0 \} \\
  {\mathbb C}_{I,+} & = & \{ z \in {\mathbb C} \ : \ {\rm Re}(z) \in I, \ {\rm Im}(z)> 0  \} 
\end{eqnarray*}  }
\item{ Commutators 
\begin{eqnarray*}
\mbox{ad}_C (B) & = & [B,C] = BC- CB \\
\mbox{ad}_C^k(B)& = & \mbox{ad}_C \big( \mbox{ad}_C^{k-1}(B) \big), \qquad k \geq 2 
\end{eqnarray*}}
\item{$ L^2 $ dilations and their generator
\begin{eqnarray}
e^{i \theta A} u (x) = e^{\theta \frac{d}{2}} u (e^{\theta}x), \qquad A = \frac{x\cdot \nabla + \nabla \cdot x}{2 i}
\label{dilatation-notation}
\end{eqnarray}}
\end{itemize}
{\noindent \bf Operators}
\begin{itemize}
\item{Differential operators
\begin{eqnarray*}
H_0 & = & - \mbox{div} \big( G(x) \nabla \big)  \\
H_z & = & H_0 - i z a (x), \qquad z \in {\mathbb C}_+  \\
{\mathcal A} & = & \left( \begin{matrix} 0 & I \\ H_0 & - i a \end{matrix} \right)
\end{eqnarray*}}
\item{ Resolvents
\begin{eqnarray*}
R (z) & = & \big(H_z - z^2 \big)^{-1}  \\
R_g (z) & = & \big(-\Delta_g - i z a - z^2 \big)^{-1} \qquad \mbox{(in Section \ref{sec-laplacien} only)}
\end{eqnarray*}}
\item{More technical definitions (to study resolvents with inserted factors)
\begin{eqnarray*}
{\mathcal R}_{j,k} (z) & = & R (z) a_{j+1} R (z) \cdots R (z) a_k R (z)  \\
\Theta_{j;\alpha_j,\ldots,\alpha_k} & = & (H_z + 1)^{-\alpha_j} a_{j+1} (H_z + 1)^{- \alpha_{j+1}} \cdots a_k (H_z+1)^{-\alpha_k} 
\end{eqnarray*}
where $ a_{j+1}, \ldots , a_k \in C^{\infty}_b ({\mathbb R}^d) $ (often 1 or derivatives of $a$ in practice).}
\end{itemize}

{\noindent \bf Specific notation for low frequency analysis}
\begin{itemize}
\item{ $ \iota  $ : index refering to perturbations "at infinity", {\it i.e.} after the removal of a 
  compactly supported part.}
\item{$ \tilde{} \  $  refers to rescaled operators. }
\item{ $ \hat{z} = z / |z| $}
\item{$ b_{|z|} (x)  =  b (x/|z|) $}
\item{Operators
\begin{eqnarray*}
P_0 & = & - \mbox{div} \big( \chi I_d + (1-\chi) G \big) \nabla \\
a_{\iota} & = & (1-\chi)a \\
K_0 & = & H_0 - P_0 \\
a_0 & = & \chi a \\ 
P_z & =& P_0 - i z a_{\iota} \\
K_z & =& K_0 - i za_0 \\
\tilde{P}_z^0 & = &  \frac{1}{|z|^2} e^{-iA \ln|z|} P_0 e^{i A \ln |z|} \\
\tilde{P}_z & = & \frac{1}{|z|^2} e^{-iA \ln|z|} P_z e^{i A \ln |z|} \\
\end{eqnarray*}}
More explicitly
\begin{eqnarray*}
\tilde{P}_z^0 & = & - \mbox{div} \big( \big( \chi (x/|z|) I_d + (1-\chi)(x/|z|) G(x/|z|) \big) \nabla \big) \\
\tilde{P}_z & = & \tilde{P}_z^0 - i \frac{z}{|z|^2} a_{\iota}(x/|z|)
\end{eqnarray*}
\item{ Resolvents
\begin{eqnarray*}
R_{\iota} (z) & = & (P_z - z^2)^{-1} \\
\tilde{R}_{\iota} (z) & = & ( \tilde{P}_z - \hat{z}^2 )^{-1}, \qquad |z| \ll 1
\end{eqnarray*}}
\item{More technical definitions
\begin{eqnarray*}
\Phi_0 & = & D^{\alpha} \ \ \mbox{with} \ |\alpha| = \nu_0 \in \{0,1\}   \\
\Phi_j & = & a_{\iota}^{\nu_j} \ \ \mbox{with} \ \nu_j \in \{0,1\}, \ j \geq 1  \\
{\mathcal R}^{\iota}_{j,k}(z) & = & R_{\iota} (z) \Phi_{j+1} R_{\iota}(z) \cdots R_{\iota}(z) \Phi_k R_{\iota} (z)
\end{eqnarray*}
and rescaled versions
\begin{eqnarray*}
\tilde{\Phi}_0 (z) & = & |z|^{\nu_0} D^{\alpha} \\
\tilde{\Phi}_j (z) & = & a_{\iota} (x/|z|)^{\nu_j}, \ \ j \geq 1  \\
\tilde{\mathcal R}^{\iota}_{j,k}(z) & = & \tilde{R}_{\iota} (z) \tilde{\Phi}_{j+1} (z) \tilde{R}_{\iota}(z) \cdots \tilde{R}_{\iota}(z) \tilde{\Phi}_k (z) \tilde{R}_{\iota} (z) \\
\tilde{\Theta}_{j;\alpha_j,\ldots,\alpha_k} & = & (\tilde{P}_z + 1)^{-\alpha_j} \tilde{\Phi}_{j+1} (z) (\tilde{P}_z + 1)^{- \alpha_{j+1}} \cdots \tilde{\Phi}_k(z) (\tilde{P}_z+1)^{-\alpha_k}  \\
\tThiota^b_{ j_1,\dots,j_m} (z) & = &  (\Ptau +1)\inv  b_{j_1,\abs z}(x)  (\Ptau +1)\inv  b_{j_2,\abs z}(x) \dots  (\Ptau +1)\inv  b_{j_m,\abs z}(x) 
\end{eqnarray*}
where, in the last line, $ b_{j,|z|}(x) $ is either $1$ or  $ a_{\iota}(x/|z|) $. 
We also consider the following quantity (see (\ref{matcalV})) which we use to count powers of $|z|$ 
$$ {\mathcal V}_{j,k} = \sum_{l=j+1}^k \nu_j . $$}
\end{itemize}

{\noindent \bf Specific notation for high frequency analysis}
\begin{itemize}
\item{Spectral cutoffs (in Section \ref{sec-time}): for any bounded $ \chi \in C^{\infty} ({\mathbb R}) $ such that $ 0 \notin \mbox{supp}(\chi) $
$$ \chi_j (H_0^{1/2}) = \chi (2^{1-j} H_0^{1/2})  $$}
\item{ Semi-classical operator (in Section \ref{sec-high-freq})
\begin{eqnarray*}
H_h^{\sigma} = h^2 H_0 - i h \sigma a (x), \qquad h = |z|^{-1}, \ \ \sigma = z / |z| 
\end{eqnarray*}}
\item{Conjugate operator
$$ F_h = \mbox{Op}_h^w (f_0 + f_c) $$
with $ f_0(x,\xi) = \langle x , \xi \rangle_{{\mathbb R}^d} = x \cdot \xi $  the usual escape function and  $ f_c \in C_0^{\infty}({\mathbb R}^{2d}) $ chosen according to Proposition \ref{prop-escape}.}
\end{itemize}

\section{Dissipative Mourre estimates: an example} \label{AppB}
In this short section, we study the simple case where
$ H_0 = - \Delta$ is the flat Laplacian and describe which type of parameters are considered in Definition \ref{defconjunif} to handle respectively the high, low and medium frequency regimes in the proofs of Theorems \ref{th-high-freq}, \ref{th-low-freq} and \ref{th-inter-freq}. In each case, we also check the positive commutator estimate which is the main assumption in Definition \ref{defconjunif}.

\begin{itemize}
\item{{\bf High energy regime $ |z| \gg 1 $.} Here we let $ H_{\lambda} = |z|^{-2} H_z $ that is
$$ H_{\lambda} = - h^2 \Delta - i h \sigma a , \qquad h = |z|^{-1}, \ \sigma = z / |z| $$
and consider
$$ J = (1/2,3/2), \qquad A_{\lambda} =  A, \qquad \alpha_{\lambda} = 1/2, \qquad \beta_{\lambda} = 0 , $$
with $A$ the generator of dilations (see (\ref{dilatation-notation})) as everywhere below.
Then, $ \mbox{Re} (H_{\lambda}) = -h^2 \Delta + h \mbox{Im}(\sigma) a $ so that
$$ i \big[ \mbox{Re}(H_{\lambda}) ,  A_{\lambda}  \big] = 2 \mbox{Re}(H_{\lambda}) - h \mbox{Im}(\sigma) \big(2 a + x \cdot \nabla a \big) . $$
The Spectral Theorem and the choice of $ J $ imply that
\begin{eqnarray}
 {\mathds 1}_J (\mbox{Re}(H_{\lambda})) \mbox{Re} ( H_{\lambda} ) {\mathds 1}_J (\mbox{Re}(H_{\lambda})) \geq \frac{1}{2}  {\mathds 1}_J (\mbox{Re}(H_{\lambda})) . \label{commutateurgeneralexemple}
\end{eqnarray} 
On the other hand, we have
$$ \big| \big| h \mbox{Im}(\sigma) ( 2 a + x \cdot \nabla a ) \big| \big|_{L^2 \rightarrow L^2} \leq C h , $$
which follows from the boundedness of $a$ and $ x \cdot \nabla a $. Therefore
\begin{eqnarray*}
 {\mathds 1}_J (\mbox{Re}(H_{\lambda}))  i \big[ \mbox{Re}(H_{\lambda}) ,  A_{\lambda}  \big]
  {\mathds 1}_J (\mbox{Re}(H_{\lambda}))  & \geq & ( 1 - C h ) {\mathds 1}_J (\mbox{Re}(H_{\lambda}))  \\
& \geq & \frac{1}{2} {\mathds 1}_J (\mbox{Re}(H_{\lambda})) , \qquad h \ll 1.
\end{eqnarray*} 
Let us comment that $ J = (1/2,3/2) $ and $ \alpha_{\lambda} = 1/2 $ are concrete examples which could be replaced respectively by $ (1- \varepsilon, 1+  \varepsilon ) $ and $ 2 - 3 \varepsilon $ for any fixed $ 0 < \varepsilon < 1 $. We also emphasize that we can take $ \beta = 0 $ and $ A_{\lambda}$ to be the usual generator of dilations since the principal symbol of the operator is $ |\xi|^2 $ which satisfies the non trapping condition (which is stronger than the geometric control condition).}
\item{{\bf Low frequency regime $ |z| \ll 1 $.} To study this regime, we start by substracting a compact part to the dissipation, {\it i.e.} replace $ a $ by $ (1-\chi) a $ with $ \chi \in C_0^{\infty} ({\mathbb R},[0,1]) $ equal to $1$ on a large enough compact set and apply the Mourre theory to
\begin{eqnarray*}
 H_{\lambda} & = & \frac{1}{|z|^2} e^{-iA \ln|z|} \big( - \Delta - i z (1-\chi)a \big) e^{i A \ln |z|} \\
 & = & - \Delta - i \frac{z}{|z|^2} a_{\iota} (x/|z|)
\end{eqnarray*} 
with $ a_{\iota} = (1-\chi)a $. We  consider again
$$ J = (1/2,3/2), \qquad A_{\lambda} =  A, \qquad \alpha_{\lambda} = 1/2, \qquad \beta_{\lambda} = 0 . $$
Then $ \mbox{Re} (H_{\lambda}) = - \Delta + \frac{{\rm Im} (z)}{|z|^2} a_{\iota} (x/z) $ and
$$ i \big[ \mbox{Re}(H_{\lambda}) ,  A_{\lambda}  \big] = 2 \mbox{Re}(H_{\lambda}) - \frac{{\rm Im}(z)}{|z|^2} \mbox{Im}(\sigma) b (x/|z|) , $$
with $ b = 2 a_{\iota} + x \cdot \nabla a_{\iota} $. To get a positive commutator estimate, we want the contribution of $b$ to be small. Using the Hardy inequality, {\it i.e.} the fact that $ |x|^{-1} $ maps $ H^1 $ to $ L^2 $ (in dimension at least 3), we have 
\begin{eqnarray}
 \left| \left| \frac{1}{|z|} b (\cdot / |z|) \right| \right|_{H^{1} \rightarrow L^2}   \leq C    \big| \big| | x|  b \big| \big|_{L^{\infty}} ,  
 \label{bpetit}
\end{eqnarray} 
whose right hand side is small if the support of the cutoff $ \chi $ is large enough since $ a $ and its derivative decay as  $ |x|^{-1-\rho} $ at infinity. 
Using a similar estimate for $ |z|^{-1} a_{\iota}(\cdot/|z|) $ and a routine perturbation argument (viewing $ \mbox{Re}(H_{\lambda}) $ as a perturbation of $ -2 \Delta $), we can show that
$$ || {\mathds 1}_J (\mbox{Re}(H_{\lambda})) ||_{L^2 \rightarrow H^{1} } \leq C , \qquad |z| \ll 1 . $$
Then, choosing $ \chi \equiv 1 $ on a large enough compact set to make the right hand side of (\ref{bpetit}) small, we get
$$ \left| \left| {\mathds 1}_J (\mbox{Re}(H_{\lambda})) \left( \frac{{\rm Im}(z)}{|z|^2}  b (x/|z|) \right)  {\mathds 1}_J (\mbox{Re}(H_{\lambda})) \right| \right|_{L^2 \rightarrow L^2} \leq \frac{1}{2} , $$ 
and thus, using (\ref{commutateurgeneralexemple}), 
$$ {\mathds 1}_J (\mbox{Re}(H_{\lambda})) i \big[ \mbox{Re}(H_{\lambda}) ,  A_{\lambda}  \big] {\mathds 1}_J (\mbox{Re}(H_{\lambda})) 
 \geq  \frac{1}{2} {\mathds 1}_J (\mbox{Re}(H_{\lambda})) , \qquad |z| \ll 1. $$ }
\item{{\bf Intermediate frequency regime $ |z| \in I \Subset ]0,+\infty[ $.} In this part, we consider
$$ H_{\lambda} = - \Delta - i z a ,  $$
and
$$  A_{\lambda} =  A, \qquad \alpha_{\lambda} = 1/2, \qquad \beta_{\lambda} = 0 . $$
Then $ \mbox{Re} (H_{\lambda}) = - \Delta + \mbox{Im}(z) a $ and
\begin{eqnarray}
 i [\mbox{Re}(H_{\lambda}) , A ] = 2 \mbox{Re}(H_{\lambda}) - \mbox{Im}(z) \big(2 a + x \cdot \nabla a \big) .   \label{Hmedium}
\end{eqnarray} 
Since we want to work near a compact subset of $ ]0,+\infty[ $, it suffices to show that for any $ E > 0 $ that there exists $ \epsilon > 0 $ such that if $ | z \mp E^{1/2} | \leq \epsilon $, we have
\begin{eqnarray}
 {\mathds 1}_J (\mbox{Re}(H_{\lambda}))  i [\mbox{Re}(H_{\lambda}) , A ] {\mathds 1}_J (\mbox{Re}(H_{\lambda})) \geq \frac{E}{2} {\mathds 1}_J (\mbox{Re}(H_{\lambda})) , \nonumber
\end{eqnarray} 
with 
$$ J = [E-\epsilon,E+\epsilon] . $$ Indeed, if $ 0 < \mbox{Im}(z) \leq \epsilon $ is small enough, we can ensure that
$$ \left| \left| \mbox{Im}(z) \big(2 a + x \cdot \nabla a \big) \right| \right|_{L^2 \rightarrow L^2} \leq E $$
so that, using (\ref{Hmedium}) and once more (\ref{commutateurgeneralexemple}) (with $ 1/2 $ replaced by $ E - \epsilon $)
\begin{eqnarray*}
 {\mathds 1}_J (\mbox{Re}(H_{\lambda}))  i [\mbox{Re}(H_{\lambda}) , A ] {\mathds 1}_J (\mbox{Re}(H_{\lambda})) & \geq & ( 2E - 2 \epsilon - E ) {\mathds 1}_J (\mbox{Re}(H_{\lambda})) , \\
& \geq & \frac{E}{2} {\mathds 1}_J (\mbox{Re}(H_{\lambda})) ,
\end{eqnarray*} 
which yields the estimate.
}
\end{itemize}

\bigskip 

\bigskip 

\bibliographystyle{alpha}
\bibliography{bibliotex}

\end{document}